\documentclass[11pt]{article}

\usepackage[utf8]{inputenc}

\usepackage[margin=2.5cm]{geometry}

\usepackage{amsmath,amssymb,amsthm}
\usepackage{mathtools}

\usepackage[linesnumbered,ruled,vlined]{algorithm2e}

\usepackage[svgnames,table]{xcolor}
\usepackage{graphicx}
\usepackage{tikz}
\usepackage{tikz-qtree}
\usepackage{forest}
\usepackage{tcolorbox}

\usetikzlibrary{decorations.pathreplacing,calc}

\tcbset{
    boxrule=0.8pt
}

\forestset{
    solid nodes/.style={
        for tree={circle,draw,inner sep=1,fill=green}
    },
    dir/.style={
        for tree={grow=#1}
    },
    leaf/.style={
        label=$#1$
    },
    mytree/.style={
        solid nodes,
        for tree={grow=south,s sep=1cm}
    }
}

\usepackage{array}
\usepackage{booktabs}
\usepackage{float}
\usepackage{hhline}
\usepackage{multirow}

\usepackage[round]{natbib}

\usepackage{nicefrac}
\usepackage{textcomp}
\usepackage{enumerate}
\usepackage{microtype}
\usepackage{tabto}

\usepackage[colorlinks=true,linkcolor=DarkBlue,citecolor=DarkBlue]{hyperref}

\usepackage{appendix}

\newtheorem{theorem}{Theorem}[section]
\newtheorem{claim}[theorem]{Claim}
\newtheorem{corollary}[theorem]{Corollary}
\newtheorem{lemma}[theorem]{Lemma}

\theoremstyle{definition}
\newtheorem{definition}[theorem]{Definition}

\theoremstyle{remark}
\newtheorem{remark}[theorem]{Remark}

\newcommand{\E}{\mathbb{E}}

\newcommand{\one}{\mathbf{1}}
\newcommand{\Prob}[1]{\Pr\big(#1\big)}
\newcommand{\cProb}[2]{\textrm{Pr}\left[\, #1 \, \middle| \, #2 \,\right]}
\newcommand{\ind}[1]{\one{\{#1\}}}

\newcommand{\bbG}{\mathbb{G}}

\newcommand{\calC}{\mathcal{C}}
\newcommand{\calD}{\mathcal{D}}
\newcommand{\calE}{\mathcal{E}}

\newcommand{\calG}{\mathcal{G}}

\newcommand{\calI}{\mathcal{I}}

\definecolor{Blue}{HTML}{2D2F92}

\newcommand{\uprightsf}[1]{\textup{\textsf{#1}}}
\DeclareMathOperator{\val}{val}

\DeclareMathOperator{\opt}{opt}

\DeclareMathOperator{\lca}{LCA}

\DeclareMathOperator{\clr}{color}

\DeclareMathOperator{\nil}{null}

\DeclareMathOperator{\sign}{sign}

\DeclareMathOperator{\tr}{tr}

\newcommand*{\I}{\ensuremath{\mathcal{I}}}
\newcommand*{\C}{\ensuremath{\mathcal{C}}}

\newcommand*{\eps}{\ensuremath{\epsilon}}

\newcommand{\phy}{\uprightsf{phy}}

\newcommand{\vnote}[1]{{\color{red}$\ll$Vaggos: #1$\gg$}}

\newcommand{\notelist}{}

\NewDocumentCommand{\addnote}{mm}{
  \expandafter\gdef\expandafter\notelist\expandafter{%
    \notelist
    \noindent \hyperlink{#2}{#1}\par
  }
  \hypertarget{#2}{#1}%
}

\title{Strong Refutation of Ordering, Phylogenetic, and Ordinary CSPs, and New Satisfiability and Refutation Thresholds for Triplet and Quartet Reconstruction}

\author{Dionysis Arvanitakis\\Northwestern University \and Vaggos Chatziafratis\\UC Santa Cruz \and  Yiyuan Luo\\UC Santa Cruz\and Konstantin Makarychev\\Northwestern University}
\date{}

\begin{document}
\maketitle
\thispagestyle{empty}

\renewenvironment{abstract}
 {\small
  \begin{center}
  \bfseries \abstractname\vspace{-.5em}\vspace{0pt}
  \end{center}
  \list{}{
    \setlength{\leftmargin}{.5cm}%
    \setlength{\rightmargin}{\leftmargin}%
  }%
  \item\relax}
 {\endlist}
 
\begin{abstract}
Understanding the behavior of Constraint Satisfaction Problems (CSPs) on random instances is a fundamental question at the intersection of mathematics, statistical physics, and computer science. In this work, we study phase transitions and efficient algorithms for (strong) refutation of CSPs arising in \textit{hierarchical clustering} (and \textit{ranking}), extensively studied in theoretical computer science and biology as well as ordinary (standard) CSPs without negations.  
In such CSPs, $n$ variables (often representing species, documents or images) are
assigned to leaves of a tree, so as to satisfy $m$ given constraints, specifying evolutionary or proximity relationships. Two canonical $\mathsf{NP}$-hard optimization problems with widespread applications are \textsc{Triplet Reconstruction} and \textsc{Quartet Reconstruction}, where the input consists of triplets ${x,y}\mid {z}$ or quartets ${x,y}\mid {z,w}$ describing relationships on small subsets of $3$ or $4$ variables (respectively), and the goal is to find a tree $T^*$ whose topology maximizes agreement with the constraints. We present three main results as constraint density $\lambda=m/n$ increases:

\begin{enumerate}
    \item We first show the existence and precisely locate the \textit{sharp} threshold $\lambda^*_{\mathrm{triplets}}\approx1.2277$ for \textsc{Triplet Reconstruction} (we give exact closed-form solution). Our proof relies on a variant of \textsc{Build} algorithm~\citep*{aho1981inferring} and the connectivity thresholds for the emergence of giant components in random graphs. To the best of our knowledge, this is the first sharp threshold for the broad family of Phylogenetic CSPs. Moreover, we give a lower and upper bound for \textsc{Quartet Reconstruction}, proving that if $\lambda<0.908$ an instance is satisfiable w.h.p., and if $\lambda>2.88$, it is unsatisfiable w.h.p.
    \item For both problems, we then provide \textit{strong refutation} algorithms that certify that $\val(T^*)\le5/9 + \epsilon$, where $\val(T^*)$ is the fraction of constraints satisfied by the (unknown) optimal tree. For triplets, the refutation algorithm succeeds w.h.p if $m > C'_{\epsilon}n$, and for quartets if  $m > C''_{\epsilon}n^{3/2}$ for some  constants $C'_{\epsilon}$ and $C''_{\epsilon}$ depending only on $\epsilon$.
    \item We also obtain the strongest possible refutations at slightly larger densities (for triplets $m=O(n^{3/2}\log ^3n)$, for quartets $m=O(n^2)$): we certify that $T^*$ is no better than the \textit{trivial random assignment} that ignores the input constraints, i.e., $\val(T^*)\le 1/3+\epsilon$ (notice that for both predicates, a random assignment achieves a $1/3$-approximation). 
    \item  We also obtain strongest possible refutations for ordinary (standard) CSPs without negations over finite domains, extending the work of~\cite{allen2015refute}.
\end{enumerate}
Prior to our work no refutations for triplets and quartets were known. In fact, our refutation results above are instantiations of our general theorem that applies more broadly to Phylogenetic and Ordering CSPs (and other CSPs failing to support $t$-wise independence), and generalizes the current algorithmic frontier on refuting random CSPs over finite alphabets by~\cite*{allen2015refute}. A crucial difference here is that such CSPs, unlike Boolean CSPs, do not have negated variables, so prior works relying on negations---an extra source of randomness---no longer apply.

\end{abstract}

\newpage
\tableofcontents
\thispagestyle{empty}
\newpage
\setcounter{page}{1}
\section{Introduction}

Random Constraint Satisfaction Problems (CSPs) over Boolean or finite alphabets ${0,1,\ldots, q-1}$, and questions concerning their satisfiability, unsatisfiability, and refutation (i.e., the task of efficiently finding certificates of unsatisfiability), have been studied for several decades, yielding deep connections between computer science, mathematics, and statistical physics.
In the random CSP literature, for a given predicate, such as $3$-SAT, a random instance $\mathcal{I}$ is generated by sampling $m$ clauses uniformly at random over a ground set $V$ of $n$ variables. A fundamental question in this area is the following: \textit{As we vary the density $m/n$ of the random instance, is it likely to be satisfiable or far from satisfiable? Can we certify these properties with an efficient algorithm?}

Historically, such average-case analyses of Boolean CSPs not only allowed researchers to bypass pessimistic worst-case hardness results, but also helped drive algorithmic developments for SAT solvers~\citep{davis1960computing,davis1962machine,goldberg1979complexity,franco1983probabilistic}, led to a better understanding of “typical” instances (see the celebrated results of~\cite{chvatal1988many,chvatal1992mick}), and expanded the algorithmic toolbox~\citep{ding2014satisfiability,ding2015proof,raghavendra2017strongly,deshpande2019threshold,mitchell1992hard,kirkpatrick1994critical,selman1996generating}.

In this work, we focus on CSPs studied in the context of hierarchical clustering (and ranking), with two canonical examples being \textsc{Triplet Reconstruction} and \textsc{Quartet Reconstruction}. Since most of the existing literature focuses on random Boolean CSPs or CSPs with fixed alphabet~$\Sigma$, $|\Sigma| = q$ ($q$ is independent of $n$), it does not capture \textsc{Triplet} and \textsc{Quartet Reconstruction}, which are important examples of \textit{infinite-domain} CSPs, as domain size grows with instance size $n$~\citep{bodirsky2010complexity,bodirsky2012complexity}. Even though we focus on triplets and quartets, many of our results extend to \textit{phylogenetic}\footnote{Phylogenetic CSPs arise in a wide range of applications across theoretical computer science, databases, algebra, logic, and computational biology~\citep{aho1981inferring,bandelt1986reconstructing,felsenstein2004inferring,steel1992complexity,jiang1998orchestrating,jiang2001polynomial,henzinger1999constructing,brodal2013efficient,alon2014compatibility,dudek2019computing}.} CSPs~\citep{bodirsky2010complexity,chatziafratis2023triplet}, that is, constraint satisfaction problems on trees --- which include \textit{ordering} CSPs~\citep{GHMRC11} as a special case.



\subsection{Triplet and Quartet Reconstruction}
We start by defining the two most popular CSPs over trees, \textsc{Triplet Reconstruction} (also known as triplet consistency) and \textsc{Quartet Reconstruction} (also known as quartet compatibility):

\begin{definition}[\textsc{Triplet Reconstruction}~\citep{aho1981inferring}]
Given a set $V$ of $n$ items and a collection of $m$ triplets of the form $ab \mid c$, with $a,b,c \in V$, the \textsc{Triplet Reconstruction} problem asks for a rooted binary tree $T$ with $n$ leaves labeled by $V$ such that, for every given triplet $ab \mid c$, the least common ancestor of $a$ and $c$ is equal to the least common ancestor of $b$ and $c$. In this case, we say that the triplet $ab \mid c$ is satisfied.
The corresponding optimization problem, which seeks to maximize the number of satisfied triplets, is denoted \textsc{MaxTripletRec}.
\end{definition}

\begin{definition}[\textsc{Quartet Reconstruction}~\citep{steel1992complexity,jiang1998orchestrating}]\label{def:quartets}
Given a set $V$ of $n$ items and a collection of $m$ quartets of the form $ab \mid cd$, with $a,b,c,d \in V$, the \textsc{Quartet Reconstruction} problem asks for a rooted binary tree $T$ with $n$ leaves labeled by $V$ such that, for every given quartet $ab \mid cd$, the unique path from $a$ to $b$ is vertex-disjoint from the unique path from $c$ to $d$ in $T$ (in which case we say that the quartet $ab \mid cd$ is satisfied). The corresponding optimization problem, which seeks to maximize the number of satisfied quartets, is denoted \textsc{MaxQuartetRec}.
\end{definition}

Both optimization problems are $\mathsf{NP}$-hard and have been extensively studied in the literature (see, e.g.,~\cite{aho1981inferring,jansson2006algorithms,byrka2010new,emamjomeh2018adaptive,chatziafratis2023triplet}) with various applications in computational biology, hierarchical clustering, and databases.  
For example, in computational biology, quartet-based reconstruction is a key component of many tree reconstruction methods, whose aim is to infer an evolutionary tree on a set of species from partial or noisy information about their relationships. They are used both in theory and in practice~\citep{chor1998quartets,semple2003phylogenetics,felsenstein2004inferring} (see also widely used quartet-puzzling methods~\citep{strimmer1996quartet,yang2013quartet,reaz2014accurate}). An important advantage of such triplet and quartet relations compared to distance-based methods is that they do not depend on numerical values that may be highly susceptible to noise. As such, quartets are also widely-used in statistics~\citep{mossel2004phase,daskalakis2006optimal,daskalakis2011evolutionary} where ``quartet tests'' on sampled sequences of characters at the leaves have been used for optimal reconstruction bounds under mutations. 
Quartets are also used to compare trees via the notion of quartet distance~\citep{brodal2013efficient,dudek2019computing}.

In terms of computational complexity, we note that the \textit{decision} version of \textsc{Triplet Reconstruction} 
is solvable in polynomial time by the \textsc{Build} algorithm in~\citep{aho1981inferring}. However, the optimization problem \textsc{MaxTripletRec} remains NP-hard and is approximation resistant under the Unique Games Conjecture; that is, even for almost satisfiable instances of \textsc{MaxTripletRec}, it is NP-hard to find a solution of value greater than the trivial baseline of $1/3+\epsilon$ achieved by a random assignment~\citep{haastad2001some}. \textsc{Quartet Reconstruction} is an NP-complete problem~\citep{steel1992complexity}. In terms of approximation, \cite{jiang1998orchestrating,snir2011linear} gave polynomial-time approximation schemes (PTAS) for dense instances of \textsc{MaxQuartetRec}. 
\cite{snir2011linear} showed how to approximately reconstruct a tree using $m=\widetilde\Theta(n^2)$ sampled quartets. Their result was recently improved by
\cite{arvanitakis2026optimal}, who showed how to PAC-learn a tree using $m=\Theta(n)$ randomly sampled quartets. For general instances, \textsc{MaxQuartetRec} admits a trivial $\nicefrac13$-approximation, and assuming the Unique Games Conjecture~\citep{khot2002power}, this approximation ratio is best possible~\citep{chatziafratis2023triplet} i.e., \textsc{MaxQuartetRec}, and more generally, every phylogenetic CSP, is approximation resistant.




\subsection{Our Results: Satisfiability, Unsatisfiability, and Refutation}

We present several new results on the satisfiability, unsatisfiability, and strong refutation of random instances of triplets, quartets (and other phylogenetic and ordering CSPs) as we vary their density $\lambda=\nicefrac{m}{n}$. Importantly, we identify the exact density where triplets have a \textit{sharp} threshold, we give upper and lower bounds for quartets transitioning from satisfiability to unsatisfiability, and provide efficient strong refutation algorithms. Our main result in refutation generalizes the current algorithmic frontier for refutation of CSPs over finite alphabets~\citep*{allen2015refute}. Prior works on random instances, studied local-vs-global compatibility for \textsc{Quartet Reconstruction}~\citep*{alon2014compatibility}.

As mentioned earlier, for ease of presentation we focus on the two most important phylogenetic CSPs, \textsc{Triplet Reconstruction} and  \textsc{Quartet Reconstruction} to instantiate our results.
Nevertheless, most of our results apply to general phylogenetic CSPs. Before presenting our results, we introduce the standard random model for the instances we study (see also Definition~\ref{def:random-instance}). 

\paragraph{Random Instances.} To generate a random instance $\mathcal{I}(n,m)$ of triplets (or quartets), we independently sample three (four) distinct variables, say $a,b,c$ (or $a,b,c,d$), and add the triplet  ${a,b}\mid{c}$ (quartet ${a,b}\mid{c,d}$). We repeat this procedure $m$ times. This generates a uniformly random instance, analogous to standard approaches for generating uniformly random Boolean formulas (and was also used in the case of quartets by~\cite{alon2014compatibility}). 

\paragraph{Satisfiability Phase Transition for Triplets and Quartets.}

Our first main result is a \textit{sharp threshold} for \textsc{Triplet Reconstruction} at $\lambda_{crit}\approx1.2277$. We give a closed form solution for the critical threshold at density value $\lambda_{crit}= -\frac{w^*}{(2 + \nicefrac{1}{w^*})}\approx 1.2277$, where 
$w^* = W_{-1}(- \nicefrac{1}{2\sqrt{e}})\approx -1.7564$ and $W_{-1}$ is the $-1$ branch of the Lambert $W$ function:\footnote{For $x\in [-1/e,0)$, $W_{-1}(x)$ is the solution to the equation $y e^y = x$ with $y \leq -1$, i.e., $W_{-1}(x)e^{W_{-1}(x)}=x$ (see~\cite*{corless1996lambert}).} 
\begin{tcolorbox}
\begin{theorem}\label{th:sharp}
A random instance of \textsc{Triplet Reconstruction} (also known as Rooted Triplets Consistency) with density $\lambda<\lambda  _{crit}$ is satisfiable w.h.p., i.e., there exists a tree satisfying all $m$ triplets, whereas a random instance with $\lambda>\lambda  _{crit}$ is not satisfiable w.h.p., i.e., no tree satisfies all triplets.
\end{theorem}
\end{tcolorbox}
This sharp threshold adds to the short list of exact thresholds known for random CSPs~\citep{achlioptas2009random,ding2014satisfiability}, and to the best of our knowledge, is the first sharp phase transition established for a natural CSP over tree structures. We note here that even though \textsc{Triplet Reconstruction} is NP-hard to optimize, the decision version of the problem for checking whether there exists a tree satisfying all constraints is actually in $\mathsf{P}$ (for quartets and other Phylogenetic CSPs this is not true as they are $\mathsf{NP}$-complete ~\citep{steel1992complexity,bodirsky2017complexity}).~\citet*{aho1981inferring}, motivated by an application in databases, gave a polytime algorithm called \textsc{Build} that returns a tree compatible with all $m$ triplets (if such a tree exists, otherwise it halts).  Our proof relies on a lazy variant of \textsc{Build} ---we call it \textsc{Gradual-Build}--- that uses connections to connectivity thresholds and the emergence of the giant component in the literature on random graphs; however, in our case, constraints are on triples $a,b,c$ of vertices, so the presence or absence of an edge in the associated constraint graph depends on whether or not \textit{all} vertices $a,b,c$ are part of the same induced connected component (as we progressively construct smaller and smaller clusters by separating some vertices from others so to construct the final hierarchy), and this makes the arguments more intricate. For details, see Section~\ref{sec:sharp-threshold}.

Next, we present upper and lower bounds for satisfiability of quartets:

\begin{tcolorbox}
\begin{theorem}\label{th:quartets_thresholds}
A random instance of \textsc{Quartet Reconstruction} with density $\lambda<0.908$ is satisfiable w.h.p., i.e., there exists a tree satisfying all $m$ quartets, whereas a random instance with $\lambda>2.89$ is not satisfiable w.h.p., i.e., no tree satisfies all quartets.
\end{theorem}
\end{tcolorbox}

For the precise formulas for the values of the quartet thresholds (as well as efficient satisfiability algorithms in the low-density regime), please refer to the Theorem~\ref{th:quartet_upper_lower} (and Theorem~\ref{thm:threshold-main}). Briefly, to obtain the lower bound we still run \textsc{Build}, where we need to appropriately modify it for quartets (adding two random edges per $ab|cd$ quartet) and analyze the probability of it successfully outputting a hierarchy, whereas for the upper bound, we analyze the probability that the top-split in a given tree does no violation on the quartets, then taking a union bound over possible top-splits/tree topologies.

\vspace{-12pt}

\paragraph{Strong Refutation.} 
We first present our strong refutation results for the specific Triplet and Quartet predicates, two important special cases of Phylogenetic CSPs. As we will see, our strong refutation algorithms apply to Ordering and Phylogenetic CSPs, and, in fact, more broadly to CSPs without negation, thereby extending the influential work of~\cite{allen2015refute}. Concurrently and independently of our work,~\cite*{chan2026strongly} also presented a strong refutation algorithm for CSPs without negation using different techniques.

For \textsc{Triplet Reconstruction} and \textsc{Quartet Reconstruction}, the random instance is satisfiable with high probability when the number of constraints is at most $cn$, and is at most $1/3+\epsilon$ satisfiable with high probability when the number of constraints is at least $C_{\epsilon}n$, for suitable constants $c$ and $C_{\epsilon}$, the upper bound here follows by a simple argument from the recent results on the Natarajan dimension of quartets (\cite{arvanitakis2026optimal}) and triplets (\cite{avdiukhin2023tree}), respectively,  that show that for both problems the dimension is $O(n)$.
Note that $1/3$ is the approximation factor achieved by the random assignment algorithm for both triplets and quartets.
We present general refutation results for Phylogenetic and Ordering CSP predicates that fail to support a $t$-wise $\mu$-independent distribution, at density $m=\tilde\Theta(n^{\nicefrac{t}{2}})$. 
Our results generalize and extend the work of~\cite*{allen2015refute}, which studied strong refutation for Boolean and finite-domain CSPs (see Section~\ref{sec:refutation} for definitions). Recall that an efficient algorithm which, given a random instance, certifies that $\opt(\mathcal{I}) < 1$ is called a (weak) \emph{refutation} algorithm. The refutation task becomes more challenging if the algorithm is required to certify that $\opt(\mathcal{I}) < 1-\delta$ for a constant $\delta$. The strongest possible form of refutation certifies that $\opt(\mathcal{I}) \le \alpha^*+\epsilon$, where $\alpha^*$ is a trivial baseline, specifically, the expected value achieved by the best (possibly biased) random assignment algorithm. For Triplet and Quartet predicates, $\alpha^* =1/3$.
We begin with the special case of quartets and triplets. 
\begin{tcolorbox}
\begin{theorem}[Strong Refutation for Quartets]\label{th:refutation-quartets}
For every $\epsilon>0$, there exists a polynomial-time refutation algorithm (see Definition~\ref{def:certify}) that certifies $\opt(\calI) \le \tfrac13+\epsilon$ with high probability for a random instance $\mathcal{I}(n,m)$ of \textsc{MaxQuartetRec} with
 $m\ge  \Omega_{\epsilon} (n^{2})$.

\medskip

There exists a polynomial-time refutation algorithm that certifies $\opt(\calI) \le \tfrac59+\epsilon$ with high probability for a random instance $\mathcal{I}(n,m)$ of \textsc{MaxQuartetRec} with
 $m\ge  \widetilde\Omega_{\epsilon} (n^{1.5})$.
\end{theorem}
\end{tcolorbox}

For triplets, we obtain a similar result.

\begin{tcolorbox}
\begin{theorem}[Strong Refutation for Triplets]\label{th:refutation-triplets}
For every $\epsilon>0$, there exists a polynomial-time refutation algorithm (see Definition~\ref{def:certify}) that certifies $\opt(\calI) \le \tfrac13+\epsilon$ with high probability for a random instance $\mathcal{I}(n,m)$ of \textsc{MaxTripletRec} with
$
m \ge \widetilde{\Omega}_{\epsilon}(n^{3/2}) .
$

\medskip

Moreover, there exists a polynomial-time refutation algorithm that certifies $\opt(\mathcal{I}) \le \tfrac59+\epsilon$ with high probability for a random instance $\mathcal{I}(n,m)$ of \textsc{MaxTripletRec} with
$
m \ge {\Omega}_{\epsilon}(n) .
$
\end{theorem}
\end{tcolorbox}

\medskip

Theorems~\ref{th:refutation-quartets} and~\ref{th:refutation-triplets} follow from Theorems~\ref{thm:main-refute} and~\ref{thm:t-wise-phy-CSP-refute}, which apply to arbitrary phylogenetic and ordering CSPs. 
To obtain the specific bound of $5/9$ in the second parts of the above theorems, we use the results from Section~\ref{sec:beta-triplets} and~\ref{sec:beta-quartets}, showing how the triplet and quartet predicates fail to support certain $t$-wise independent distributions for different values of the parameter $t$.
In Section ~\ref{sec:bounds-uniform-distr}, we also discuss applications of our general results to several popular ordering CSPs: \textsc{Betweenness} and \textsc{Cyclic Ordering}. 
\medskip

\paragraph{Organization.} The next section proves the sharp satisfiability threshold for triplet reconstruction. In Section \ref{sec:ub-quartets} we provide upper and lower bounds on the satisfiability threshold for quartets.  In Section \ref{sec:bounds-uniform-distr} we bound the probability of pairwise uniform and three-wise uniform distributions satisfying the triplet and quartet predicates, this coupled together with the theory developed in Appendix ~\ref{sec:refutation} for general ordering and phylogenetic CSPs gives us strong refutation in asymptotically lower density than $m\approx n^{k/2}$ (arity $k=3$ for triplets, $k=4$ for quartets). In Appendix \ref{sec:prelim}, we provide extended preliminaries on ordering and phylogenetic CSPs (\cite{chatziafratis2018hierarchical}, \cite{makarychev2026approximation}, \cite{guruswami2008beating}). We then develop the general refutation theory for all ordering and phylogenetic CSPs. Omitted details and several examples are also provided in the Appendix.

\section{Satisfiability Threshold for Triplet Reconstruction}
\label{sec:sharp-threshold}

We prove that a random Triplet Reconstruction instance with $m=\lambda n/2$ constraints is satisfiable w.h.p. if $\lambda<\lambda_1^*$ and not satisfiable w.h.p. if $\lambda>\lambda_1^*$; equivalently, the threshold in the fixed-size model is $m=\lambda_1^* n/2\approx 1.2277n$.

In fact, we consider a slightly more general setting in which each constraint has $r$ negative examples
$\mathbf{c}^{(1)},\dots,\mathbf{c}^{(r)}$.
Such a constraint is satisfied if each triplet $ab|\mathbf{c}^{(j)}$ is satisfied or, equivalently, if each $\mathbf{c}^{(j)}$ is separated from the pair $a,b$ before $a$ and $b$ are separated.
In this more general case, the satisfiability threshold is
\begin{equation}\label{eq:def:lambda-r}
\lambda_r^*=\min_{s\in(0,1)}\frac{-\ln(1-s)}{s\bigl(1-(1-s)^r\bigr)}.
\end{equation}
For $r=2$, $\lambda_2^*\approx 1.81604$ and $m\approx 0.908n$.
Although our proof works for every fixed $r\geq 1$, we recommend that the reader first focus on the most important case of  $r=1$.

In this section, instead of fixing the number of constraints $m$ in a random instance of Triplet Reconstruction, we fix the probability $\approx\lambda/n^{r+1}$ that a constraint $ab|\mathbf{c}$ is included in the set of constraints of a random instance. We prove the main result for this model of random instances and then provide an easy reduction to the model with a fixed number of constraints. We identify constraints $ab|\mathbf{c}$ and $ba|\mathbf{c}$ since they are equivalent.
Also, for technical reasons, we allow a random instance to contain multiple copies of the same constraint $ab|\mathbf{c}$. Specifically, we assume that a random instance $\calI_r(n,\lambda)$ has $X_{ab|\mathbf{c}}$ copies of the constraint $ab|\mathbf{c}$, where $X_{ab|\mathbf{c}}$ is a Poisson random variable with parameter
$$
p=\frac{\lambda}{(n-1)(n-2)^r}.
$$
That is,
$$
\Pr[X_{ab|\mathbf{c}}=k]=e^{-p}p^k/k!.
$$
All random variables $X_{ab|\mathbf{c}}$ are mutually independent. Note that the probability that at least one copy of the constraint $ab|\mathbf{c}$ is present in the instance equals
$$
\Pr[X_{ab|\mathbf{c}}\ge1]
=
1-\Pr[X_{ab|\mathbf{c}}=0]
=
1-e^{-p}
=
p+O(p^2)
=
\frac{\lambda}{(n-1)(n-2)^r}
\Bigl(1+O_\lambda\!\left(\frac1{n^r}\right)\Bigr).
$$
The number of distinct constraints $ab|\mathbf{c}$ is $n(n-1)(n-2)^r/2$ (since we identify constraints $ab|\mathbf{c}$ and $ba|\mathbf{c}$ and do not allow $\mathbf{c}^{(j)}$ to be equal to $a$ or $b$). Thus, the expected number of distinct constraints in the random instance $\calI_r(n,\lambda)$ is $\lambda n/2 + O(1/n)$.

\begin{theorem}\label{thm:threshold-main}
Let $\lambda^*_r$ be as above. For every positive $\lambda<\lambda^*_r$, a random instance $\calI_r(n,\lambda)$ is satisfiable w.h.p. For every positive $\lambda>\lambda^*_r$, a random instance $\calI_r(n,\lambda)$ is not satisfiable w.h.p.
\end{theorem}

\begin{figure}[ht]
    \centering   \includegraphics[width=6in]{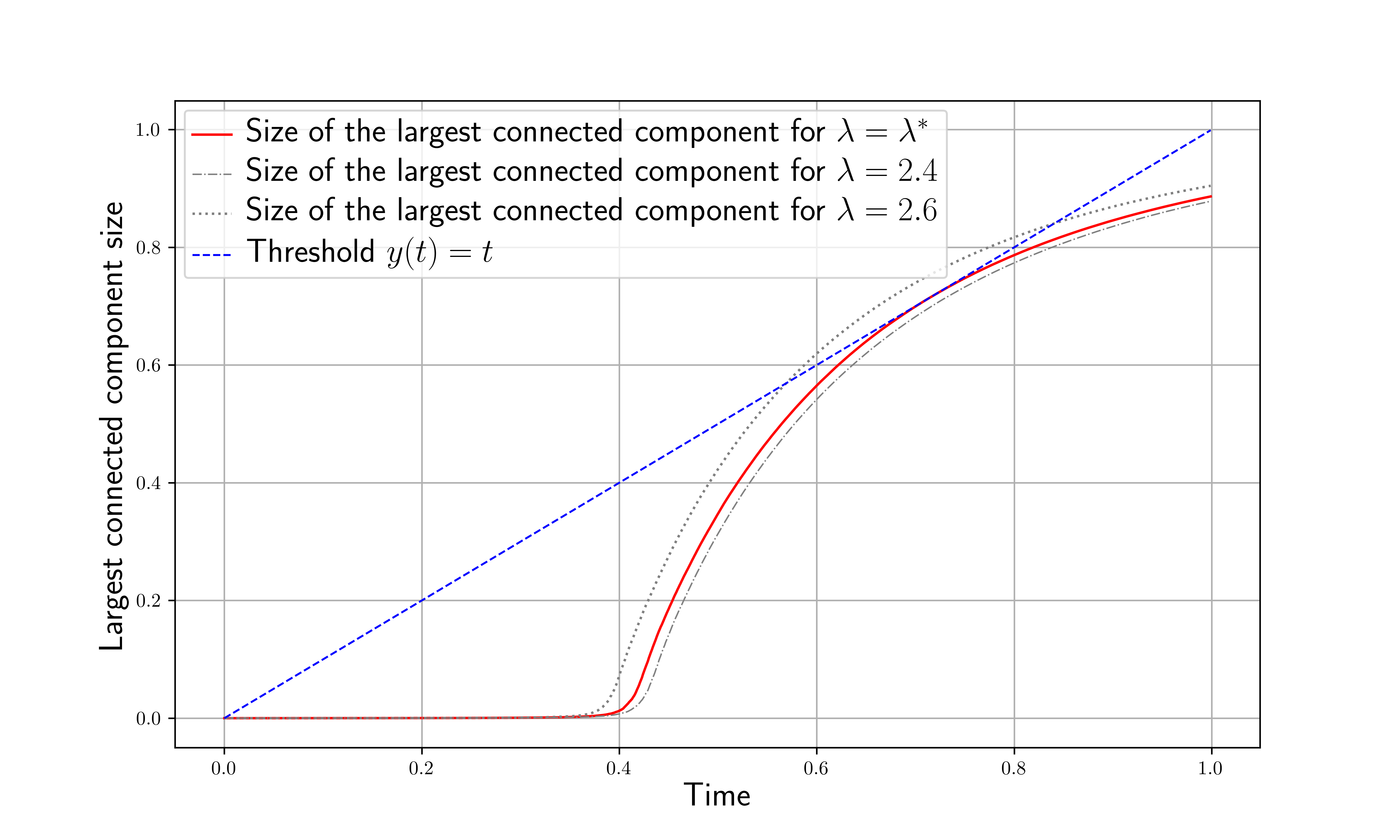}
    \caption{The size of the largest connected component in $\calG_{n,\Lambda_r}(t)$ for $\lambda=\lambda^*_1\approx 2.4554$, $p=\lambda/n^2$. The time and size of the largest connected component are scaled by $1/n$.}
    \label{fig:enter-label}
\end{figure}
Consider first the case $r=1$. Define the continuous random multigraph process as follows: $\calG_{n,p}(t)$ is a graph on $n$ vertices; the multiplicity of each edge $(u,v)$ is a Poisson process with parameter $pt$ (see~\cite{stepanov1970probability} and also~\cite{erdos1960evolution}, \cite*{janson2011random}). Note that $\calG_{n,p}(t)$ contains very few parallel edges for $t\le n$. If we replace parallel edges in $\calG_{n,p}(t)$ with a single edge, then, for a fixed $t$, we obtain an Erd\H{o}s--R\'enyi graph $\bbG(n,1-e^{-pt})$.

In the general case, for a fixed $r\geq 1$, we use a  non-homogeneous random multigraph process, which we denote by the same symbol $\calG_{n,\Lambda_r}(t)$. Each edge $(u,v)$ is generated by an independent non-homogeneous Poisson process with cumulative intensity function
\[
\Lambda_r(t)
=
p\Bigl((n-2)^r-\bigl((n-2)-t\bigr)^r\Bigr),
\]
where $t\in [0,n-2]$.
Equivalently, at time $t$ the multiplicity of $(u,v)$ has distribution $\operatorname{Pois}(\Lambda_r(t))$, increments over disjoint intervals are independent, and the instantaneous arrival rate is
\[
\Lambda_r'(t)
=
pr\bigl((n-2)-t\bigr)^{r-1}.
\]
Note that for $r=1$, $\Lambda'_r(t) = p$. At the final time each edge has multiplicity distributed as 
$
\mathrm{Pois}(\Lambda_r(n-2))=\mathrm{Pois}(\lambda/(n-1))$, exactly as in the construction of the random instance. 

We prove the following lemmas, which imply Theorem~\ref{thm:threshold-main}.

\begin{lemma}\label{lem:prob-equal}
The probability that a random instance $\calI_r(n,\lambda)$ is not satisfiable equals the probability that the continuous  random multigraph process $\calG_{n,\Lambda_r}(t)$ has a connected component of size at least $t+2$ at some time $t\in\{1,\dots, n-2\}$, where $p=\tfrac{\lambda}{(n-1)(n-2)^r}$.
\end{lemma}

\begin{lemma}\label{lem:max-comp-pois-process}
Consider the random graph process $\calG_{n,\Lambda_r}(t)$ with $p=\tfrac{\lambda}{(n-1)(n-2)^r}$. 
\begin{enumerate}
\item If $\lambda < \lambda^*_r$, then the size of the largest connected component in $\calG_{n,\Lambda_r}(t)$ is less than $t+2$ for all $t\in\{1,\dots, n-2\}$ w.h.p.  
\item If $\lambda > \lambda^*_r$, then $\calG_{n,\Lambda_r}(t)$ has a connected component of size at least $t+2$ for some $t\in\{1,\dots, n-2\}$ w.h.p.
\end{enumerate}
\end{lemma}

We first prove Lemma~\ref{lem:prob-equal}.

\begin{figure}
    \centering
    \includegraphics[width=0.7\linewidth]{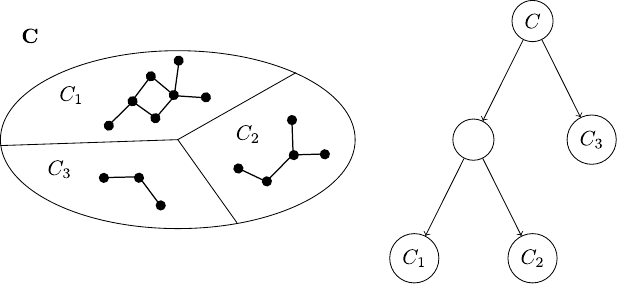}
    \caption{Cluster $C$ split into $C_1$, $C_2$, and $C_3$ and the corresponding Phylogenetic tree structure.}
    \label{fig:splits}
\end{figure}

\begin{proof}[Proof of Lemma~\ref{lem:prob-equal}] 
We already know that a Triplet Reconstruction instance $\calI$ is satisfiable if and only if the \textsc{Build} algorithm by~\cite*{aho1981inferring} succeeds in finding a solution that satisfies all constraints in $\calI$. We now define a \emph{gradual} variant of \textsc{Build}. This is an iterative algorithm. As in non-gradual \textsc{Build}, the \textsc{Gradual Build} algorithm creates a multigraph representing not yet satisfied constraints $ab|c$. We denote the state of this graph after iteration $i$ by $G(i)$. Initially, $G(0)$ contains an edge $(a,b)$ for every constraint $ab|\mathbf{c}$ in $\calI$. As the algorithm builds the Phylogenetic Tree, it satisfies some of the constraints in $\calI$ and removes the corresponding edges from $G$. We now discuss the details.
For every vertex $u\in V$, the algorithm maintains its own list of explored vertices. At iteration $i$, this list is $F_u(i)$. Initially, $F_u(0)=\varnothing$ for all $u$. Then, at every step the algorithm adds exactly one vertex to each list $F_u(i)$ ensuring that $|F_u(i)|=i$. The algorithm also guarantees that $F_u(i) = F_v(i)$ if vertices $u$ and $v$ are in the same connected component of $G(i-1)$, i.e., the list $F_C(i)$ is the same for all vertices in one connected component of $G(i-1)$.

Nodes of the Phylogenetic Tree constructed by the algorithm correspond to connected components of graphs $G(i)$. The root of the tree corresponds to the entire vertex set $V$. At every iteration $i\in\{1,\dots,n\}$, the algorithm partitions $G(i-1)$ into connected components by splitting each connected component in $G(i-2)$ and updates the Phylogenetic tree  accordingly. For $i=1$, the algorithm splits the set $V$ into connected components. See Figure~\ref{fig:splits}.  Then, for every connected component $C$ of $G(i-1)$ of size at least $2$, it picks an arbitrary vertex $f_C(i)$ outside of $C \cup F_u(i-1)$ (if such vertex exists), where $u$ is an arbitrary vertex in $C$. Note that $F_u(i-1) = F_v(i-1)$ for all $u,v\in C$ ($C$ is a connected component in $G(i-1)$ and in $G(i-2)$ if $i\geq 2$), so the choice of $u$ does not matter. If such vertex does not exist because $F_u(i-1) = V\setminus C$, then we say that the algorithm has failed and let $f_C(i)$ be an arbitrary vertex not yet explored by vertices in $C$. We include vertex $f_C(i)$ in every list $F_u(i)$, where $u\in C$, i.e., 
$$F_u(i) = F_u(i-1) \cup \{f_C(i)\}.
$$
Note that even if the algorithm fails at step $i$, we continue running it until iteration $n$. We do so for the sake of analysis. After we update all lists $F_u(i)$, we remove  edges $(a,b)$ from $G(i-1)$ that correspond to constraints $ab|\mathbf{c}$ with all  coordinates $\mathbf{c}^{(j)}$ in $F_a(i)$. Note that $F_a(i) = F_b(i)$ because $a$ and $b$ are connected with an edge in $G(i-1)$ and, consequently, are in the same connected component. Hence, the condition  $\mathbf{c}^{(j)}$ in $F_a(i)$ is equivalent to $\mathbf{c}^{(j)}$ in $F_b(i)$. We denote the new graph $G(i)$.
If the algorithm has not failed, it returns the Phylogenetic Tree after iteration $n$.

\begin{claim}
\textsc{Gradual  Build} fails on $\calI$ if and only if $\calI$ is not satisfiable.
\end{claim}
\begin{proof}
Observe that if \textsc{Gradual  Build} succeeds then it finds a solution satisfying all the constraints (the proof is the same as for non-gradual  \textsc{Build}). If \textsc{Gradual  Build} fails, then at some step it could not find a vertex 
$f_C(i)$ outside of some connected component $C$. This means that $F_u(i-1) = V\setminus C$ for all $u\in C$. Therefore, all edges $(a,b)$ in the connected component $C$ correspond to constraints  labeled with $\mathbf{c}$ for which at least one coordinate $\mathbf{c}^{(j)}$ is in $C$. This means that the induced instance $\calI[C]$ is connected and, hence, $\calI[C]$ and $\calI$ are not satisfiable (see Theorem~\ref{thm:Build-Correct}).
\end{proof}

We now find the probability that \textsc{Gradual Build} fails on the random instance $\calI(n,\lambda)$.
Consider random graph process 
$\calG_{n,\Lambda_r}(t)$ at time points $\{n-2,n-3,\dots, 0\}$. Specifically, let $H(i) = \calG_{n,\Lambda_r}(n-2-i)$. We show that processes $G(i)$ and $H(i)$ have the same distribution. We are going to prove that $G(i)$ and $H(i)$ are Markov chains with the same initial distribution and same transition probabilities.

We first examine process $H(i)$. Recall that
$H(i)=\calG_{n,\Lambda_r}(n-2-i)$.
In multigraph
$H(0)=\calG_{n,\Lambda_r}(n-2)$,
each edge is present with multiplicity distributed according to the Poisson distribution with parameter
$p(n-2)^r=\lambda/(n-1)$.
Suppose that edge $(a,b)$ is present in
$H(i)=\calG_{n,\Lambda_r}(n-2-i)$
with multiplicity $X_{ab}$.
If we condition on the value of $X_{ab}$, then the transformed arrival times
$\Lambda(T_1),\dots,\Lambda(T_{X_{ab}})$
are independently and uniformly distributed in
$[0,\Lambda(n-2-i)]$
(see, e.g., Theorem~8.15 in~\cite*{mitzenmacher2017probability}).
When we transition from step $i$ to step $i+1$, edge copies with transformed arrival times in the interval
$[\Lambda(n-3-i),\Lambda(n-2-i)]$
disappear.
Thus, every edge copy disappears with probability
\[
\frac{\Lambda(n-2-i)-\Lambda(n-3-i)}
{\Lambda(n-2-i)}
=
\frac{(i+1)^r-i^r}
{(n-2)^r-i^r},
\]
independently of all other edge copies.

We now analyze process $G(i)$. Recall that the random instance $\calI_r(n,\lambda)$ is constructed as follows. For every constraint $ab|\mathbf{c}$ (where we identify $ab|\mathbf{c}$ and $ba|\mathbf{c}$ and do not allow $\mathbf{c}^{(j)}$ to be equal to $a$ or $b$), we pick the number of copies of $ab|\mathbf{c}$ according to the Poisson distribution with parameter
$p=\tfrac{\lambda}{(n-1)(n-2)^r}$. It will be convenient for us to construct $\calI_r(n,\lambda)$ in two steps. First, for every unordered pair $(a,b)$, we draw a random variable $X_{ab}$ with the Poisson distribution with parameter $\lambda/(n-1)$. Then, we create $X_{ab}$ constraints of the form $ab|\mathbf{c}$. The value of $\mathbf{c}$ is chosen uniformly from $(V\setminus\{a,b\})^r$, independently for each constraint. According to the Poisson process splitting theorem (see, e.g., Theorem~8.13 in~\cite{mitzenmacher2017probability}), the multiplicity of each constraint $ab|\mathbf{c}$, denoted by $X_{ab|\mathbf{c}}$, has the Poisson distribution with parameter $\tfrac{\lambda}{(n-1)(n-2)^r}$, and all random variables $X_{ab|\mathbf{c}}$ are independent.

For the sake of analysis, let us assume that the identity of $\mathbf{c}$ in every constraint is not revealed to us before the \textsc{Gradual Build} algorithm discovers its coordinates. Then, at every iteration $i$ of the algorithm and for every not yet satisfied constraint $ab|\mathbf{c}$, the vector $\mathbf{c}$ is distributed uniformly over $(V\setminus\{a,b\})^r\setminus F_a(i)^r$. Such a constraint is removed at iteration $i+1$ if every coordinate of $\mathbf{c}$ belongs to $F_a(i+1)$, where $C$ is the connected component containing $a$ and $b$. The probability of this event is
$
\tfrac{(i+1)^r-i^r}{(n-2)^r-i^r}$.

We showed that processes $G(i)$ and $H(i)$ have the same distribution. To finish the proof, observe that \textsc{Gradual Build} fails if and only if at some step $i$, we have $F_a(i-1) = V\setminus C$ for some connected component $C$ in $G(i-1)$ and $a\in C$. If \textsc{Gradual Build} has not yet failed, then $F_a(i)$ and $C$ are disjoint. Also, $|F_a(i)| = i$. Hence, \textsc{Gradual Build} fails if and only if at some step $i$, the size of the largest connected component in $G(i-1)$ is $n-i +1$. The probability of this event is the same as the probability that $H(i-1)$ has a connected component of size $n-i + 1$, which in turn equals the probability that $\calG_{n,\Lambda_r}(t)$ has a connected component of size at least $t+2$.
This concludes the proof.
\end{proof}

\begin{proof}[Proof of Lemma~\ref{lem:max-comp-pois-process} for $r=1$]
We now prove Lemma~\ref{lem:max-comp-pois-process} in the case $r=1$. The proof for every fixed $r>1$ is given in Appendix~\ref{sec:arbitrary-r}.\footnote{At this point we split the final random-graph estimate into two parts: the main text proves the simpler $r=1$ case, while Appendix~\ref{sec:arbitrary-r} contains the analogous argument for fixed $r>1$.}
Consider a fixed time $t$. When does a random graph $\calG_{n,p}(t)$ have a connected component of size at least $t+2$?
The answer to this question is known -- it follows from the works of \citet*{erdos1960evolution}, \cite{bollobas1984evolution}, \cite{luczak1990component} on the evolution of random graphs. We will use the following theorem from the book by~\cite{frieze2016introduction}.

\begin{theorem}[Theorem 2.14 in \cite{frieze2016introduction}; see also Theorem 5.12 in~\cite*{janson2011random}]\label{thm:giant-connected-comp}
    If $p = \frac{c}{n}$, $c > 1$, then w.h.p. $\bbG(n,p)$ consists of a unique giant component, with $\left( 1 - \frac{x}{c} + o(1) \right) n$ vertices and $\left( 1 - \frac{x^2}{c^2} + o(1) \right) \frac{cn}{2}$ edges. Here $0 < x < 1$ is the solution of the equation $xe^{-x} = ce^{-c}$. The remaining components are at most $O(\log n)$. The theorem also holds for $\bbG(n,m)$ with $m = \frac{cn}{2}$.
\end{theorem}

After solving the equation $xe^{-x}=ce^{-c}$ for $c\geq 1$, we get the following corollary, which gives us a tight linear upper bound on the size of the largest connected component in $\bbG(n,c/n)$.

\begin{corollary}\label{cor:thm:giant-connected-comp}
For $\lambda_1^*$ defined as above (see~\eqref{eq:def:lambda-r}), the following two claims hold.
\begin{enumerate}
\item 
For every positive $\delta$ and $c$, the size of the largest connected component in the random $\bbG(n,c/n)$ graph is at most $(1+\delta)(c/\lambda_1^*)n$ with high probability.
\item 
There exists a positive $c^*$ such that for every $\delta\in(0,1)$, the largest connected component in the random $\bbG(n,c^*/n)$ graph is at least $(1-\delta)(c^*/\lambda_1^*)n$ with high probability.
\end{enumerate}
\end{corollary}
We prove an analog of this corollary for $r>1$ in Appendix~\ref{sec:arbitrary-r}.
\begin{proof}
Fix a positive $c$. If $c\leq 1$, then all components in $\bbG(n,c/n)$ have size $o(n)$ w.h.p., and Part I follows. If $c\geq \lambda_1^*$, then Part I is trivial because $(1+\delta)(c/\lambda_1^*)n\geq n$. Thus, we assume that $c\in(1,\lambda_1^*)$.

Let $\rho=\rho(c)$ be the asymptotic fraction of vertices in the giant component of $\bbG(n,c/n)$. By Theorem~\ref{thm:giant-connected-comp}, the largest connected component has size $(\rho+o(1))n$ w.h.p. Moreover, $\rho$ satisfies
$1-\rho=e^{-c\rho}$, or equivalently,
\[
c=\frac{-\ln(1-\rho)}{\rho}.
\]
Therefore, by the claim proved below,
\[
\frac{c}{\rho}
=
\frac{-\ln(1-\rho)}{\rho^2}
\geq \lambda_1^*.
\]
Hence $\rho\leq c/\lambda_1^*$, and Part I follows.

We now prove Part II. Let $s^*\in(0,1)$ be a point where the function
\[
h(s)=-\frac{\ln(1-s)}{s^2}
\]
attains its minimum, and let
\[
c^*=\frac{-\ln(1-s^*)}{s^*}.
\]
Since $h(s^*)=\lambda_1^*$, we have
$c^*=\lambda_1^*s^*$.
Also, the equation defining $c^*$ is equivalent to
$1-s^*=e^{-c^*s^*}$. Thus, $s^*$ is the giant-component fraction in $\bbG(n,c^*/n)$. By Theorem~\ref{thm:giant-connected-comp}, the largest connected component in $\bbG(n,c^*/n)$ has size $(s^*+o(1))n$ w.h.p. Since $s^*=c^*/\lambda_1^*$, Part II follows.
\end{proof}

We now use Corollary~\ref{cor:thm:giant-connected-comp} to finish the proof of Lemma~\ref{lem:max-comp-pois-process} in the case $r=1$. We give the proof for fixed $r>1$ in Appendix~\ref{sec:arbitrary-r}.

First suppose that $\lambda>\lambda_1^*$. Let $c^*$ be the constant from Part II of Corollary~\ref{cor:thm:giant-connected-comp}. Choose $\delta\in(0,1)$ such that
$
(1-\delta)\frac{c^*}{\lambda_1^*}
>
\frac{c^*}{\lambda}.
$
Let
$
t=\frac{c^*}{\lambda}(n-1).
$
Since $c^*<\lambda_1^*<\lambda$, we have $t\in[1,n-2]$ for all sufficiently large $n$. Also,
$
1-e^{-pt}
=
\frac{c^*+o(1)}{n}.
$
Thus, after replacing parallel edges by single edges, $\calG_{n,p}(t)$ has the same distribution as
$\bbG\left(n,\frac{c^*+o(1)}{n}\right)$. By Part II of Corollary~\ref{cor:thm:giant-connected-comp}, the graph $\calG_{n,p}(t)$ contains a connected component of size at least
$
(1-\delta)\frac{c^*}{\lambda_1^*}n
$
w.h.p. By our choice of $\delta$, this size is at least $t+2$ for all sufficiently large $n$. This proves Part II of Lemma~\ref{lem:max-comp-pois-process}.

We now turn our attention to Part 1 of the lemma. We are going to do the following. For every $\lambda<\lambda^*$, we pick an integer $s$ such that
$(1+\frac{1}{s})^2<\lambda^*/\lambda$. Note that $s$ does not depend on $n$. Then, for every $n$, we define time points $t_0<\dots<t_{2s+2}$, where $t_0=1$ and $t_{2s+2}=n-2$, and show that for every $i\geq1$, with high probability the largest connected component in $\calG_{n,p}(t_i)$ has size less than $t_{i-1}+2$.
We let $t_0=1$, $t_1=\sqrt[3]{n}$, and
$t_i=\frac{1}{3}(1+\frac{i-2}{s})(n-2)$ for $i\geq2$. The time points $t_i$ need not be integers, since the process $\calG_{n,p}(t)$ is defined for all real $t\in[0,n-2]$.

We first consider $t_1=n^{1/3}$. Since $1-e^{-pt_1}=O(n^{-5/3})$, the graph $\calG_{n,p}(t_1)$ has maximum degree at most $1$ w.h.p. Hence, its largest connected component has size at most $2$, which is less than $t_0+2=3$.
Next consider $t_2=(n-2)/3$. Since $1-e^{-pt_2}=(\lambda/3+o(1))n^{-1}$ and $\lambda/3<1$, the graph $\calG_{n,p}(t_2)$ is subcritical. Therefore, its largest connected component has size $O(\log n)$ w.h.p., which is less than $t_1+2$.

Now consider $i\geq3$ and let $j=i-2$. Then $t_i=\frac{1}{3}(1+\frac{j}{s})(n-2)$, and graph $\calG_{n,p}(t_i)$ is a
$\bbG(n,(\frac{1}{3}(1+\frac{j}{s})\lambda+o(1))n^{-1})$ graph. By Corollary~\ref{cor:thm:giant-connected-comp}, the size of its largest connected component is less than
\[
\frac{1+\frac{1}{s}}{3}\Bigl(1+\frac{j}{s}\Bigr)
\frac{\lambda}{\lambda^*}n
\]
with high probability. Since $(1+\frac{1}{s})^2\lambda<\lambda^*$, we have
\[
\Bigl(1+\frac{1}{s}\Bigr)\frac{\lambda}{\lambda^*}
<
\frac{1}{1+\frac{1}{s}}.
\]
Therefore,
\[
\frac{1+\frac{1}{s}}{3}\Bigl(1+\frac{j}{s}\Bigr)
\frac{\lambda}{\lambda^*}n
<
\frac{1+\frac{j}{s}}{3(1+\frac{1}{s})}n
=
\frac{s+j}{3(s+1)}n
\leq
\frac{1}{3}\Bigl(1+\frac{j-1}{s}\Bigr)n
\leq
t_{i-1}+2.
\]
Hence, the largest connected component in $\calG_{n,p}(t_i)$ has size less than $t_{i-1}+2$ w.h.p.

The number of time points $t_0,\dots,t_{2s+2}$ does not depend on $n$. Thus, with high probability, for all $i\in\{1,\dots,2s+2\}$, the size of the largest connected component in $\calG_{n,p}(t_i)$ is less than $t_{i-1}+2$. Suppose that this event holds. We claim that, in this case, the size of the largest connected component in $\calG_{n,p}(t)$ is less than $t+2$ for all real $t\in[1,n-2]$. Pick such a $t$. It belongs to one of the intervals $[t_{i-1},t_i]$. Graph $\calG_{n,p}(t)$ is a subgraph of $\calG_{n,p}(t_i)$ because the multiplicity of every edge $(a,b)$ in $\calG_{n,p}(t)$ is a non-decreasing function of $t$. The size of the largest connected component in $\calG_{n,p}(t)$ is at most the size of the largest connected component in $\calG_{n,p}(t_i)$, which in turn is less than $t_{i-1}+2\leq t+2$.
\end{proof}

\begin{figure}[ht]
    \centering    \includegraphics[width=6in]{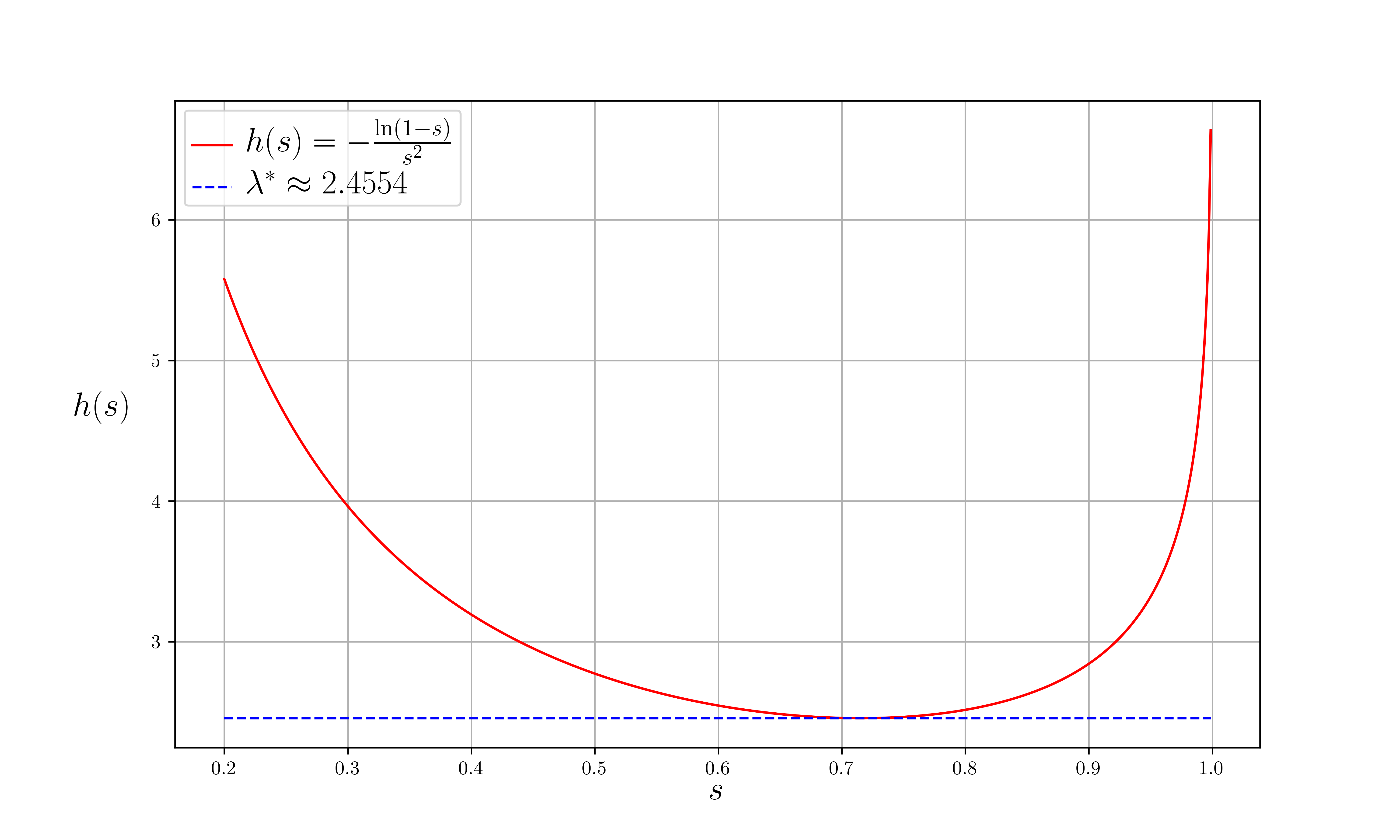}
    \caption{Function $h(s) = -\frac{\ln (1-s)}{s^2}$.}
    \label{fig:function-h}
\end{figure}

We now prove that $\lambda_1^*= -\frac{w^*}{(1 + \nicefrac{1}{2w^*})}\approx 2.4554$, where
$w^* = W_{-1}(- \nicefrac{1}{2\sqrt{e}})\approx -1.7564$ and $W_{-1}$ is the $-1$ branch of the Lambert $W$ function (see Figure~\ref{fig:function-h}).
\begin{claim}
$$\min_{s\in(0,1)}-\frac{\ln(1-s)}{s^2} = -\frac{w^*}{(1 + \nicefrac{1}{2w^*})}.$$
\end{claim}
\begin{proof}
Let $h(s) = -\frac{\ln(1-s)}{s^2}$. The derivative of $h$ equals:
$$
h'(s) = \frac{1}{(1-s)s^2} + \frac{2\ln(1-s)}{s^3}.
$$
It is equal to $0$ when $2\ln(1-s) + \frac{s}{1-s} = 0$. Let $y = 1-s$. Then, we have $2\ln y +  \frac{1-y}{y} = 0$, or
$-\frac{1}{2y} + \ln \frac{1}{y} = -\frac{1}{2}$. We exponentiate both sides and then multiply them by $-\nicefrac{1}{2}$:
$$-\frac{1}{2y} e^{-\frac{1}{2y}} = -\frac{1}{2\sqrt{e}}.$$
The solution on the $W_{-1}$ branch is $-\frac{1}{2y} = W_{-1}\big(-\frac{1}{2\sqrt{e}}\big)$ (by the definition of the Lambert $W$ function). The principal branch solution corresponds to $s=0$, which is not in the interval $(0,1)$. Denote $w^*  = W_{-1}\big(-\frac{1}{2\sqrt{e}}\big)$, $y^* = -\frac{1}{2w^*}$, and $s^* = 1 +\frac{1}{2w^*}$. Function $h(s)$ is minimized when $s=s^*$, since $s^*$ is the only zero of $h'(s)$ in $(0,1)$, and $\lim_{s\to 0^+} h(s)= \lim_{s\to 1^-} h(s)=\infty$. Thus,
$$\min_{s\in(0,1)} h(s) =
h(s^*) = -\frac{\ln (1-s^*)}{(s^*)^2}.$$
We know that $\ln (1-s^*) = -\frac{s^*}{2(1-s^*)}$ because $s^*$ is the solution of $2\ln(1-s)+\frac{s}{1-s}=0$. Therefore,
$$
h(s^*) = \frac{1}{2s^*(1-s^*)} =
-\frac{1}{
2 (1+\frac{1}{2w^*})\frac{1}{2w^*}
} =
-\frac{w^*}{
(1+\frac{1}{2w^*})
}.
$$
\end{proof}

\subsection{Reduction to the \texorpdfstring{$\calI_r(n,m)$}{I_r(n,m)} Model}
We now sketch a reduction from the model $\calI_r(n,\lambda)$, examined in this section, to the fixed-size model $\calI_r(n,m)$ studied in this paper. Here, $\calI_r(n,m)$ denotes a uniformly random instance with exactly $m$ distinct constraints $ab|\mathbf{c}$. Consider a random instance $\calI_r(n,\lambda)$ with $\lambda<\lambda_r^*$ and pick a small positive number $\rho$. As we just showed, this instance is satisfiable w.h.p. By the Chernoff bound, w.h.p. it contains at least $(\lambda-\rho)n/2$ distinct constraints $ab|\mathbf{c}$. Thus, w.h.p. $\calI_r(n,\lambda)$ is satisfiable and contains at least $(\lambda-\rho)n/2$ distinct constraints. Condition on this event. If we keep a uniformly random subset of $(\lambda-\rho)n/2$ distinct constraints in this instance, then we get a random instance distributed as $\calI_r(n,(\lambda-\rho)n/2)$. It is satisfiable w.h.p.

Similarly, if we choose $\lambda>\lambda_r^*$, then w.h.p. $\calI_r(n,\lambda)$ is not satisfiable and contains at most $(\lambda+\rho)n/2$ distinct constraints. Condition on this event. We can include additional uniformly random distinct constraints not already present in the instance so that it has exactly $(\lambda+\rho)n/2$ distinct constraints. The new instance has the same distribution as $\calI_r(n,(\lambda+\rho)n/2)$ and remains unsatisfiable.
\section{Upper and Lower Bounds on the Threshold for Quartets}\label{sec:ub-quartets}

In this section, we provide lower and upper bounds on the unsatisfiability threshold for the Quartets CSP on trees (\textsc{Quartet Reconstruction}). 

\begin{theorem}\label{th:quartet_upper_lower}
Consider a random instance of the \textsc{Quartet Reconstruction} Problem with $n$ variables and $m=\lambda n$ constraints.  
\begin{enumerate}
\item For $\lambda < \lambda^*_2/2$ (where $\lambda^*_2$ is defined in Equation~\eqref{eq:def:lambda-r}; 
$\lambda_2^*/2\approx 0.908$), the instance is satisfiable with high probability.
\item For $\lambda > 
\frac{\ln 3 - \tfrac{2}{3}}{\ln(81/65)}\approx 2.89$, the instance is not satisfiable with high probability.
\end{enumerate}
\end{theorem}
Part 1 of the theorem follows from Theorem~\ref{thm:threshold-main} because every quartet constraint $ab|cd$ is implied by the triplet constraint $ab|(c,d)$ with two negative examples $c$ and $d$ (so $r=2$ for this case). Indeed, if $c$ and $d$ are separated from $a$ and $b$ before $a$ and $b$ are separated from each other, then the quartet constraint $ab|cd$ is satisfied. Note that the converse does not hold. Specifically, the quartet constraint $ab|cd$ may be satisfied while the triplet constraint $ab|(c,d)$ is violated.

We now prove Part 2. Consider an instance $\calI$ of the Quartets problem. Let $T$ be an unrooted tree whose leaves are labeled with variables $x_1, \dots, x_n$. We say that two leaves are separated by an edge $(u,v)$ in $T$ if the unique path between them in $T$ passes through $(u,v)$. Observe that for every \emph{satisfiable} constraint $ab|cd$ in $\calI$, it is not possible for $a$ and $b$ to be separated by $(u,v)$ and for $c$ and $d$ to be separated by $(u,v)$. If that were the case, the paths from $a$ to $b$ and from $c$ to $d$ would both pass through $(u,v)$, and hence the constraint $ab|cd$ would not be satisfied.

We may assume that $T$ is binary, since refining an unrooted tree does not destroy any quartet constraint that is already satisfied.\footnote{The original argument referred to a partition of the vertices of $T$; the necessary condition for quartets uses the induced partition of the labeled leaves. Refining the tree lets us use the standard balanced-edge property for binary trees.} In every binary tree with $n$ labeled leaves, there exists an edge whose removal partitions the leaves of the tree into two parts, $L$ and $R$, such that the size of $L$ is between $n/3$ and $n/2$, and the size of $R$ is between $n/2$ and $2n/3$. Fix such an edge $(u,v)$ in $T$. The observation above implies that for every satisfiable constraint $ab|cd$, it is not possible that
\begin{enumerate}
    \item exactly one of the variables $a$ and $b$ belongs to $L$, and
    \item exactly one of the variables $c$ and $d$ belongs to $L$.
\end{enumerate}
We conclude that if the instance $\calI$ is satisfiable, then there exists a partitioning of the set of variables into sets $L$ and $R$ such that $n/3\leq |L|\leq n/2$ and, for every constraint $ab|cd$, conditions (1) and (2) are not both satisfied. Note that this is a necessary but not sufficient condition for $\calI$ to be satisfiable.

Our goal now is to show that if
\[
c > c^* \equiv \frac{\ln 3 - \frac{2}{3} \ln 2}{\ln\left( \frac{81}{65} \right)}\approx
2.89,
\]
then the random instance $\calI(m,n)$ with $m=cn$ is not satisfiable with high probability. To this end, we show that for every partition of the variables into $L$ and $R$ with $n/3\leq |L|\leq n/2$, there exists at least one constraint $ab|cd$ for which both conditions (1) and (2) hold.

To establish this claim, we fix an arbitrary partitioning of the variables into sets $L$ and $R$ with $|L| = k$ (where $n/3 \leq k \leq n/2$), compute the probability that, for this partitioning, every constraint $ab|cd$ does not satisfy both conditions (1) and (2), and then apply the union bound over all such partitions $L$, $R$, and values $k$ as described above.

For a given partition $L, R$ with $|L| = k$, the probability that $|\{a,b\} \cap L| = 1$ for a random pair $(a,b)$ is $ \frac{2k(n-k)}{n^2}. $ Consequently, the probability that for a random constraint $ab|cd$, we have $|\{a,b\} \cap L| = 1$ and $|\{c,d\} \cap L| = 1$ equals \[ \frac{4k^2(n-k)^2}{n^4}, \] and the probability that for all $m$ (where $m = cn$ ) random constraints $ab|cd$ at least one of the conditions (1) or (2) is violated is 
\[ \left(1 - \frac{4k^2(n-k)^2}{n^4} \right)^m. \]

The number of partitions of $n$ variables into sets $L$ and $R$ with $|L| = k$ is $\binom{n}{k}$. Hence, for a fixed value of $k$, the probability that there exist sets $L$ and $R$ with $|L| = k$ such that for every constraint at least one of the conditions (1) or (2) does not hold is upper bounded by \[ \binom{n}{k} \cdot \left(1 - \frac{4k^2(n-k)^2}{n^4} \right)^{cn}. \] Finally, the probability that there exists some $k$ and corresponding partitioning $L$, $R$ as above is upper bounded by \[ \sum_{k = \lceil n/3 \rceil}^{\lfloor n/2 \rfloor} \binom{n}{k} \cdot \left(1 - \frac{4k^2(n-k)^2}{n^4} \right)^{cn} \leq \frac{n}{6} \cdot \max_{k \in \{\lceil n/3 \rceil, \dots, \lfloor n/2 \rfloor\}} \binom{n}{k} \left(1 - \frac{4k^2(n-k)^2}{n^4} \right)^{cn}. \]

We now prove that the right-hand side tends to $0$ as $n \to \infty$. To this end, we consider the logarithm of the $n$-th root of the right-hand side and show that its limit as $n\to \infty$ is less than $0$ for $c > c^*$. Let \[ \vartheta_c = \lim_{n \to \infty} \ln \Bigg[ \Bigg( \max_{n/3 \leq k \leq n/2} \frac{n}{6} \cdot \binom{n}{k} \cdot \bigg(1 - \frac{4k^2(n-k)^2}{n^4} \bigg)^{cn} \Bigg)^{1/n} \Bigg]. \] Taking the $n$-th root of each factor, we get \[ \vartheta_c = \lim_{n \to \infty} \ln \Bigg( \max_{n/3 \leq k \leq n/2} \bigg(\frac{n}{6}\bigg)^{1/n} \cdot \binom{n}{k}^{1/n} \cdot \bigg(1 - \frac{4k^2(n-k)^2}{n^4} \bigg)^c \Bigg). \] We further simplify this expression by noting that $\big(n/6\big)^{1/n} \to 1$, and by applying the standard asymptotic approximation (stated formally in Lemma~\ref{lem:binom-nk-approx}), $\ln \big( \binom{n}{k}^{1/n} \big) \approx H(k/n)$ for large $n$, where $H(x)$ denotes the binary entropy function: 
\[ H(x) = -x \ln x - (1 - x) \ln(1 - x), \quad \text{for } x \in (0,1). \] We thus obtain \[ \vartheta_c \leq \lim_{n \to \infty} \max_{n/3 \leq k \leq n/2} \bigg( H(k/n) + \varepsilon_n + c \ln \bigg(1 - \frac{4k^2(n-k)^2}{n^4} \bigg) \bigg). \] 
Letting $x = k/n$, we rewrite the bound as 
\[ \vartheta_c \leq \lim_{n \to \infty} \max_{x \in [1/3, 1/2]} \bigg( H(x) + \varepsilon_n + c \ln \big(1 - 4x^2(1 - x)^2 \big) \bigg) = \max_{x \in [1/3, 1/2]} \bigg( H(x) + c \ln \big(1 - 4x^2(1 - x)^2 \big) \bigg). \]
Denote $g(x) = 4x^2(1 - x)^2$. Then, 
\[ \vartheta_c \leq \max_{x \in [1/3, 1/2]} \bigg( H(x) + c \ln \big(1 - g(x) \big) \bigg). \] Corollary~\ref{cor:lem:f-decreasing}, proved below, shows that $\vartheta_c < 0$ for $c > c^*$.

\medskip

\noindent We now prove Lemma~\ref{lem:binom-nk-approx}, Lemma~\ref{lem:f-decreasing}, and Corollary~\ref{cor:lem:f-decreasing}.

\begin{lemma}\label{lem:binom-nk-approx}
There exists a function $ \varepsilon_n = o(1) $ as $ n \to \infty $ such that for all integers $ k \in \{1, \dots, n-1\} $,
\[
\left| \frac{1}{n} \ln \binom{n}{k} - H\left( \frac{k}{n} \right) \right| \leq \varepsilon_n.
\]
\end{lemma}

\begin{proof}
We will use Stirling's approximation in the form
$\ln(n!) = n \ln n - n + O(\ln n)$,
as $ n \to \infty$. Write
\begin{align*}
\ln \binom{n}{k} &= \ln\Big(\frac{n!}{k!(n-k)!}\Big)\\
&= \ln n! - \ln k! - \ln(n - k)! \\
&= \left(n \ln n - n\right) - \left(k \ln k - k\right) - \left((n - k) \ln(n - k) - (n - k)\right) + O(\ln n) \\
&= n \ln n - k \ln k - (n - k) \ln(n - k) + O(\ln n).
\end{align*}
Using $n\ln n = k\ln n + (n-k)\ln n$, we obtain
\[
\ln \binom{n}{k} = - k \ln \frac{k}{n} - (n - k) \ln\left(\frac{n - k}{n}\right) + O(\ln n).
\]
Dividing both sides by $n$, we get
\[
\frac{1}{n}\ln \binom{n}{k}
=
- \frac{k}{n} \ln \frac{k}{n}
- \frac{n-k}{n} \ln\left(\frac{n-k}{n}\right)
+ O\left(\frac{\ln n}{n}\right).
\]
Thus, the desired bound holds with $\varepsilon_n=C\ln n/n$ for a sufficiently large constant $C$.
\end{proof}

\begin{figure}[h!]
  \centering
\includegraphics[width=0.7\linewidth]{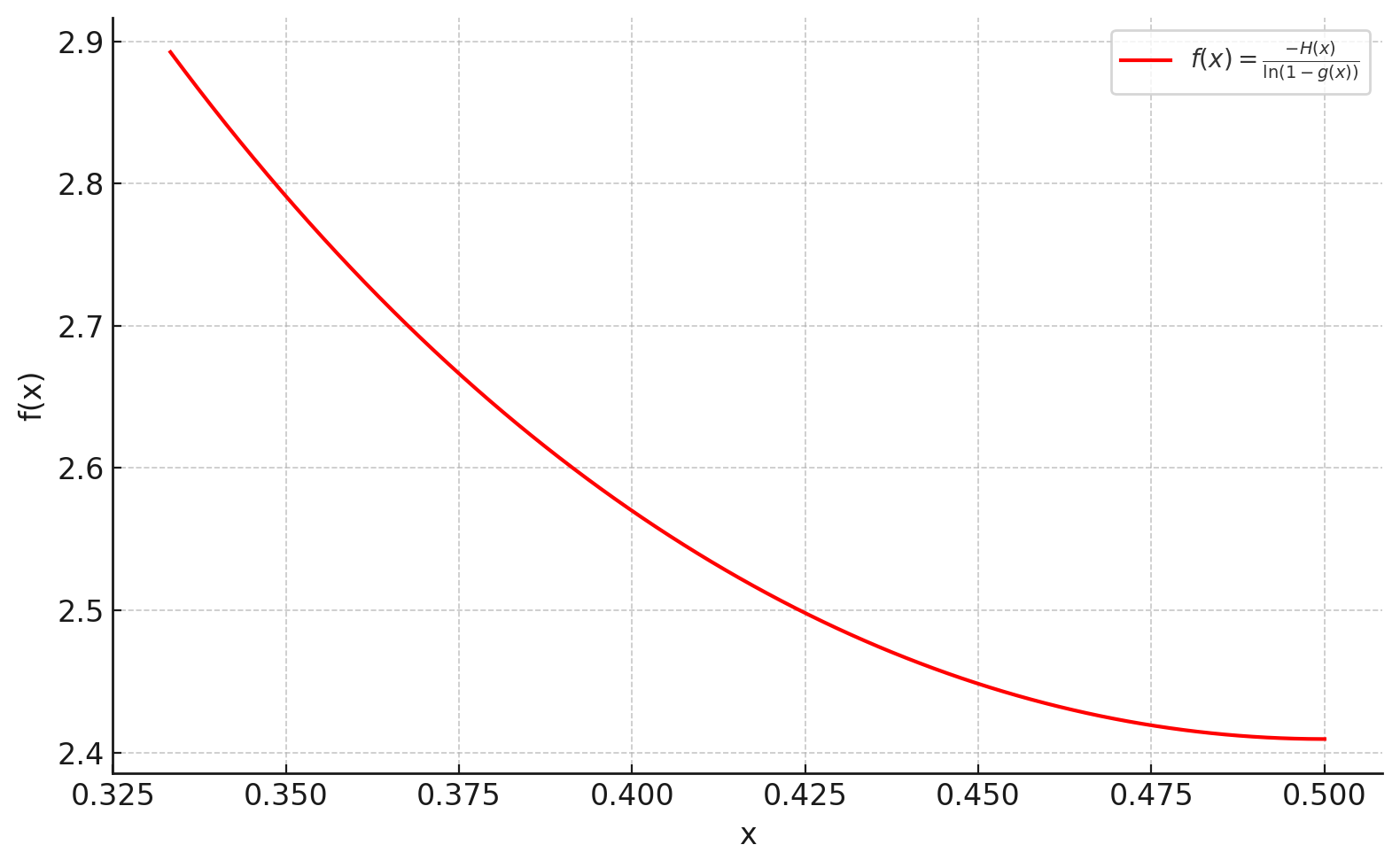} 
  \caption{Plot of $ f(x) = \frac{-H(x)}{\ln(1 - g(x))} $ over the interval $[1/3, 1/2]$.}
  \label{fig:fx_plot}
\end{figure}

\begin{lemma}\label{lem:f-decreasing}
The function
\[
f(x)=\frac{-H(x)}{\ln(1-g(x))}
\]
is strictly decreasing on $\left[\frac13,\frac12\right]$, where
\[
H(x)=-x\ln x-(1-x)\ln(1-x),
\qquad
g(x)=4x^2(1-x)^2.
\]
\end{lemma}

\noindent
To prove this lemma, we show that $f'(x) < 0$ for $x\in[1/3,1/2]$. We present the details in Section~\ref{sec:lem:f-decreasing}.

\begin{corollary}\label{cor:lem:f-decreasing}
Let $ f(x) $ be as in Lemma~\ref{lem:f-decreasing}. Then
\[
\sup_{x\in[1/3,1/2]} f(x)
=
f\left( \frac{1}{3} \right)
=
-
\frac{\ln 3 - \frac{2}{3} \ln 2}{\ln\left( \frac{65}{81} \right)}
=
\frac{\ln 3 - \frac{2}{3} \ln 2}{\ln\left( \frac{81}{65} \right)}
\approx 2.89.
\]
\end{corollary}

\begin{proof}
By Lemma~\ref{lem:f-decreasing}, the function $ f(x) $ is strictly decreasing on the interval $[1/3,1/2]$. Therefore, its supremum over this interval is attained at $x=1/3$. We have
\[
H\left( \frac{1}{3} \right)
=
-\left( \frac{1}{3} \ln \frac{1}{3} + \frac{2}{3} \ln \frac{2}{3} \right)
=
\ln 3 - \frac{2}{3} \ln 2,
\]
and
\[
g\left( \frac{1}{3} \right)
=
4 \cdot \left( \frac{1}{3} \right)^2 \cdot \left( \frac{2}{3} \right)^2
=
\frac{16}{81}.
\]
Thus,
\[
f\left( \frac{1}{3} \right)
=
-\frac{\ln 3 - \frac{2}{3}\ln 2}{\ln(65/81)}
=
\frac{\ln 3 - \frac{2}{3}\ln 2}{\ln(81/65)}.
\]
\end{proof}

\section{Bounds for the Probability of Success of \texorpdfstring{$t$}{t}-wise Uniform Distributions}\label{sec:bounds-uniform-distr}

We now show how to certify that the value of the optimal tree for a random \textsc{Triplet} and \textsc{Quartet Reconstruction} instance is at most $\tfrac{5}{9}+\epsilon$ (strong refutation task). For triplets, the refutation algorithm succeeds w.h.p if $m > C'_{\epsilon}n$ and for quartets the algorithm succeeds w.h.p. if  $m > C''_{\epsilon}n^{3/2}$ for some  constants $C'_{\epsilon}$ and $C''_{\epsilon}$ depending only on $\epsilon$. This result relies on Theorem~\ref{thm:t-wise-phy-CSP-refute}, which we prove in Appendix~\ref{subsec:t-wise}. To prove the results using Theorem~\ref{thm:t-wise-phy-CSP-refute}, we show that the expected value of the triplet predicate on
every \textit{pairwise} independent distribution on the leaves of an arbitrary tree is at most $5/9$ and similarly  the expected value of the quartet predicate on
every \textit{3-wise} independent distribution is at most $5/9$. See Lemmas~\ref{lem:pairwise-unirofm-triplets-bound}
and~\ref{lem:3wise-unirofm-quartet-bound}.

\subsection{Triplet Reconstruction}\label{sec:beta-triplets}
Let $T$ be a binary, rooted tree. For an internal node $u$, we use $l(u)$ and $r(u)$ to denote the left and right child of $u$ respectively and $T_u$ for the tree rooted at $u$. We will use $f_{\uprightsf{tri}}$ to denote the triplet predicate, that is, $f_{\uprightsf{tri}}(a,b,c)=1$ if $ab|c$ is satisfied.

\begin{lemma}
\label{lem:pairwise-unirofm-triplets-bound}
    Let $T$ be a rooted binary tree with $q$ leaves. For any pairwise uniform distribution $\mathcal{D}$ over $[q]^{3}$, it holds that, for $(a,b,c)\sim \mathcal{D}$:
    \begin{align*}
        \Prob{f_{\uprightsf{tri}}(a,b,c)=1}\leq \frac{5}{9}
    \end{align*}
\end{lemma}
\begin{proof}
      We first rewrite the probability that the predicate is satisfied using the law of total probability and conditioning on the lowest common ancestor of leaves $a$ and $b$. We have that:
    \begin{align*}
        \Prob{f_{\uprightsf{tri}}(a,b,c)=1}=\sum_{v\in T}\Prob{f_{\uprightsf{tri}}(a,b,c)=1 \text{ and } \lca(a,b)=v}
    \end{align*}
    For a fixed vertex $v$, we have that
    \begin{align*}
        \Prob{f_{\uprightsf{tri}}(a,b,c)=1 \text{ and } \lca(a,b)=v}=&\Prob{a\in T_{r(v)} \text { and }b\in T_{l(v)}\text{ and } c\in T_{v}^{c}}\\&+\Prob{a\in T_{l(v)} \text { and }b\in T_{r(v)}\text{ and } c\in T_{v}^{c}}
    \end{align*}
    where we have used  $T_v^{c}$ to denote the vertices that are not in the subtree rooted at $v$. We now have that
    \begin{align*}
        \Prob{a\in T_{l(v)} \text { and }b\in T_{r(v)}\text{ and } c\in T_{v}^{c}}\leq \min(\Prob{a\in T_{l(v)} \text {, }b\in T_{r(v)}},\\ \Prob{a\in T_{l(v)} \text{, } c\in T_{v}^{c}},\\ \Prob{b\in T_{r(v)}\text{, } c\in T_{v}^{c}})
    \end{align*}
    We use the following convenient notation, for a vertex $v$, we let $L(v)$ be the fraction of leaves on the left subtree rooted at $v$ and similarly $R(v)$ for the fraction of leaves on the right subtree. Finally, we use $E(v)$ to denote the fraction of leaves that are not in the subtree rooted at $v$, in other words:
    \[
        E(v)=1-L(v)-R(v).
    \]
    Using this notation and the fact that $\mathcal{D}$ is a pairwise uniform distribution, we can write the bound on the probability $\Prob{a\in T_{l(v)} \text { and }b\in T_{r(v)}\text{ and } c\in T_{v}^{c}}$ as:
    \begin{align*}
        \Prob{a\in T_{l(v)} \text { and }b\in T_{r(v)}\text{ and } c\in T_{v}^{c}}\leq \min\left(L(v) R(v), L(v)E(v), R(v)E(v)\right)
    \end{align*}
    we have that
    \begin{equation} \label{eq:triplet_refute_strong_bound}       \Prob{f_{\uprightsf{tri}}(a,b,c)=1}\leq 2\sum_{v}\min(L(v)R(v), L(v)E(v), R(v)E(v)).
    \end{equation}
    For an edge $e=(x,y)$ we let $A_e$ and $B_e$ be the sets of leaves of the two connected components that are created if $e$ is removed from the tree.
        Let $e^*=(x^*, y^*)$ be an edge such that $1/3\leq\frac{|A_{e^*}|}{q}\leq 2/3$ (such an edge always exists in a binary tree) and assume without loss of generality that $x^*$ is the left child of $y^*$ and furthermore that $A_{e^*}$ is the set of leaves of the subtree rooted at $x^*$. We also define function $\kappa$ as follows:
        \begin{align*}
            \kappa(v)=\begin{cases}
                E(v)\cdot R(v) \text{, if $v$ is in the path from $y^*$ to the root; }\\
                L(v)\cdot R(v) \text{, otherwise.}
            \end{cases}
        \end{align*}
        Note that by Equation \ref{eq:triplet_refute_strong_bound} we have that
        \begin{align}
        \label{eq:triplets_refute_kappa_bound}
            \Prob{f_{\uprightsf{tri}}(a,b,c)=1}\leq \sum_{v}2\kappa(v).
        \end{align} Now consider an experiment where two leaves $X,Y$ are sampled uniformly and independently among the leaves of $T$. Let $\mathcal{E}_{v}$ be the event that $v$ is the first common vertex in the paths from $X$ to $y^*$ and $Y$ to $y^*$. We now use the law of total probability to get that
        \begin{align*}
            \Prob{X,Y \in A_{e^*} \text{ or } X,Y\in B_{e^*}}=\sum_{v}\Prob{(X,Y \in A_{e^*} \text{ or } X,Y\in B_{e^*})\text{ and } \mathcal{E}_v}
        \end{align*}
        Observe that $\Prob{(X,Y \in A_{e^*} \text{ or } X,Y\in B_{e^*})\text{ and } \mathcal{E}_v}=2\kappa(v)$ which in turn means that
        \begin{align*}
            \Prob{X,Y \in A_{e^*} \text{ or } X,Y\in B_{e^*}}=\sum_{v}2\kappa(v)
        \end{align*}
        We now use the bound from Equation \ref{eq:triplets_refute_kappa_bound} to get that:
        \begin{align*}
            \Prob{f_{\uprightsf{tri}}(a,b,c)=1}&\leq\Prob{X,Y \in A_{e^*} \text{ or } X,Y\in B_{e^*}}\\
            &=\left(\frac{|A_{e^*}|}{q}\right)^2+\left(1-\frac{|A_{e^*}|}{q}\right)^2\\
            &\leq \frac{5}{9},
        \end{align*}
        where we have used that $1/3\leq \frac{|A_{e^*}|}{q}\leq 2/3$
\end{proof}

\subsection{Quartets with 3-wise uniform distributions}\label{sec:beta-quartets}
We now prove that for every rooted, binary tree $T$ and every $3$-wise uniform distribution $\calD$ the probability of the event that $(a,b,c,d)\sim \calD$ satisfies the quartet predicate $ab|cd$ is bounded away from $1$. For convenience, let $f_{\uprightsf{quar}}(a,b,c,d)=\ind{ab|cd \text{ is satisfied}}$. 
Formally, we have the following lemma.
\begin{lemma}\label{lem:3wise-unirofm-quartet-bound}
Let $T$ be a rooted binary tree with $q$ leaves and $\calD$ a $3$-wise uniform distribution over $[q]^4$, then for $(a,b,c,d)\sim \calD$, we have that:
\begin{align*}
    \Prob{f_{\uprightsf{quar}}(a,b,c,d)=1}\leq \frac{5}{9}
\end{align*}
\end{lemma}

\begin{proof}
Our proof uses a simple reduction from quartet reconstruction to triplet reconstruction. Consider a tree $T$ and a quartet constraint $ab|cd$ as well as a tree $T_{r=p(a)}$ constructed from $T$ by the following procedure. First, smooth out the root of $T$ (replace the root and its two outgoing edges by an edge connecting the two children of the root) then remove $a$ and make the parent of $a$ the root of the new tree. Observe that the quartet constraint $ab|cd$ is satisfied in $T$ if and only if the triplet constraint $cd|b$ is satisfied in $T_{r=p(a)}$. We now use that conditioned on where $a$ is, $b,c$ and $d$ come from a pairwise uniform distribution and use Lemma \ref{lem:pairwise-unirofm-triplets-bound} to bound the probability that the triplet constraint is satisfied. We have that:
\begin{align*}
    \Prob{f_{\uprightsf{quar}}(a,b,c,d)=1}&=\frac{1}{q}\sum_{v}\Prob{f_{\uprightsf{tri}}(T_{r=p(v)}, b,c,d)=1|a=v}\\
    &\leq\frac{1}{q}\sum_{v}\frac{5}{9}\\
    &=\frac{5}{9}
\end{align*}
\end{proof}

Our above approach also extend to popular ordering CSPs from the literature such as the Betweenness and the Cyclic Ordering predicates~\cite{guruswami2008beating,guruswami2015towards}. For the analogous bounds, we defer the reader to Appendix~\ref{subsec:BTW}.
\newpage

\bibliographystyle{abbrvnat}
\bibliography{references.bib}
\newpage
\appendix
\begin{appendices}

\section{Preliminaries}\label{sec:prelim}

In this section, we recall the definitions of the two most important phylogenetic CSP problems, namely the triplet and quartet consistency problems, and then introduce the random phylogenetic CSP model. A more general definition of arbitrary ordering and phylogenetic CSPs is given in Appendix~\ref{sec:prelim-general}.
We first formally define triplet and quartet predicates.

\begin{definition}
A \emph{triplet constraint} $\mathit{triplet}$ (often denoted $(ab \mid c)$) is a constraint of arity $3$. It takes three leaves $a,b,c$ of a binary rooted tree and is satisfied if the least common ancestor of $a$ and $b$, denoted $\lca(a,b)$, lies in the subtree rooted at $\lca(a,c)$. Equivalently, the constraint is satisfied if $\lca(a,c) = \lca(b,c)$.
\end{definition}

\begin{definition}
A \emph{quartet constraint} $\mathit{quartet}$ (often denoted $(ab \mid cd)$) is a constraint of arity $4$. It takes four leaves $a,b,c,d$ of a binary rooted tree and is satisfied if the unique shortest paths from $a$ to $b$ and from $c$ to $d$ are vertex-disjoint.
\end{definition}


We now give a formal definition of a phylogenetic CSP with predicate $f_{\uprightsf{phy}}$. For now, the reader may assume that $f_{\uprightsf{phy}}$ is either a triplet or a quartet constraint. In the general case, however, $f_{\uprightsf{phy}}$ may be an arbitrary phylogenetic constraint, as described in Appendix~\ref{sec:prelim-general}. A predicate, or more generally a payoff function, $f_{\uprightsf{phy}}$ of arity $k$ takes as arguments the leaves of a tree $T$ together with the tree $T$ itself, and returns a value in $[0,1]$. For brevity, we will often omit the parameter $T$ when it is clear from the context to which tree the leaves belong.

\begin{definition}[Phylogenetic CSP]\label{def:phy-csp}
An instance $\calI_{f_{\uprightsf{phy}}} = (V,\C)$ of a phylogenetic CSP $\Gamma$ with payoff function $f_{\uprightsf{phy}}$ consists of a set of variables $V$ and a set of constraints $\C$ with payoff $f_{\uprightsf{phy}}$. Each constraint in $\C$ is a tuple of $k$ distinct variables from $V$, representing an application of $f_{\uprightsf{phy}}$ to these $k$ variables.

A solution to the instance is given by a bijective assignment $\phi$ of the variables in $V$ to leaves of an ordered tree $T$; the tree $T$ is also part of the solution. The value of a solution $\phi$, denoted by $\val(\phi,\calI)$, is the average value over all constraints in $\C$:
\begin{align*}
    \val(\phi,\calI)
    = \frac{1}{|\C|}\sum_{(u_1,\ldots,u_k)\in \C}
    f_{\uprightsf{phy}}(T,\phi(u_1),\phi(u_2),\ldots,\phi(u_k)).
\end{align*}
\end{definition}

We note that the results described in the following sections apply to arbitrary ordered phylogenetic CSPs, unless stated otherwise. However, we encourage the reader to focus on the triplet and quartet consistency problems, which are unordered phylogenetic CSPs. Throughout the paper, we use the notion of the best approximation factor achievable by a biased random assignment scheme, denoted by $\alpha^*(f_{phy})$ for a phylogenetic payoff function $f_{phy}$; this notion is defined in Definition~\ref{def:biased}. For triplet and quartet constraints, $\alpha^*(f_{phy}) = 1/3$, which is the approximation factor of the most basic random assignment algorithm that assigns variables to random leaves of an arbitrary tree.

Finally, we define a random instance of a phylogenetic CSP with payoff function $f_{\uprightsf{phy}}$.
\begin{definition}\label{def:random-instance}
A random instance $\calI_{f_{\uprightsf{phy}}}(n,m)$ of a phylogenetic CSP with payoff function $f_{\uprightsf{phy}}$ consists of a set $V$ of $n$ variables and a multiset of $m$ constraints. Each constraint is an ordered tuple $(X^1,\ldots,X^k)$ of \emph{distinct} variables from $V$, sampled independently and uniformly at random from the set of all such tuples.
\end{definition}

We use $\Phi(\calI)$ to denote the set of all solutions for instance $\calI$ and $\opt(\calI)$ to denote the value of an optimal solution in $\Phi(\calI)$. We say that an instance $\calI$ is satisfiable if $\opt(\calI)=1$ and unsatisfiable otherwise.

\subsection{Notion of Coarse Solutions.}\label{sec:coarse-soln} 
Following~\cite{chatziafratis2023triplet}, we define \emph{coarse solutions} of Ordered Phylogenetic CSPs.
For a fixed accuracy $\epsilon$, a \emph{coarse solution} $\xi$ is a mapping from the $n$ variables to a fixed-size binary tree, denoted as $T_\xi$, with only $q$ leaves ($q$ is a large constant depending on $\epsilon$). Each leaf of $T_\xi$ (referred to as a \textit{coarse leaf}) can be thought of as a ``bucket" containing a subset of variables, and each variable is assigned to exactly one coarse leaf. Furthermore, the coarse solution $\xi$ assigns a color to every coarse leaf. Therefore, variables located at the same coarse leaf of $T_\xi$ must have the same color, while variables assigned to different coarse leaves may have the same or different colors. Importantly, coarse solutions $\xi$ are such so that the number of variables having the same color is at most $\delta n$ and $\delta = O(\nicefrac{1}{q})$, so we have essentially partitioned the $n$ variables into roughly balanced parts. We write that such a coarse solution $\xi$ is in class $\Xi_{\delta, q}$. We denote the coarse leaf assigned to variable $x$ by $\xi(x)$ and the color of $x$ by $\clr(\xi(x))$.

Once we have this bucketing into $q$ buckets, we use it to evaluate the performance of the coarse tree solution $\xi$ in two different ways, in terms of the fraction of input constraints that it satisfies. But there seems to be a problem since $T_{\xi}$ is not a true solution for the Phylogenetic CSP: in particular, what happens if two or more of the $n$ variables are mapped to the same bucket (this is certainly going to happen since we map $n$ variables to the much smaller set of $q$ leaves of the coarse solution)?



There are two value functions, $\val^-(\xi,\calI)$ and $\val^+(\xi,\calI)$%
, which represent the pessimistic and optimistic evaluations of the coarse solution, respectively. Let $f_{\uprightsf{phy}}(x_1,\dots,x_k)$ be a payoff function. If all variables $x_1,\dots, x_k$ have distinct colors in the coarse solution, i.e., $\clr(\xi(x_i)) \neq \clr(\xi(x_j))$ for all $i\neq j$, then we define
$$
f_{\uprightsf{phy}}^-(\xi(x_1),\dots,\xi(x_k)) = 
f_{\uprightsf{phy}}^+(\xi(x_1),\dots,\xi(x_k))
=f_{\uprightsf{phy}}(\xi(x_1),\cdots,\xi(x_k)).
$$
If two or more variables have the same color (i.e., $\clr(\xi(x_i)) = \clr(\xi(x_j))$ for some $i \neq j$), then we let
$f_{\uprightsf{phy}}^-(\xi(x_1),\dots,\xi(x_k))
= 0$ and
$f_{\uprightsf{phy}}^+(\xi(x_1),\dots,\xi(x_k)) = 1$. We then define $\val^-(\xi)$ and $\val^+(\xi)$ as the average of $f_{\uprightsf{phy}}^-$ and $f_{\uprightsf{phy}}^+$ over all constraints $(x_1,\dots,x_k)$:
\begin{align}
\label{eq:def:val-pm-1}
\val^-(\xi) &= 
\operatorname*{Avg}\limits_{(x_1,\dots,x_k)\in\calI}
f_{\uprightsf{phy}}^-(\xi(x_1),\dots,\xi(x_k))\\
\val^+(\xi) &= 
\label{eq:def:val-pm-2}
\operatorname*{Avg}\limits_{(x_1,\dots,x_k)\in\calI}
f_{\uprightsf{phy}}^+(\xi(x_1),\dots,\xi(x_k)).
\end{align}

We will need the following result from~\cite*{chatziafratis2023triplet}.
\begin{lemma}[Lemma 6.4 in~\citet{chatziafratis2023triplet}]\label{lem:cs-good}
Consider an instance $\calI$ of (ordered) Phylogenetic CSP. For every solution $\varphi$ and every positive $q$, there exists a coarse solution $\xi \in \Xi_{\frac{16}{q}, q}$ such that
$$
\val^+(\xi,\calI) \geq \val(\varphi,\calI).
$$    
\end{lemma}

\subsection{General Phylogenetic and Ordering CSPs} \label{sec:prelim-general}

Here, we define ordering and phylogenetic payoff functions.

\subsubsection{Ordering CSPs}
\label{subsec:ordering_definitions}
Our definition follows that of \cite{GHMRC11}. We start by the definition of the permutation pattern.
\begin{definition}[Permutation Pattern]
A permutation pattern $P(x_1,x_2,\ldots,x_k)$ is a permutation (ordering) of  $x_1,x_2,\ldots,x_k$, i.e. $x_{i_1}<x_{i_2}<\ldots<x_{i_k}$. Here $k$, the  arity of the pattern, is a fixed constant.
\end{definition}

Similarly to finite domain CSPs, an ordering CSP is defined by one (or more) \textit{ordering payoff function}. We now give the definition of an ordering payoff function.
\begin{definition}[Ordering payoff $f_{\uprightsf{ord}}$]
    An ordering payoff function $f_{\uprightsf{ord}}$ is described by a list of distinct ordering patterns of arity $k$ (out of $k!$ possible ordering patterns). Each pattern in the list is associated with a payoff, i.e. a real value in $[0,1]$. To evaluate the ordering payoff function on a permutation $\sigma$ of variables $l_1,l_2,\ldots,l_k$, denoted by $f(\sigma, l_1,l_2,\ldots, l_k)$, we search for a pattern in the list such that $\sigma$ matches $P$, i.e., if $x_{i_1}<x_{i_2}<\ldots<x_{i_k}$ according to $P$ then $\sigma(l_{i_1})<\sigma(l_{i_2})<\ldots <\sigma(l_{i_k})$. If no such pattern exists then the evaluation is $0$, otherwise the evaluation is the payoff associated with the pattern that matches $\sigma$. 
\end{definition}
To give a concrete example, for Betweenness, a constraint $a|b|c$ is satisfied by a permutation $\pi$ if $\pi(a)<\pi(b)<\pi(c)$ or $\pi(c)<\pi(b)<\pi(a)$, thus the payoff function in that case will be described by two permutation patterns $P_1(x_1,x_2,x_3)$ and $P_2(x_1,x_2,x_3)$, that are $x_1<x_2<x_3$ and $x_3<x_2<x_1$ respectively.
We say that an ordering payoff function $f_{\uprightsf{ord}}$ is satisfiable if there exists a pattern with payoff $1$, we call such patterns satisfying patterns. We will only deal with satisfiable payoff functions in this paper.

 An ordering CSP $\Lambda$ is defined by one or more payoff functions of arity $k$. Even though our results can be generalized to ordering CSPs with multiple payoff functions, for ease of presentation we focus on problems with only one payoff function for this paper. Note that Maximum Acyclic Subgraph, Betweenness and Cyclic Ordering are all ordering CSPs with one payoff function.

 We now give the formal definition of an ordering CSP.

\begin{definition}[Ordering CSP]
    An instance $\calI=(V,\C)$ of an ordering CSP $\Lambda$ with payoff function $f_{\uprightsf{ord}}$ is described by a set of variables $V$ and a set of constraints $\C$ with payoff $f_{\uprightsf{ord}}$. Every constraint in $\C$ is a tuple of $k$ distinct variables in $V$, representing a constraint $f_{\uprightsf{ord}}$ on the $k$ variables. A solution to the instance is given by a total ordering (permutation) $\pi$ of the variables in $V$. The value of a constraint $(u_1,u_2,\ldots,u_k)$ is the evaluation of its payoff $f_{\uprightsf{ord}}$ on the induced ordering $\pi(u_1), \pi(u_2), \ldots, \pi(u_k)$. The value of the solution $\pi$, denoted by $\val(\pi, \calI)$, is the average value over all constraints in $\C$:
\begin{align*}
    \val(\pi, \calI)=\frac{1}{|\C|}\sum_{(u_1,\ldots,u_k)\in \C}f_{\uprightsf{ord}}(\pi(u_1),\ldots,\pi(u_k)).
\end{align*}
\end{definition}

\subsubsection{Phylogenetic CSPs}
\label{subsec:phylogenetic_definitions}
We now define a general phylogenetic payoff function $f_{\phy}$, which is defined on the leaves of an ordered binary tree (that is, a binary tree in which each internal node has a designated left and right child, inducing a left-to-right order on the leaves). Note that every ordering payoff function can be viewed as a special case of a phylogenetic payoff function.
We begin by introducing the notion of tree patterns, which relies on the concepts of isomorphism and homeomorphism between trees.


\begin{definition}[Isomorphism of trees with labeled leaves]
    Let $T_1 = (V_1, E_1)$, $T_2 = (V_2, E_2)$ be trees and $u_1,u_2,\ldots,u_k$, $v_1,v_2,\ldots,v_k$ be distinct leaves in $T_1$ and $T_2$ respectively. We say that $T_1$ with labeled leaves $u_1,\ldots,u_k$ is isomorphic to $T_2$ with labeled leaves $v_1,\ldots,v_k$ if there exists a graph-theoretic bijection $g: V_1 \to V_2$ with
    \[
      \{v,w\}\in E_1\iff\{g(v),g(w)\}\in E_2,
      \quad\text{and}\quad
      g(u_i)=v_i \quad\forall i\in [k].
    \]
    \textbf{Note:} For ordered rooted trees, we additionally require that $g$ preserves the left-to-right order of vertices. That is, if node $v$ is the left (right) child of node $w$ in $T_1$ then $g(v)$ must be the left (right) child of node $g(w)$ in $T_2$.
\end{definition}


\begin{definition}[Homeomorphism of trees with labeled leaves]
\label{def:homeomorphism_of_trees}
    Let $T_1$, $T_2$ be trees with $u_1,u_2,\ldots,u_k$ and $v_1,v_2,\ldots,v_k$ being distinct leaves in $T_1$ and $T_2$ respectively. We say that $T_1$ with labeled leaves $u_1,\ldots,u_k$ is homeomorphic to $T_2$ with labeled leaves $v_1,\ldots,v_k$ if they can be transformed to isomorphic trees $T_1'$ and $T_2'$, by repeatedly using the following operations:
    \begin{enumerate}
    \item Removing every non-labeled leaf in $T_1$ or $T_2$ (i.e., any leaf other than $u_1, \ldots, u_k$ in $T_1$; and any leaf other than $v_1, \ldots, v_k$ in $T_2$);
    \item Smoothing out vertices of degree $2$ (except the root). If a vertex $w$ that is not the root in $T_1$ or $T_2$ has degree $2$, remove $w$ and connect its two neighbors (i.e., contract the path through $w$);
    \item Removing the root of $T_1$ or $T_2$ if its degree is $1$ and making its only child the new root.
    \end{enumerate}
     \textbf{Note:} For ordered rooted trees, the isomorphism between $T_1'$ and $T_2'$ must also preserve the order of vertices as defined above.
\end{definition}

\begin{definition}[Tree Pattern]
A tree pattern $P(x_1,x_2,\ldots,x_k)$ is a rooted, binary tree with $k$ leaves that are labeled by variable names $x_1,x_2,\ldots,x_k$. Here, $k$ is a fixed constant. In an ordered tree pattern we require that $P(x_1,x_2,\ldots, x_k)$ is also ordered.
\end{definition}

Let $T$ be a tree with $n$ leaves and $P(x_1,x_2,\ldots,x_k)$ be a pattern. We say that $T$ with labeled leaves $u_1,u_2,\ldots,u_k$ matches pattern $P(x_1,x_2,\ldots, x_k)$ if $T$ with labeled leaves $u_1,u_2,\ldots,u_k$ is homeomorphic to $P$ with labeled leaves $x_1,x_2,\ldots,x_k$. (For ordered patterns, we use the definition of homeomorphism for ordered trees.) When $T$ is clear from the context we equivalently say that leaves $u_1,u_2,\ldots,u_k$ match pattern $P$.

We now give a formal definition of phylogenetic payoff function. The definition follows \cite{chatziafratis2023triplet}.

\begin{definition}[Phylogenetic payoff $f_{\uprightsf{phy}}$]
    A phylogenetic payoff function $f_{\uprightsf{phy}}$ is described by a list of distinct (non-homeomorphic) trees (patterns) with $k$ leaves. The leaves are labeled $x_1,x_2,...,x_k$. Every tree in the list is associated with a payoff, i.e. a real value in $[0,1]$. To evaluate $f_{\uprightsf{phy}}$ on a tree $T$ with labeled leaves $l_1,l_2,\ldots, l_k$, denoted by $f(T,l_1,\ldots,l_k)$, we search for a pattern in the list such that $T$ with labeled leaves $l_1,l_2,\ldots,l_k$ is homeomorphic to that pattern (matches that pattern). If no such pattern exists then the evaluation is $0$, otherwise the evaluation is the payoff associated to the pattern that is homeomorphic to $T$. 
\end{definition}

Throughout the paper, we think of $f_{\uprightsf{phy}}$ as specifying a constraint for how variables $x_1,x_2,...,x_k$ should be related in the final tree $T$. Similarly to ordering payoff functions,  we say that a phylogenetic payoff function $f_{\uprightsf{phy}}$ is satisfiable if there exists a pattern with payoff $1$. We will only deal with satisfiable payoff functions $f_{\uprightsf{phy}}$ in this paper. We sometimes omit $T$ when it is clear from the context and write $f_{\uprightsf{phy}}(l_1,\ldots,l_k)$ instead of $f_{\uprightsf{phy}}(T,l_1,\ldots,l_k)$.

A phylogenetic CSP $\Gamma$ is defined by one or more payoff functions of arity $k$. Although our results extend to the more general setting with multiple payoff functions, for ease of presentation we restrict attention to phylogenetic CSPs with a single payoff function. In this case, we use Definition~\ref{def:phy-csp}.

\begin{remark}
\label{rmrk: Ordered phylogenetic CSPs}
In ordered phylogenetic CSPs there is a distinction between the left and right child of every internal node. Specifically,  it is easy to see that the class of ordered phylogenetic CSPs contains both ordering CSPs and Phylogenetic CSPs.
\end{remark}

Finally, we provide the definition of a biased random assignment given by \cite{chatziafratis2023triplet}.

\begin{definition}\label{def:biased}[\cite{chatziafratis2023triplet}]
A \emph{biased} random assignment algorithm for a Phylogenetic CSP with payoff function $f_{\uprightsf{phy}}$ is defined by a tree $T$ with $q$ leaves $1,\dots, q$ and probabilities $p_1,\dots,p_q$. First, the algorithm assigns variables in $V$ ($|V|=n$) independently to random leaves of 
 $T$. The probability of assigning variable $v$ to leaf $l$ is $p_l$. Then, it recursively partitions variables that got mapped to the same leaf,
splitting\footnote{In fact, after variables are assigned to the $q$ leaves of the tree $T$, one can arbitrarily separate the variables that got mapped to the same leaf, without affecting the performance of the best biased assignment, see also~\cite{chatziafratis2023triplet}.} them w.p. 1/2 to left or to right at every step, to get a hierarchy over the $n$ variables. We denote the approximation factor of this algorithm by $\alpha_{T,p}(f_{\uprightsf{phy}})$ and let $\alpha^*_{\uprightsf{biased}} (f_{\phy})= \sup_{T,p}(\alpha_{T,p}(f_{\uprightsf{phy}}))$.
\end{definition}

\subsection{\textsc{Build} algorithm}

\label{sec:build-algo}
\paragraph{\textsc{Build} algorithm for triplets.}
In this section we describe the \textsc{Build} algorithm of \citet{aho1981inferring} for the \textsc{Triplet Reconstruction} problem.\footnote{We should note that the algorithm proposed in \cite{aho1981inferring} works for a more general case of lower common ancestor constraints of the form $\lca(i,j)<\lca(k,l)$ ($\lca(i,j)$ is a descendant of $\lca(k,l)$).} .  In section \ref{sec:sharp-threshold} we use a gradual version of this algorithm ---we call it \textsc{Gradual Build}--- to determine the satisfiability threshold for \textsc{Triplet Reconstruction}.

As mentioned, the algorithm ``builds'' the tree in a top-down manner. In order to partition a given set $S$, the algorithm considers the induced set of constraints on that set and constructs the constraint graph. The set $S$ is partitioned in $S_1$ and $S_2$ so that there are no edges between the two sets (if such a partition is possible). The algorithm terminates either when all sets are singletons or a set cannot be partitioned without violating a triplet constraint. For the sake of completeness, we provide here a proof of correctness for \textsc{Build} (for pseudocode of the algorithm, see~Figure~\ref{fig:build-algo}).

\begin{figure}[ht]
\centering
\begin{tcolorbox}
\textbf{Build algorithm as used in triplet reconstruction (variables $S$, constraints $C$)}\\\\
Build($S$,$C$):
\begin{enumerate}

\item Remove all constraints involving variables outside $S$ to get set $C'$
\item If $S$ contains a single variable return a tree with a single node, that variable.
\item Construct constraint graph $G$ with parameters $S$, $C'$.
\item If $G$ is connected then instance is unsatisfiable, halt.
\item Otherwise, partition $S$ into $S_1,S_2\subseteq S$ such that there are no edges between them.
\item Solve the problem recursively for $(S_1,C')$ and $(S_2,C')$ to get trees $T_1,T_2$.
\item Combine $T_1$ and $T_2$ as the left and right subtree of a new tree $T$ and return $T$.
\end{enumerate}
\end{tcolorbox}
\caption{\textsc{Build} algorithm.}
\label{fig:build-algo}
\end{figure}

\begin{theorem}\label{thm:Build-Correct}(\cite{aho1981inferring})
    Let $\mathcal{I}=(V,C)$ be an instance of the triplet reconstruction problem. \text{Build}$(V,C)$ returns a tree satisfying all constraints if and only if such a tree exists.
\end{theorem}
\begin{proof}
    We first show that if the instance is not satisfiable then the algorithm declares the instance unsatisfiable and halts. For the sake of contradiction assume that this is not the case and \text{Build}$(V,C)$ returns a tree $T$. Let $(u,v,w)\in C$ be a constraint violated by $T$. Consider the point where vertices $u$, $v$ and $w$ were separated. We know that $(u,v,w)\in C$ so it must have been the case that there was an edge connecting $u$ and $v$ with label $w$, meaning that $u$ and $v$ must have been in the same part of the partition and $w$ in the other. This results in the constraint being satisfied which in turn leads to a contradiction. We now turn our attention to the case where the instance is satisfiable and assume, for the sake of contradiction, that \text{Build}$(V,C)$ fails to find a tree satisfying all triplet constraints. Let $T^*$ be a tree satisfying all the constraints and let $S_f$ be the set of variables that cause the algorithm to fail, meaning that when the algorithm was considering the corresponding constraint graph $G_f$, $G_f$ was connected. Consider the lowest common ancestor, $u$, of all the variables $S_f$ in $T^*$. Notice that the left and right subtrees of node $u$ induce a non-trivial partitioning of $S_f$; we view this partitioning as a cut $(S_1,S_2)$ for graph $G_f$. Recall that $G_f$ is connected, meaning that there is at least one edge crossing the cut. Let $(u,v,w)$ be the constraint that this edge corresponds to with $u\in S_1$, $v\in S_2$ and $w\in S_f$. Notice that $w$ must either be in $S_1$ or $S_2$; in both cases the constraint $(u,v,w)$ is violated, leading to a contradiction.
\end{proof}

\section{Strong Refutation}\label{sec:refutation}
In this section, we present our results on the refutation of Phylogenetic CSPs for both unordered and ordered trees (the latter of which includes the class of ordering CSPs). These results generalize the work of~\citet*{allen2015refute} on refuting ordinary CSPs (Boolean or with fixed alphabet size $q$). We begin by formally defining a certifying algorithm.

\begin{definition}\label{def:certify}
We say that an algorithm certifies that its input satisfies a condition $\calE$ if, given the input, the algorithm either succeeds or fails, and it succeeds only if the condition $\calE$ holds. In other words, the algorithm must fail on all inputs that do not satisfy $\calE$. However, it may also fail on some instances that do satisfy $\calE$.
\end{definition}

We are interested in polynomial-time certifying algorithms that succeed with high probability on random CSP instances. Our main theorem is as follows.

\begin{theorem}\label{thm:main-refute}
There exists a polynomial-time algorithm that certifies the optimal value of a Phylogenetic CSP instance $\calI_{f}$ with a $k$-ary payoff function $f$ is at most $\alpha^*(f) + \epsilon$,
where $\alpha^*(f)$ denotes the optimal biased random assignment value (see Definition~\ref{def:biased}).
The algorithm succeeds with high probability for random Phylogenetic CSP instances $\calI_{f}(n, m)$ with $m \geq C_{\ref{thm:main-refute}}(k, \epsilon)\, n^{k/2}\,\log^{3}n$ ($m \geq C_{\ref{thm:main-refute}}(k, \epsilon)\, n^{k/2}$ for even $k$), where $C_{\ref{thm:main-refute}}(k, \epsilon)$ is a function depending only on $k$ and $\epsilon$.
\end{theorem}

The subscript in the constant is useful later in the proof, to avoid confusion with other constants that may appear. We also establish a similar result for predicates that do not support $t$-wise independence (see below for details).

Our proof proceeds by reducing the problem of refuting Phylogenetic CSPs to certifying that a random instance does not admit a good coarse solution (see Section~\ref{sec:coarse-soln}). 
The task of finding such a coarse solution can be viewed as an ordinary constraint satisfaction problem over a large domain. However, unlike the CSPs
considered by \citet{allen2015refute}, these do not have \textit{negations} (i.e., variables cannot appear negated). 
As a result, we cannot directly apply the refutation scheme developed by \citet{allen2015refute} for ordinary CSPs to this setting. 

Note that random negation patterns introduce an additional source of randomness, which makes refutation easier. In fact, for CSPs of even arity $k$, 
one can refute an ordinary CSP with negations even if the list of argument tuples $C$ is chosen adversarially. That is, 
if an adversary selects a set of $m > Cn^{3/2}$ tuples $(X_1,\dots,X_k)$, and the negation pattern for each tuple is then chosen uniformly at random, 
the resulting CSP can still be refuted with high probability (see Theorem~\ref{thm:even-allen2015refute}). In our setting, such coarse CSP would be completely adversarial as we do not have negations.
Nevertheless, we show that the techniques of \citet{allen2015refute} can be adapted to the Phylogenetic CSP setting.
Specifically, we take advantage of the following key technical result from their work (see also~\cite{coja2007strong,barak2016noisy} for related results). 
We also slightly improve the lower bound on $m$ in the theorem below, shaving off polylogarithmic factors in the case of even $k$.

\begin{theorem}[Theorem~4.1 in ~\cite{allen2015refute}]\label{thm:AOW15}
For $k \geq 2$ and $p \geq n^{-k/2}$, let \{$w(T)\}_{T \in [n]^k}$ be independent random variables such that for each $T \in [n]^{k}$,
\begin{align*}
    \mathbb{E}[w(T)] &= 0, \\
    \Pr[w(T) \neq 0] &\leq p,\\
    |w(T)| &\leq 1. 
\end{align*}
Then there is an efficient algorithm certifying that
\begin{equation}
\label{eq:bound-AOW15}
    \sum_{T \in [n]^{k}} w(T) 
    \cdot
    x_{T_1}\cdots x_{T_k}
    \leq 2^{O(k)} \sqrt{p} n^{3k/4} \log^{3/2} n, 
\end{equation}
for all $x \in \mathbb{R}^n$ with $\|x\|_{\infty} \leq 1$ with high probability.
\end{theorem}

\begin{remark}
In Theorem~\ref{thm:even-allen2015refute}, we slightly improve the bound~\eqref{eq:bound-AOW15} for even $k$ to
$$
    \sum_{T \in [n]^{k}} w(T)  \cdot
    x_{T_1}\cdots x_{T_k}\leq 2^{O(k)} \sqrt{p} n^{3k/4}.
$$
\end{remark}

We first use this theorem to prove the lemma below for ordinary CSPs without negations, and then apply this lemma to coarse solutions of Phylogenetic CSPs. 


\begin{lemma}\label{lem:random-vs-random}
I. Let $f$ be a payoff function over an alphabet $\Sigma$, with $|\Sigma|= q$. Then, for an ordinary CSP defined by $f$, there exists a polynomial-time algorithm that certifies 
\begin{equation}\label{eq:lem:random-vs-random-1}
\max_{\xi: [n] \to \Sigma}    
\big( \val(\xi,\calI'_{f}) 
-   
\val(\xi,\calI''_{f}) \big) 
\leq \epsilon
\end{equation}
for given instances $\calI'_{f}$ and $\calI''_{f}$ of the CSP $\Gamma$. The algorithm succeeds with high probability for random instances $\calI'_f$ and $\calI''_f$ with $n$ variables and $m \ge C_{\ref{lem:random-vs-random}}(q,k,\epsilon) n^{k/2} \log^{3} n
$ ($m \geq C_{\ref{lem:random-vs-random}}(q,k, \epsilon)\, n^{k/2}$ for even $k$) constraints, where $C_{\ref{lem:random-vs-random}}(q,k,\epsilon) = q^{O(k)}\epsilon^{-2}$ is a function of $q$, $k$, and $\epsilon$.

II. Furthermore, there exists an algorithm certifying that for a given instance $\calI'_f$, we have
\begin{equation}\label{eq:lem:random-vs-random-2}
\max_{\xi: [n]\to \Sigma}    
\big( \val(\xi,\calI'_{f}) 
-   
\mathbb{E}_{\calI''_f} \big[\val(\xi,\calI''_{f}) \big] \big) 
\leq \epsilon.
\end{equation}
Above, the expectation is taken over a random choice of instance $\calI''_f$ with $n$ variables and $m$ constraints. The algorithm succeeds with high probability for random instances $\calI'_f$ with $n$ variables and $m$ constraints, provided that $m$ satisfies the bound from Part~I.
\end{lemma}
\begin{remark}\label{remark:lem:random-vs-random}
 Note that this lemma can be easily generalized to ordinary CSPs with multiple payoff functions, $f_1, \dots, f_s$. A random instance of such a CSP is specified by the number of variables $n$, the number of constraints $m$, and the probabilities $\lambda_1, \dots, \lambda_s$ associated with the payoff functions $f_1, \dots, f_s$. In such random instances, each constraint is created as follows: we first select a random hyperedge $(x_1, \dots, x_k)$, and then choose a payoff function $f_i$ with probability~$\lambda_i$. In particular, this lemma can be applied to ordinary CSPs with negations: For a given payoff function $f$, we can construct $2^k$ functions $f_i$ by applying all possible patterns of negations to the inputs and assigning each function a probability of $1/2^k$. Therefore, this result is a generalization of Theorem~2.3 in \cite{allen2015refute}.
\end{remark}
\begin{proof}[Proof of Lemma~\ref{lem:random-vs-random}]
We first prove these results for random instances $\calI'_f$ and $\calI''_f$ drawn from the random model $\calI_f(n, p)$, in which each constraint is included independently with probability $p = m / n^k$. Specifically, we show how to certify that
\begin{equation}\label{eq:max-total-val}
\max_{\xi:[n]\to \Sigma} \Big(
\operatorname{TotalVal}(\xi,\calI'_f) - 
\operatorname{TotalVal}(\xi,\calI''_f)
\Big)
\leq \epsilon m,
\end{equation}
where 
$\operatorname{TotalVal}(\xi,\calI'_f)$ and $\operatorname{TotalVal}(\xi,\calI''_f)$ denote the total (rather than average) payoffs of the solution $\xi$ on the instances 
$\calI'_f$ and $\calI''_f$, respectively. The instances $\calI'_f$ and $\calI''_f$ are drawn independently from the $\calI_f(n,p)$ model. To simplify the exposition, we allow constraints $(x_1,\dots,x_k)$ to include repeated variables. Since the fraction of such constraints is negligible, this does not affect the value of the random CSP. We assume that every payoff function evaluates to $1$ on any constraint with repeated variables. We also suppose that the parameter $p$ is known to the algorithm.
We then use a standard argument to show that the same result holds for the desired models $\calI_f(n, m)$ and $\calI_f(n, p)$ with unknown $p$.

We write a polynomial integer program to bound the left-hand side of~\eqref{eq:max-total-val}. For every variable $x_i$ and every possible value $a \in \Sigma$, we introduce a decision variable  
$y_{i,a} \in \{0,1\}$. Each assignment $\xi: [n] \to \Sigma$ corresponds to an integer solution $y$ defined as follows:  
\[
y_{i,a} =
\begin{cases} 
1, & \text{if } \xi(i) = a, \\ 
0, & \text{otherwise}.
\end{cases}
\]

We define the following polynomial objective function:
\begin{multline*}
P(y) = 
\sum_{(x_1,\dots,x_k) \in \calC'}
\sum_{(a_1,\dots,a_k) \in \Sigma^k}
f(a_1,\dots,a_k) 
\prod_{j=1}^{k} y_{x_j,a_j}
\\
-
\sum_{(x_1,\dots,x_k) \in \calC''}
\sum_{(a_1,\dots,a_k) \in \Sigma^k}
f(a_1,\dots,a_k)
\prod_{j=1}^{k} y_{x_j,a_j}.
\end{multline*}

Observe that the values of the solution $\xi$ on the CSP instances $\calI'_f$ and $\calI''_f$ (i.e., $\operatorname{TotalVal}(\xi,\calI')$ and $\operatorname{TotalVal}(\xi,\calI'')$) are equal to the first and second terms of $P(y)$, respectively, for $y$ defined as above. This equality holds because the product $\prod_{j=1}^{k} y_{x_j,a_j}$ equals $1$ if $\xi(x_j) = a_j$ for all $j$, and $0$ otherwise. Consequently,
\[
f(\xi(x_1),\dots,\xi(x_k)) =
\sum_{(a_1,\dots,a_k) \in \Sigma^k}
f(a_1,\dots,a_k) 
\prod_{j=1}^{k} y_{x_j,a_j}.
\]

We do not impose any additional constraints in the polynomial integer program (such as $\sum_a y_{i,a} = 1$), making it a relaxation of the exact optimization problem in the left-hand side of~\eqref{eq:max-total-val}. We further relax the integrality condition $y_{i,a} \in \{0,1\}$ to allow $y_{i,a} \in [-1,1]$. This yields a continuous relaxation of the original task of maximizing
$
\operatorname{TotalVal}(\xi,\calI'_f) 
-   
\operatorname{TotalVal}(\xi,\calI''_f)
$
over assignments $\xi : [n] \to \Sigma$.

We certify that $P(y) \leq \epsilon m$ for all $y_{i,a} \in [-1,1]$ using Theorem~\ref{thm:AOW15}, which in turn implies inequality~\eqref{eq:max-total-val}.  To this end, we define $|\Sigma|^k=q^k$ polynomials $P_{a_1,\dots,a_k}$ for all $(a_1,\dots,a_k) \in \Sigma^k$:
\[
P_{a_1,\dots,a_k}(y) = 
\sum_{(x_1,\dots,x_k) \in \calC'}
\prod_{j=1}^{k} y_{x_j,a_j}
-
\sum_{(x_1,\dots,x_k) \in \calC''}
\prod_{j=1}^{k} y_{x_j,a_j},
\]
where $\calC'$ and $\calC''$ are the sets of constraints for instances $\calI'$ and $\calI''$. Then,
$$P(y) = \sum_{(a_1,\dots,a_k)\in \Sigma^k} P_{a_1,\dots,a_k}(y)f(a_1,\dots,a_k).$$
The  coefficients of each polynomial $P_{a_1,\dots,a_k}$ are random variables determined by the choice of the constraints in $\calI'_f$ and $\calI''_f$. 

Let us verify the conditions of Theorem~\ref{thm:AOW15} for every polynomial $P_{a_1,\dots,a_k}$. All coefficients are equal to $1$, $-1$, or $0$ and are thus upper bounded by $1$ in absolute value. Moreover, the coefficient of the monomial $\prod y_{x_j,a_j}$ is $1$ if $(x_1,\dots,x_k) \in \calC' \setminus \calC''$, $-1$ if $(x_1,\dots,x_k) \in \calC'' \setminus \calC'$, and $0$ otherwise. Thus, each coefficient takes the values $-1$ and $1$ with equal probability. Finally, the probability that a given coefficient is nonzero is $2p(1 - p)$. 

We run the algorithm from Theorem~\ref{thm:AOW15} to certify that each polynomial  
$P_{a_1,\dots,a_k}(y)$ is upper bounded by  
\[
2^{O(k)} \sqrt{2p} \, N^{3k/4} \log^{3/2} N,
\]
where $N = |\Sigma| n$ is the number of variables $y$.

The algorithm succeeds with high probability for each polynomial and, therefore, by the union bound, for all polynomials simultaneously (the number of polynomials is $q^k$ where $q$ and $k$ are constants).
Thus, we certify that  
\[
P(y) \equiv \sum_{(a_1,\dots,a_k)\in \Sigma^k} P_{a_1,\dots,a_k}(y) f(a_1,\dots,a_k) \leq  
|\Sigma|^k 2^{O(k)} \sqrt{2p} \, N^{3k/4} \log^{3/2} (N).
\]
Using that $p = m/n^k$, write 
\[
|\Sigma|^k 2^{O(k)} \sqrt{2p} \, N^{3k/4} \log^{3/2} N
\leq 
|\Sigma|^{O(k)} m^{1/2} n^{k/4} \log^{3/2} N.
\]
The right-hand side is upper bounded by $\epsilon m$ for  
$
m > C_{\ref{lem:random-vs-random}}(q,k,\epsilon)n^{k/2} \log^{3} n
$, where $C_{\ref{lem:random-vs-random}}(q,k,\epsilon) = q^{O(k)}\epsilon^{-2}$ is some function of $q$, $k$, and $\epsilon$.

II. We now prove Part II. We present a simple algorithm that certifies~\eqref{eq:lem:random-vs-random-2} with high probability over the choice of $\calI'_f$ and the \emph{internal randomness of the algorithm}. In Appendix~\ref{sec:derandom-refutation}, we show how to obtain a deterministic algorithm by analyzing the proof of Theorem~\ref{thm:AOW15} from~\cite{allen2015refute}.

Let us first assume that $p = K/L$, where $K$ and $L$ are integers upper bounded by $\mathrm{poly}(n)$. We create $L$ random instances $\calI^{(1)}_f, \dots, \calI^{(L)}_f$ of the CSP as follows: for each tuple $(x_1, \dots, x_k) \in [n]^k$, we choose $K$ random instances among the $L$ instances and add the payoff function $f$ on $(x_1, \dots, x_k)$ to each of them. Then, each instance $\calI^{(i)}_f$ is distributed according to $\calI_f(n, p)$.
Note that the instances $\calI^{(1)}, \dots, \calI^{(L)}$ are not independent. In fact, each tuple $(x_1, \dots, x_k) \in [n]^k$ appears in exactly $K$ of them. Thus, for every $\xi:[n]\to \Sigma$, we have
\begin{equation}\label{eq:expect-equals-empirical-avg}
\mathbb{E}_{\calI''_f} \big[\val(\xi,\calI''_{f}) \big] = 
\frac{1}{L}
\sum_{i=1}^L
\val(\xi,\calI^{(i)}).
\end{equation}
Now, for each $i$, we certify that  
\begin{equation}\label{eq:lem:random-vs-random-3}
\max_{\xi: [n] \to \Sigma}    
\left( \val(\xi,\calI'_f) 
-   
\val(\xi,\calI^{(i)}_f) \right) 
\leq \epsilon/2
\end{equation}  
using Part~I. For each instance $\calI^{(i)}_f$, the certification succeeds with high probability. Thus, the expected number of failures is $o(1) \cdot L$, and by Markov's inequality, we successfully certify~\eqref{eq:lem:random-vs-random-3} for a $(1 - \epsilon/2)$ fraction of the $L$ instances with high probability. In this case, we are guaranteed that the desired bound~\eqref{eq:lem:random-vs-random-2} holds, since by \eqref{eq:expect-equals-empirical-avg} we have
\begin{multline*}
\max_{\xi: [n]\to \Sigma}    
\big( \val(\xi,\calI'_{f}) 
-   
\mathbb{E}_{\calI''_f} \big[\val(\xi,\calI''_{f}) \big] \big) =
\max_{\xi: [n] \to \Sigma}    
\left( \val(\xi,\calI'_f) 
-   
\frac{1}{L}\sum_{i=1}^L\val(\xi,\calI^{(i)}_f) \right) 
\leq\\ \leq
\frac{1}{L}\sum_{i=1}^L
\max_{\xi: [n] \to \Sigma}    
\underbrace{
\left( \val(\xi,\calI'_f) 
-   
\val(\xi,\calI^{(i)}_f) \right)}_{\Delta_i} 
\leq \epsilon/2 + \epsilon/2 = \epsilon.
\end{multline*}
Here we used that for at least a $(1-\epsilon/2)$ fraction of indices $i$, we have $\Delta_i \le \epsilon/2$,
and for the remaining at most $\epsilon/2$ fraction of indices,
we trivially have $\Delta_i \le 1$, since the value of any CSP is at most $1$.

III. We now sketch an argument showing how to obtain the result for the $\calI_f(n,m)$ model from the one for the $\calI_f(n,p)$ model. Choose a small parameter $\rho > 0$ and set $p = (1 - \rho)m / n^k$. Then, a random instance drawn from the $\calI_f(n,p)$ model contains between $(1 - 2\rho)m$ and $m$ constraints with high probability. Moreover, the distributions $\calI_f(n,m)$ and $\calI_f(n,p)$ can be coupled so that, with high probability, the instance $\calI_f(n,m)$ consists of all constraints of $\calI_f(n,p)$, followed by at most $2\rho m$ additional constraints.
Thus, it suffices to refute the instance $\calI_f(n,p)$. A minor complication is that the algorithm does not know in advance how many constraints are in $\calI_f(n,p)$. However, it can try all possible values $m' \in [(1 - 2\rho)m,\, m]$. The same reasoning applies to the $\calI_f(n,p)$ model when $p$ is unknown.

If $p = (1 - \rho)m / n^k$ as above, and $\rho$ is rational, then $p$ is a rational number with both numerator and denominator polynomially bounded in $n$. (Otherwise, we may approximate $p$ by a rational number $K / L$ on a sufficiently fine grid.)

\end{proof}

\begin{proof}[Proof of Theorem~\ref{thm:main-refute}]
By Lemma~\ref{lem:cs-bad}, 
if $m/n$ exceeds a constant that depends on $\epsilon$, then there exists $q$ such that the following inequality holds for random instances $\calI_{f_{\uprightsf{phy}}}(n,m)$:
\[
\max_{\xi \in \Xi_{\frac{16}{q}, q}} \mathbb{E}\left[   
\val^+(\xi,\calI_{f_{\uprightsf{phy}}}(n,m)) \right]
\leq \alpha^*(f_{\uprightsf{phy}}) + \epsilon.
\]
Fix such a value of $q$, and certify that
\[
\max_{\xi: [n] \to \Sigma}    
\left( \val^+(\xi,\calI_{f}) 
-   
\mathbb{E}_{\calI_{f_{\uprightsf{phy}}}(n,m)} \left[\val^+(\xi,\calI_{f_{\uprightsf{phy}}}(n,m)) \right] \right)
\leq \epsilon
\]
using Lemma~\ref{lem:random-vs-random} (Part~II), for every coarse tree with $q$ leaves. Here, $\Sigma$ denotes the set of leaves of the coarse tree, and $f = f^+_{\uprightsf{phy}}$ is the payoff function associated with that coarse tree. With high probability, all certifications succeed, since the number of ordered full binary trees with $q$ leaves is given by the $(q-1)$st Catalan number and is upper bounded by $4^{q}$; in particular, this bound does not depend on $n$.
As a result, we certify that
\[
\max_{\xi \in \Xi_{\frac{16}{q}, q}}
\val^+(\xi,\calI_{f}) \leq \alpha^*(f_{\uprightsf{phy}}) + \epsilon.
\]
This yields the desired result because, by Lemma~\ref{lem:cs-good}, the optimal value of every Phylogenetic CSP instance $\calI_{f_{\uprightsf{phy}}}$ is upper bounded by $\val^+(\xi,\calI_{f_{\uprightsf{phy}}})$ for some coarse solution $\xi \in \Xi_{\frac{16}{q}, q}$.
\end{proof}

\subsection{Certifying near \texorpdfstring{$t$}{t}-wise independence}\label{subsec:t-wise}


We now show how to refute  ordered and unordered Phylogenetic CSPs that do not support $t$-wise independent distributions, obtaining results similar to those of~\cite*{allen2015refute} for ordinary CSPs with predicates that do not support $t$-wise uniformity. We provide the most popular examples of such Phylogenetic and Ordering CSPs, together with improved refutations, in Appendix~\ref{sec:bounds-uniform-distr}. As a byproduct, we also obtain a similar result for ordinary CSPs without negations.

The term $t$-wise independence is used in slightly different ways in the literature. Thus, we adopt the term \emph{$t$-wise $\mu$-independent distribution}. We say that a distribution $\mathcal{D}$ over $\Omega^k$ of $(X_1,\dots,X_k)$ is $t$-wise $\mu$-independent ($t\geq 2$) if  
(1) each $X_i$ has marginal distribution $\mu$, and  
(2) every subset of $t$ variables $X_{i_1}, \dots, X_{i_t}$ is mutually independent.

We define $\beta^*_t(f)$ for ordinary, ordering, and Phylogenetic CSPs as the maximum value achievable by a $t$-wise $\mu$-independent distribution. The formal definition is as follows.
For ordinary CSPs with a payoff function $f$ of arity $k$ and domain $\Omega$:
\begin{equation}\label{eq:def:beta:ordinary}
    \beta^*_t(f) =
    \sup_{\calD} \E_{(a_1, \dots, a_k) \sim \calD}[f(a_1, \dots, a_k)],
\end{equation}
where the supremum is taken over all $t$-wise $\mu$-independent distributions $\calD$ over $\Omega^k$, for some distribution $\mu$ on $\Omega$.
For Phylogenetic CSPs on ordered and unordered trees:
\begin{equation}\label{eq:def:beta:phy}
    \beta^*_t(f) =
    \sup_{T,\,\calD} \E_{(a_1, \dots, a_k) \sim \calD}[f^-(a_1, \dots, a_k)],
\end{equation}
where $T$ is a finite or infinite tree, and $\calD$ is a $t$-wise $\mu$-independent distribution on the leaves of $T$. The function $f^-$ is defined by $f^-(x_1,\dots,x_k) = f(x_1,\dots,x_k)$ if all $x_i$ are distinct, and $0$ otherwise. In this definition, we may also assume that $\mu$ is the uniform distribution. This can be achieved by first splitting each leaf into multiple leaves so that the probability mass of each leaf becomes approximately equal (note that splitting leaves can only increase the expected value of $f^-(a_1,\dots,a_k)$), and then approximating the resulting distribution by the uniform distribution over the new set of leaves.

For Ordering CSPs (a subclass of ordered Phylogenetic CSPs), this definition can be simplified: $\beta^*_t(f)$ is given by~(\ref{eq:def:beta:ordinary}), where the domain is $\Omega = [0,1]$ and $\calD$ must be a $t$-wise uniform-independent distribution.
Note that $\alpha^*(f) = \beta^*_k(f)$.

We prove the following theorem, which is a generalization of Theorem~\ref{thm:main-refute}.

\begin{theorem}\label{thm:t-wise-phy-CSP-refute}
There exists a polynomial-time algorithm that certifies the optimal value of a Phylogenetic CSP instance $\calI_{f_{\uprightsf{phy}}}$ with a $k$-ary payoff function $f_{\uprightsf{phy}}$ is at most $\beta^*_t(f_{\uprightsf{phy}}) + \epsilon$, for any given $t \geq 2$.
The algorithm succeeds with high probability for random Phylogenetic CSP instances $\calI_{f_{\uprightsf{phy}}}(n, m)$ with
$
m \geq C_{\ref{thm:t-wise-phy-CSP-refute}}(k, \epsilon)\, n^{t/2} \log^{3} n
$
($
m \geq C_{\ref{thm:t-wise-phy-CSP-refute}}(k, \epsilon)\, n^{t/2},
$ for even $t$),
where $C_{\ref{thm:t-wise-phy-CSP-refute}}(k, \epsilon)$ is a function depending only on~$k$ and~$\epsilon$.
\end{theorem}

We begin by proving several results for ordinary CSPs.
Consider a finite domain $\Omega$ and an assignment $\xi : [n] \to \Omega$. This assignment naturally induces a distribution $\mu_{\xi}$ on $\Omega$, where $\mu_{\xi}(a) = |\xi^{-1}(a)|/n$ is the fraction of elements in $[n]$ that are mapped to $a$. 

\begin{lemma}\label{lem:eps-close-assignment-distr}
There exists a polynomial-time algorithm that certifies the following property: given a set
of $k$-tuples $C\subseteq [n]^k$ and a finite domain $\Omega$ of size $q$, the following holds. For every assignment 
$\xi : [n] \to \Omega$, there exists a $t$-wise $\mu_{\xi}$-independent
distribution $\calD$ on $\Omega^k$ such that the distribution of the random variables $\xi(X_1), \dots, \xi(X_k)$ 
is $\epsilon$-close to $\calD$ in total variation distance, where $(X_1, \dots, X_k)$ is drawn uniformly at random from $C$. The algorithm succeeds with high probability if $C$ consists of $m$ random $k$-tuples sampled from $[n]^k$, for $m \geq C_{\ref{lem:eps-close-assignment-distr}}(q,k,\epsilon)\,n^{t/2}\log^{3}n$ ($m \geq C_{\ref{lem:eps-close-assignment-distr}}(k, \epsilon)\, n^{t/2}$ for even $t$), where $C_{\ref{lem:eps-close-assignment-distr}}(q,k,\epsilon)$ is some function of $q$, $k$,and $\epsilon$.
\end{lemma}
\begin{proof}
We certify that for every assignment $\xi : [n] \to \Omega$, the following holds: for every subset of $t$ distinct indices $i_1, \dots, i_t$ and every choice of labels $a_{i_1}, \dots, a_{i_t}$, we have
\begin{equation}
\label{eq:eps-mu-independ}
\left|
\Pr(\xi(X_{i_1}) = a_{i_1}, \dots, \xi(X_{i_t}) = a_{i_t}) 
- \mu_{\xi}(a_{i_1}) \cdots \mu_{\xi}(a_{i_t})
\right| \leq \epsilon'
\end{equation}
for a sufficiently small $\epsilon' > 0$ (to be specified later). We do so by running the certifying algorithm from Lemma~\ref{lem:random-vs-random} with the payoff function  
$f(x_1,\dots,x_k) = \mathbf{1}\{x_{i_1} = a_{i_1}, \dots, x_{i_t} = a_{i_t}\}$  
for every choice of distinct indices $(i_1, \dots, i_t)$ and labels $(a_1, \dots, a_t)$. Note that this payoff function depends only on the $t$ variables $x_{i_1}, \dots, x_{i_t}$, and therefore the algorithm can be applied as long as $m \ge C_{\ref{lem:random-vs-random}}(q,k,\epsilon) n^{t/2} \log^{3} n$ ($m \geq C_{\ref{lem:random-vs-random}}(k, \epsilon)\, n^{t/2}$ for even $t$) . Each run of the algorithm succeeds with high probability, so by the union bound over all such choices of indices and labels, all runs succeed with high probability.

Observe that if the random variables $\xi(X_1), \dots, \xi(X_k)$ satisfied condition~(\ref{eq:eps-mu-independ}) with $\epsilon' = 0$, then they would be $t$-wise $\mu_{\xi}$-independent. However, in our case, $\epsilon' > 0$. We claim that, nevertheless, condition~(\ref{eq:eps-mu-independ}) implies that the distribution of $(\xi(X_1), \dots, \xi(X_k))$ is $\epsilon$-close (in total variation distance) to some $t$-wise $\mu_{\xi}$-independent distribution.
For the uniform distribution $\mu$ over the binary domain, this claim follows from the work of~\cite*{alon2003almost} (see also~\cite{alon2007testing, goldreich2011three}).
For general (non-uniform) distributions $\mu$ over larger domains, we prove Theorem~\ref{thm:t-wise-mu-approx} in Appendix~\ref{app:lem:t-wise-mu-approx}. 

To apply Theorem~\ref{thm:t-wise-mu-approx}, we verify that its conditions hold. Let $f : \Omega^k \to \mathbb{C}$ be an arbitrary function depending on at most $t$ variables, say $a_{i_1}, \dots, a_{i_t}$. Then
\begin{multline*}
\Bigl|
\sum_{(a_1, \dots, a_k) \in \Omega^k}
\Big(
\Pr(\xi(X_1) = a_1, \dots, \xi(X_k) = a_k)
- \mu_{\xi}(a_1) \cdots \mu_{\xi}(a_k)
\Big)
f(a_1, \dots, a_k)
\Bigr|
=
\\
=
\Bigl|
\sum_{(a_{i_1}, \dots, a_{i_t}) \in \Omega^t}
\underbrace{\Big(
\Pr(\xi(X_{i_1}) = a_{i_1}, \dots, \xi(X_{i_t}) = a_{i_t})
- \mu_{\xi}(a_{i_1}) \cdots \mu_{\xi}(a_{i_t})
\Big)
f(a_1, \dots, a_k)}_{(*)}
\Bigr|.
\end{multline*}
Each expression $(*)$ is bounded in absolute value by $\epsilon'\|f\|_{\infty}$, and there are  $q^t$ such terms. Thus, the entire expression is upper bounded by $q^t \epsilon' \|f\|_\infty$. 

Theorem~\ref{thm:t-wise-mu-approx} guarantees the existence of a $t$-wise $\mu_{\xi}$-independent distribution $\calD$ that is $\epsilon$-close in total variation distance to the distribution of $(\xi(X_1), \dots, \xi(X_k))$, where $\epsilon = O\big((2q)^k (q^t \epsilon')^{1/(k+1)}\big)$. Choosing $\epsilon'$ sufficiently small completes the proof.
\end{proof}

We now prove Theorem~\ref{thm:t-wise-phy-CSP-refute}.

\begin{proof}[Proof of Theorem~\ref{thm:t-wise-phy-CSP-refute}]
As in the proof of Theorem~\ref{thm:main-refute}, let $q$ be a sufficiently large constant to be specified later. Let $(X_1,\dots,X_k)$ be a random constraint in $\calC$. Using the algorithm from Lemma~\ref{lem:eps-close-assignment-distr}, we certify that for every assignment $\xi : [n] \to [q]$, there exists a $t$-wise distribution $\calD$ on $[q]^k$ that is $\epsilon$-close (in total variation distance) to the distribution of $(\xi(X_1),\dots,\xi(X_k))$.
We show that if the certification algorithm succeeds (which happens with high probability), then the optimal value of the Phylogenetic CSP instance $\calI_{f_{\uprightsf{phy}}}$ is at most $\beta^*_t(f_{\uprightsf{phy}}) + 2\epsilon$.

The value of any solution to $\calI_{f_{\uprightsf{phy}}}$ is upper bounded by the value of a coarse solution, i.e., $\val^+(\xi, \calI_{f_{\uprightsf{phy}}})$ for some $\xi \in \Xi_{\frac{16}{q}, q}$. Let $\xi^* \in \Xi_{\frac{16}{q}, q}$ be the assignment that maximizes this value. The assignment $\xi^*$ maps variables in $[n]$ to the leaves of some tree, where each leaf is assigned a color, and at most a $16/q$ fraction of variables are mapped to any single color class. For convenience, we associate each leaf with an element of $[q]$ and treat $\xi^*$ as a function from $[n]$ to $[q]$.
Then, by Lemma~\ref{lem:eps-close-assignment-distr}, there exists a $t$-wise $\mu_{\xi^*}$-independent distribution $\calD$ on $[q]^k$ that is $\epsilon$-close to the distribution of $(\xi^*(X_1), \dots, \xi^*(X_k))$. Therefore,
\[
\val^+(\xi^*, \calI_{f_{\uprightsf{phy}}})
=
\E_{(X_1, \dots, X_k) \in \calC}[f^+(\xi^*(X_1), \dots, \xi^*(X_k))]
\leq \E_{(a_1, \dots, a_k) \sim \calD}[f^+(a_1, \dots, a_k)] + \epsilon.
\]
Also, by the definition of $\beta^*_t(f_{\uprightsf{phy}})$,
\[
\E_{(a_1, \dots, a_k) \sim \calD}[f^-(a_1, \dots, a_k)] \leq \beta^*_t(f_{\uprightsf{phy}}).
\]
It remains to bound the difference:
\[
\E_{(a_1, \dots, a_k) \sim \calD}[f^+(a_1, \dots, a_k) - f^-(a_1, \dots, a_k)] \leq \epsilon.
\]

Note that $f^+(a_1,\dots,a_k)$ and $f^-(a_1,\dots,a_k)$ differ only when $a_i$ and $a_j$ share the same color for some distinct $i,j \in \{1,\dots,k\}$. Since $\xi^*$ maps at most a $16/q$ fraction of variables to any single color, the marginal distribution $\mu_{\xi^*}$ assigns at most $16/q$ mass to each color class. Hence, the probability that $\clr(a_i) = \clr(a_j)$ for some $i \neq j$ is at most $16/q \cdot \binom{k}{2} < 8k^2 / q$. Choosing $q = 8k^2 / \epsilon$ ensures that this quantity is at most $\epsilon$, completing the proof.
\end{proof}

We conclude this section by stating an analogous result for ordinary CSPs without negations. We omit the proof, as it is essentially identical to that of Theorem~\ref{thm:t-wise-phy-CSP-refute}, except that it does not require handling the functions $f^+$, $f^-$, or different color classes.

\begin{tcolorbox}
\begin{theorem}\label{thm:ordinary-CSP-t-wise-refute}
There exists a polynomial-time algorithm that certifies the optimal value of an ordinary CSP instance $\calI_{f}$ with a $k$-ary payoff function $f$ is at most $\beta^*_t(f) + \epsilon$, for any given $t \geq 2$.
The algorithm succeeds with high probability for random CSP instances $\calI_{f}(n, m)$ with
$
m \geq C_{\ref{thm:ordinary-CSP-t-wise-refute}}(k, \epsilon)\, n^{t/2} \log^{3} n
$ ($
m \geq C_{\ref{thm:ordinary-CSP-t-wise-refute}}(k, \epsilon)\, n^{t/2}
$ for even $t$),
where $C_{\ref{thm:ordinary-CSP-t-wise-refute}}(k, \epsilon)$ is a function depending only on~$k$ and~$\epsilon$.
\end{theorem}
\end{tcolorbox}

\begin{remark} This theorem is similar to Theorem~2.4 in the work of~\cite{allen2015refute}. However, it applies to CSPs without negations. Examples of such CSPs include predicates on $k$ variables $x_1, \dots, x_k \in \{0,1\}$ that evaluate to $1$ if exactly $s$ of the variables are equal to $1$, and the rainbow $q$-coloring predicate on $k$ variables $x_1, \dots, x_k \in [q]^k$, which evaluates to $1$ if $x_1, \dots, x_k$ include every label (or color) from $[q]$. This theorem can also be extended to the setting with multiple payoff functions $f_1, \dots, f_s$ (see Remark~\ref{remark:lem:random-vs-random}). This variant generalizes the result of~\cite{allen2015refute}: For a given payoff function $f$, we can construct $2^k$ functions $f_i$ by applying all possible patterns of negations to the inputs and assigning each function a probability of $1/2^k$.
\end{remark}

\subsection{Derandomization of the Refutation Algorithm}\label{sec:derandom-refutation}

In this section we provide a deterministic algorithm that for every $a\in \Sigma^k$ certifies that for every $y\in [-1,1]^{N}$:
\begin{align*}
    \mathbb{E}_{\calI_f''}[P_a(y)]\leq 2^{O(k)}\sqrt{2p}N^{3k/4}\log^{3/2}(N).
\end{align*}
It follows from our analysis that this implies a deterministic algorithm for refutation. We will prove the claim in the more challenging case where $k$ is odd. Our analysis follows the proof Theorem 4.1 in \cite{allen2015refute}. Notice that the polynomial $P_a$ is over $N=n\cdot q$ variables indexed by $(i,j)$ for $i\in [n]$ (the set of variables) and $j\in [q]$ (the elements of the alphabet $\Sigma$), for notational convenience for $T\in [N]^k$, we will use $x(T)$ to refer to the tuple of variables induced by tuple $T$ and $a(T)$ to refer to the set of elements in $\Sigma$ induced by tuple $T$. With this notation, we can rewrite the polynomial $P_a(y)$ as:
\begin{align*}
    P_a(y)=\sum_{T\in [N]^k}\ind{a(T)=a}\left(\ind{x(T)\in \calI_f'}-\ind{x(T)\in \calI_f''}\right)\prod_{j\in T}y_j
\end{align*}
Following the notation in \cite{allen2015refute}, we let
\[
    w(T) = \ind{a(T)=a}\left(\ind{x(T)\in \calI_f'}-\ind{x(T)\in \calI_f''}\right),
\]
so we can simply write
\[
    P_a(y) = \sum_{T\in[N]^k} w(T) \prod_{j\in T} y_j.
\]
We let
\[
    W_i = \left| \sum_{T\in [N]^{k-1}} w(T,i) \prod_{j\in T} y_j \right|.
\]
We also define the $N^{k-1}\times N^{k-1}$ matrix $A$, indexed by $[N]^{k-1}$. For $i_1,i_2,j_1,j_2\in [N]^{\frac{k-1}{2}}$:
\begin{align}
\label{eq:definition_of_A}
    A_{(i_1,i_2),(j_1,j_2)}=\begin{cases}
        \sum_{l\in[N]}w(i_1,j_1,l)w(i_2,j_2,l)&\text{ if } (i_1,j_1)\neq (i_2,j_2)\\
        0 &\text{ otherwise}
    \end{cases}
\end{align}
We prove the following bound on $\mathbb{E}_{\calI_f''}[P_a(y)]$.
\begin{lemma}
\label{lem:bound_on_Exp_Pa}
For every $y\in [-1,1]^N$
    \begin{align*}
        \mathbb{E}_{\calI_f''}[P_a(y)]\leq \sqrt{N}\sqrt{N^{k-1}\|\mathbb{E}_{\calI_{f}''}[A]\|+\mathbb{E}_{\calI_f''}\left[\sum_{T\in [N]^k}w(T)^2\right]}
    \end{align*}
\end{lemma}
\begin{proof}
    First observe that $P_a(y)\leq \sum_{i=1}^{N}W_i$. By Cauchy-Schwarz inequality, we have that:
    \begin{align*}
        \mathbb{E}_{\calI_f''}[P_a(y)]&\leq \mathbb{E}_{\calI_f''}\left[\sum_{i=1}^{N}W_i\right]\\
        &\leq \mathbb{E}_{\calI_f''}\left[\sqrt{N}\sqrt{\sum_{i=1}^{N}W_i^2}\right]\\
        &=\sqrt{N}\cdot \mathbb{E}_{\calI_f''}\left[\sqrt{\sum_{i}^{N}W_i^2}\right]
    \end{align*}
    By Jensen's inequality, we get that:
    \begin{align*}
        \mathbb{E}_{\calI_f''}\left[\sqrt{\sum_{i=1}^{N}W_i^2}\right]\leq \sqrt{\mathbb{E}_{\calI_f''}\left[\sum_{i=1}^{N}W_i^2\right]}
    \end{align*}
    We now bound $\mathbb{E}_{\calI_f''}\left[\sum_{i=1}^{N}W_i^2\right]$. Observe that for $A$ defined as in \ref{eq:definition_of_A}, we have that:
    \begin{align*}
        \sum_{i=1}^{N}W_i^2&=(y^{\otimes k-1})^{\top}Ay^{\otimes k-1}+\sum_{T,U\in [N]^{\frac{k-1}{2}}}\prod_{l\in T}y_l^2\prod_{j\in U}y_j^2\sum_{i=1}^{N}w(T,U,i)^2\\
        &\leq (y^{\otimes k-1})^{\top}Ay^{\otimes k-1}+\sum_{T\in [N]^k}w(T)^2
    \end{align*}
    Applying expectations, we have that:
    \begin{align*}
        \mathbb{E}_{\calI_f''}\left[\sum_{i=1}^{N}W_i^2\right]&\leq \mathbb{E}_{\calI_f''}\left[(y^{\otimes k-1})^{\top}Ay^{\otimes k-1}+\sum_{T\in [N]^k}w(T)^2\right]\\
        &=(y^{\otimes k-1})^{\top}\mathbb{E}_{\calI_f''}\left[A\right]y^{\otimes k-1}+\mathbb{E}_{\calI_f''}\left[\sum_{T\in[N]^k}w(T)^2\right]\\
        &\leq N^{k-1}\|\mathbb{E}_{\calI_f''}\left[A\right]\|+\mathbb{E}_{\calI_f''}\left[\sum_{T\in[N]^k}w(T)^2\right]
    \end{align*}
    Where from the second to the third line we use that $y\in[-1,1]^{N}$, meaning that $\|y^{\otimes k-1}\|\leq N^{\frac{k-1}{2}}$. The claim of the lemma follows.
\end{proof}
Observe that the right-hand side of the bound in Lemma \ref{lem:bound_on_Exp_Pa} can be efficiently computed by the refutation algorithm to certify that it is indeed less than or equal to $2^{O(k)}\sqrt{2p}N^{3k/4}\log^{3/2}(N)$. We now need to show that for random instance $\calI_f'\sim \calI_f(n,p)$, with high probability over the randomness of $\calI_f'$, we have that:
\begin{align*}
    \sqrt{N}\sqrt{N^{k-1}\|\mathbb{E}_{\calI_{f}''}[A]
    \|+\mathbb{E}_{\calI_f''}\left[\sum_{T\in [N]^k}w(T)^2\right]}\leq 2^{O(k)}\sqrt{2p}N^{3k/4}\log^{3/2}(N)
\end{align*}
To that end, it suffices to show the following lemma.
\begin{lemma}
    Let $\calI_{f}'\sim \calI(n,p)$, then with high probability, over the randomness of $\calI_f'$, we have that:\\
    I. \begin{align*}
        \|\mathbb{E}_{\calI_f''}[A]\|\leq 2^{O(k)}pN^{k/2}\log^3(N)
    \end{align*}
    II. \begin{align*}
        \mathbb{E}_{\calI_f''}\left[\sum_{T\in [N]^k}w(T)^2\right]\leq 3pn^{k} 
    \end{align*}
\end{lemma}
\begin{proof}
    We first prove (I). For convenience, we let $B=2^{O(k)}pN^{k/2}\log^3(N)$. We use the following lemma from \cite{allen2015refute}.
    \begin{lemma}[Lemma A.12 in \cite{allen2015refute}] 
    \label{lem:trace_bound}
    Let $\{w(T)\}_{T\in [N]^k}$ be independent random variables such that for every $T\in [N]^k$:
    \begin{align*}
    \mathbb{E}[w(T)] &= 0, \\
    \Pr[w(T) \neq 0] &\leq p,\\
    |w(T)| &\leq 1. 
\end{align*}
    Then, for $A$ defined as in equation \ref{eq:definition_of_A}, we have that:
    \begin{align*}
        \E[\tr((AA^{\top})^r)]\leq N^{O(k)}2^{O(r)}r^{6r}p^{2r}N^{kr}
    \end{align*}
\end{lemma}
Notice that the conditions of Lemma~\ref{lem:trace_bound} hold in our case, so we get the bound on $\E[\tr((AA^{\top})^r)]$. By applying Jensen's inequality twice, we get that:
\begin{align*}    \Prob{\|\mathbb{E}_{\calI_{f}''}[A]\|\geq B}&\leq \Prob{\mathbb{E}_{\calI_{f}''}[\|A\|]\geq B}\\
    &= \Prob{\left(\mathbb{E}_{\calI_{f}''}[\|A\|]\right)^{2r}\geq B^{2r}}\\
    &\leq \Prob{\mathbb{E}_{\calI_{f}''}[\|A\|^{2r}]\geq B^{2r}}
\end{align*}
We now use that $\|A\|^{2r}\leq \tr\left((AA^{\top})^{r}\right)$, to get that:
\begin{align*}
    \Prob{\mathbb{E}_{\calI_{f}''}[\|A\|^{2r}]\geq B^{2r}}\leq \Prob{\mathbb{E}_{\calI_f''}\left[\tr\left((AA^{\top})^r\right)\right]\geq B^{2r}}
\end{align*}
Applying Markov's inequality, we get that:
\begin{align*}
    \Prob{\|\mathbb{E}_{\calI_f''}[A]\|\geq B}&\leq\frac{\E\left[\mathbb{E}_{\calI_f''}\left[\tr\left((AA^{\top})^{r}\right)\right]\right]}{B^{2r}}\\
    &=\frac{\mathbb{E}\left[\tr\left((AA^{\top})^{r}\right)\right]}{B^{2r}}.
\end{align*}
Taking $r=\Theta(\log(N))$ yields (I). 
We now focus on (II). We have that:
\begin{align*}
    \E_{\calI_f''}\left[\sum_{T\in [N]^{k}} w(T)^2\right]&=\sum_{T\in [N]^{k}}\E_{\calI_f''}\left[w(T)^2\right]\\
    &=\sum_{\substack{T\in [N]^k\\a(T)=a}}\ind{x(T)\in \calI_{f}'}+p-2p\ind{x(T)\in \calI_f'}\\
    &\leq \sum_{\substack{T\in[N]^k\\a(T)=a}}\ind{x(T)\in \calI_f'}+p\\
    &=pn^{k}+\sum_{\substack{T\in[N]^k\\a(T)=a}}\ind{T\in \calI_f'}
\end{align*}
We now use the Chernoff bound for the sum $\sum_{\substack{T\in[N]^k\\a(T)=a}}\ind{T\in \calI_f'}$, this is a sum of $n^k$ independent Bernoulli random variables with parameter $p$. By Theorem 4.4 in \cite{mitzenmacher2017probability} we get that:
\begin{align*}
    \Prob{\sum_{\substack{T\in [N]^k\\a(T)=a}}\ind{T\in \calI_f'}\geq 2pn^k}\leq \exp\left(-\frac{4}{3}pn^k\right)
\end{align*}
Part (II) of the lemma follows. 
\end{proof}
\subsection{Certifying Polynomials for Even \texorpdfstring{$k$}{k}}\label{sec:polynomial-even-k}
For the case where $k$ is even, we prove the following strengthening of Theorem 4.1 in \citet{allen2015refute}.
\begin{theorem}\label{thm:even-allen2015refute}
Let $\{T^{(1)}, T^{(2)},\ldots, T^{(m)}\}\subseteq [n]^k$ (can be adversarially selected) \\
and let $\{w(T^{(i)})\}_{i\in \{1,2,\ldots, m\}}$ be independent random variables such that for each $i\in \{1,2,\ldots, m\}$:
\begin{align*}
        \E[w(T^{(i)})]&=0,\\
        |w(T^{(i)})|&\leq 1.
    \end{align*}
    Then there is an efficient algorithm certifying that
    \begin{align*}
        \sum_{i=1}^{m}w(T^{(i)})\prod_{j\in T}x_j\leq 4K_G\sqrt{\frac{m}{n^k}}n^{\frac{3k}{4}}
    \end{align*}
    for all $x\in \mathbb{R}^n$ with $\|x\|_{\infty}\leq 1$ with high probability, where $K_G$ is the Grothendieck's constant. 
\end{theorem}
\begin{proof}
    For convenience, let $P(x)=\sum_{i=1}^{m}w(T^{(i)})\prod_{j\in T}x_j$.
    We will relax the problem by introducing for every $T\in [n]^{\frac{k}{2}}$ a vector $U_{T}$, which in the intended solution is $U_{T}=e\cdot \prod_{i\in T}x_i $ where $e$ is a unit vector. We have that:
    \begin{align*}
P(x_1,\ldots,x_n)\leq \max_{\substack{U_T,\\ \|U_T\|\leq 1}}\sum_{i=1}^{m}w(T^{(i)})\langle U_{T^{(i)}_1}, U_{T^{(i)}_2}\rangle 
    \end{align*}
    The algorithm solves this semidefinite program that upper bounds $P(x_1,\ldots,x_n)$ and certifies that it is indeed less that or equal to $ 4K_G\sqrt{\frac{m}{n^k}}n^{\frac{3k}{4}}$. It remains to show that with high probability over the randomness in $\{w(T^{(i)})\}_{i\in \{1,2,\ldots, m\}}$ the semidefinite program is bounded by $ 4K_G\sqrt{\frac{m}{n^k}}n^{\frac{3k}{4}}$. To that end, we use Grothendieck's inequality (see \cite{grothendieck1956resume,braverman2013grothendieck,guedon2016community}) to get that:
    \begin{align*}
        \max_{\substack{U_T,\\ \|U_T\|\leq 1}}\sum_{i=1}^{m}w(T^{(i)})\langle U_{T^{(i)}_1}, U_{T^{(i)}_2}\rangle \leq K_G\cdot \max_{Y,Z\in \{-1,1\}^{n^{k/2}}}\sum_{i=1}^{m}w(T^{(i)})Y_{T^{(i)}_1}Z_{T^{(i)}_2}.
    \end{align*}
 We now fix $Y,Z$ and use the Hoeffding bound (see e.g., Theorem 2.2.6 in \cite{vershynin}) to get that:
 \begin{align*}
     \Prob{\sum_{i=1}^{m}w(T^{(i)})Y_{T_1^{(i)}}Z_{T_2^{(i)}}\geq 4\sqrt{\frac{m}{n^k}}n^{3k/4}}\leq \exp\left(-8n^{k/2}\right) 
 \end{align*}
 We apply the union bound over all $Y,Z$ to get that:
 \begin{align*}
     \Prob{\exists Y,Z: \sum_{i=1}^{m}w(T^{(i)})Y_{T_1^{(i)}}Z_{T_2^{(i)}}\geq 4\sqrt{\frac{m}{n^k}}n^{3k/4}}&\leq 2^{2n^{k/2}}\exp\left(-8n^{k/2}\right)\\
     &\leq \exp\left(2n^{k/2}-8n^{k/2}\right)\\
     &=\exp\left(-6n^{k/2}\right)
 \end{align*}
 It follows that with probability at least $1-\exp\left(-6n^{k/2}\right)$ it holds that:
 \begin{align*}
     \max_{\substack{U_T,\\ \|U_T\|\leq 1}}\sum_{i=1}^{m}w(T^{(i)})\langle U_{T^{(i)}_1}, U_{T^{(i)}_2}\rangle \leq  4K_G\sqrt{\frac{m}{n^k}}n^{\frac{3k}{4}}.
 \end{align*}
 The theorem follows. 
\end{proof}
\subsection{\texorpdfstring{Constructing Nearly $t$-wise $\mu$-independent Distributions}{Constructing Nearly t-wise Independent Distributions}}
\label{app:lem:t-wise-mu-approx}
In this section, we extend the results of~\cite{alon2003almost} to non-uniform distributions over large domains. We prove the following theorem.

\begin{theorem}\label{thm:t-wise-mu-approx}
Consider a probability measure $\nu$ on $[q]^k$ and a probability measure $\mu$ on $[q]$. Suppose that for every function 
$f : [q]^k \to \mathbb{C}$ 
depending on at most $t$ variables $a_{i_1}, \dots, a_{i_t}$, we have
\begin{equation}   
\label{eq:t-wise-mu-approx:main}
\left|
\sum_{(a_1, \dots, a_k) \in [q]^k} 
\left( \nu(a_1, \dots, a_k) - \mu(a_1) \cdots \mu(a_k) \right)
f(a_1, \dots, a_k)
\right|
\leq \epsilon \|f\|_{\infty}.
\end{equation}
Then there exists a 
$t$-wise $\mu$-independent
probability measure $\nu'$ such that $$d_{\uprightsf{TV}}(\nu,\nu')\leq 
O((2q)^k\;\epsilon^{1/(k+1)}).$$
\end{theorem}

We begin by proving the theorem under the additional assumption that $\mu(z) \geq \gamma / q$ for every $z$.

\begin{lemma}
\label{lem:t-wise-mu-approx}
Under the conditions of Theorem~\ref{thm:t-wise-mu-approx}, when $\mu(z) \geq \gamma / q$ for every $z \in [q]$ and some $\gamma \in (0, 1]$, there exists a $t$-wise $\mu$-independent probability measure $\nu'$ such that
$
d_{\uprightsf{TV}}(\nu, \nu') \leq 
\binom{k}{t} \, \gamma^{-k} q^t \, \epsilon.
$
\end{lemma}

\begin{proof}
Let $\mu^{\otimes k}$ be the product measure on $[q]^k$ defined as 
$\mu^{\otimes k} (a_1,\dots,a_k) = \mu(a_1)\cdots \mu(a_k)$ and 
$F(a_1,\dots,a_k) = 
\nu(a_1,\dots,a_k)-\mu^{\otimes k}(a_1,\dots,a_k)$.
Consider the $k$-dimensional Discrete Fourier Transform of the function $F$:
\[
\widehat{F}(b) = \sum_{a \in [q]^k} F(a) \cdot \overline{\chi_b(a)},
\]
where
$
\chi_b(a) = e^{\frac{2\pi i \langle a, b \rangle}{q}}
$
is the character associated with the frequency $b = (b_1, \dots, b_k)$. The inverse transform is given by
\[
F(a) = 
\frac{1}{q^k} 
\sum_{b\in [q]^k} \widehat{F}(b) \cdot \chi_b(a).
\]

Let $B_t$ be the set of all vectors $b = (b_1, \dots, b_k)$ in $[q]^k$ with at most $t$ nonzero coordinates. Every character $\chi_b(a_1, \dots, a_k)$ with $b$ in $B_t$ depends on at most $t$ of the coordinates $a_i$ (those for which $b_i \neq 0$). Also, 
$\|\chi_b\|_{\infty}
=\|\overline{\chi_b}\|_{\infty}=1$. Thus, by the assumption of Theorem~\ref{thm:t-wise-mu-approx}, we have 
\[
|\widehat{F}(b)| = \left| \sum_{a \in [q]^k} F(a) \cdot \overline{\chi_b(a)} \right| \leq \epsilon.
\]

We now construct the measure $\nu'$ by removing all Fourier frequencies from $\nu$ corresponding to $b \in B_t$ to make the measure  $t$-wise $\mu$-independent, and by slightly rescaling the remaining Fourier coefficients to guarantee that $\nu'$ is nonnegative. Specifically, define
\[
\nu'(a) = \mu^{\otimes k}(a) + \frac{(1 - \delta)}{q^k} \sum_{b \in [q]^k \setminus B_t} \widehat{F}(b)\, \chi_b(a),
\]
where $\delta = \epsilon \gamma^{-k}|B_t|$. Note that $\nu(a')$ is a real number for every $a \in [q]^k$, because 
$\widehat{F}(b) = \overline{\widehat{F}(-b)}$ and 
$\chi_b(a) = \overline{\chi_{-b}(a)}$. 
Thus, the sum 
$\widehat{F}(b)\chi_b(a) + 
\widehat{F}(-b)\chi_{-b}(a)$ is real. 
Here, $-b$ denotes the coordinate-wise negation modulo $q$. Also, observe that 
\[\widehat{F}(0) = 
\sum_{a\in[q]^k} F(a) =
\nu([q]^k) - \mu^{\otimes k}([q]^k) = 1 - 1 =0.
\]

We next verify that $\nu'$ satisfies the conditions of Lemma~\ref{lem:t-wise-mu-approx}. First, we check that $\nu'$ is a probability measure; that is, $\nu'([q]^k) = 1$ and $\nu'$ is nonnegative.
Since $\sum_{a \in [q]^k} \chi_b(a) = 0$ for all $b \neq 0$, we have
\[
\nu'([q]^k)\equiv
\sum_{a \in [q]^k} \nu'(a) = \sum_{a \in [q]^k} \mu^{\otimes k}(a)
\equiv \mu^{\otimes k}([q]^k)= 1.
\]
Now, using that $\nu(a) = \mu^{\otimes k}(a) + F(a)$, we get
\begin{align*}   
\nu'(a) &= 
\mu^{\otimes k}(a) + 
(1-\delta) \Big(F(a)
- 
\frac{1}{q^k} \sum_{b \in B_t} \widehat{F}(b)\, \chi_b(a)
\Big) 
\\&=
\delta\mu^{\otimes k}(a) + 
(1-\delta)\,\nu(a) - 
\frac{(1 - \delta)}{q^k} \sum_{b \in B_t} \widehat{F}(b)\, \chi_b(a)
\\&\geq 
\delta\mu^{\otimes k}(a) 
- 
\frac{(1 - \delta)}{q^k} \sum_{b \in B_t} \widehat{F}(b)\, \chi_b(a).
\end{align*}
Recall that $\mu(z)\geq \gamma/q$ for all $z\in[q]$ and hence
$\mu^{\otimes k}(a) \geq  \gamma^k/q^k$. Therefore,
\[
\nu'(a) \geq \frac{1}{q^k} \Big[
\delta \gamma^k - (1 - \delta)
\sum_{b \in B_t} \widehat{F}(b)\, \chi_b(a)
\Big].
\]
Since each Fourier coefficient $\widehat{F}(b)$ is bounded in absolute value by $\epsilon$, $|\chi_b(a)| = 1$, and $\delta = \epsilon \gamma^{-k} |B_t|$, the right-hand side is at least $\tfrac{1}{q^k}\Big[\epsilon |B_t| - (1 - \delta)\, \epsilon\, |B_t|\Big]  > 0$. This proves that $\nu'(a) > 0$ for all $a$, and hence $\nu'$ is a probability measure. 

We now show that $\nu'$ is $t$-wise $\mu$-independent, and then estimate $d_{\uprightsf{TV}}(\nu, \nu')$. Consider $k$ random variables $(X_1, \dots, X_k)$ with joint distribution given by the measure $\nu'$. We need to show that they are $t$-wise $\mu$-independent. That is, for every $t$ indices $i_1, \dots, i_t$ and values $a^*_{i_1}, \dots, a^*_{i_t}$, we must have
\[
\Pr(X_{i_1} = a^*_{i_1}, \dots, X_{i_t} = a^*_{i_t}) = 
\mu(a^*_{i_1}) \cdots \mu(a^*_{i_t}).
\]
The probability above can be expressed as the expectation of the indicator function of the event $\mathbf{1}\{a_{i_1} = a^*_{i_1}, \dots, a_{i_t} = a^*_{i_t}\}$:
\[
\Pr(X_{i_1} = a^*_{i_1}, \dots, X_{i_t} = a^*_{i_t}) = 
\sum_{a \in [q]^k}
\mathbf{1}\{a_{i_1} = a^*_{i_1}, \dots, a_{i_t} = a^*_{i_t}\} \cdot \nu'(a).
\]
The indicator function $\mathbf{1}\{a_{i_1} = a^*_{i_1}, \dots, a_{i_t} = a^*_{i_t}\}$ depends only on the $t$ variables $a_{i_1}, \dots, a_{i_t}$ and is therefore orthogonal to all characters $\chi_b$ with $b \notin B_t$. That is,
\[
\sum_{a \in [q]^k} \chi_b(a) \cdot
\mathbf{1}\{a_{i_1} = a^*_{i_1}, \dots, a_{i_t} = a^*_{i_t}\} = 0.
\]
Hence,
\begin{align*}
\sum_{a \in [q]^k}
\mathbf{1}\{a_{i_1} = a^*_{i_1}, \dots, a_{i_t} = a^*_{i_t}\} \cdot \nu'(a)
&= 
\sum_{a \in [q]^k}
\mathbf{1}\{a_{i_1} = a^*_{i_1}, \dots, a_{i_t} = a^*_{i_t}\}
\cdot \mu^{\otimes k}(a)
\\
&=
\mu^{\otimes k}(\{a \in [q]^k : a_{i_1} = a^*_{i_1}, \dots, a_{i_t} = a^*_{i_t}\})\\
&=
\mu(a^*_{i_1}) \cdots \mu(a^*_{i_t}).
\end{align*}

We are almost done. It only remains to estimate $d_{\uprightsf{TV}}(\nu, \nu')\equiv \frac{1}{2} \|\nu - \nu'\|_1
$.
We have
\begin{align*}
\nu(a) - \nu'(a) &= 
\Big(\mu^{\otimes k}(a) + F(a)\Big) - 
\Big(
\mu^{\otimes k}(a) 
+ \frac{(1 - \delta)}{q^k} \sum_{b \in [q]^k \setminus B_t} \widehat{F}(b)\, \chi_b(a)
\Big)\\
&= F(a) - \frac{(1 - \delta)}{q^k} \sum_{b \in [q]^k \setminus B_t} \widehat{F}(b)\, \chi_b(a).
\end{align*}
We write $F(a) = \delta F(a) + (1 - \delta) F(a)$ and use the Fourier expansion of $(1 - \delta) F(a)$:
\begin{align*}
\nu(a) - \nu'(a) &= 
\delta F(a) + 
\frac{(1 - \delta)}{q^k} \sum_{b \in [q]^k} \widehat{F}(b)\, \chi_b(a)
- \frac{(1 - \delta)}{q^k} \sum_{b \in [q]^k \setminus B_t} \widehat{F}(b)\, \chi_b(a)\\
&=
\delta F(a) + 
\frac{(1 - \delta)}{q^k} \sum_{b \in B_t} \widehat{F}(b)\, \chi_b(a).
\end{align*}

We now bound the $\ell_1$ norm of the first term: $\|\delta F\|_1 = \delta \|F\|_1 = \delta \|\nu - \mu^{\otimes k}\|_1 \leq 2\delta$.
For the second term, we use that the $\ell_1$ norm of each function $\chi_b / q^k$ is $1$, and that $|\widehat{F}(b)| \leq \epsilon$ for all $b \in B_t$. Thus, the $\ell_1$ norm of each term $\frac{(1 - \delta)}{q^k} \widehat{F}(b)\, \chi_b(a)$ is at most $(1 - \delta)\epsilon <\epsilon$. Therefore, the  $\ell_1$ norm of $\nu-\nu'$ is upper bounded by 
$\delta + |B_t|\epsilon
= (\gamma^{-k} +1)|B_t|\epsilon$.  We get (using that $\gamma\leq 1$)
$$
d_{\uprightsf{TV}}(\nu, \nu') \equiv \frac{1}{2}\|\nu' - \nu\|_1
\leq |B_t|\gamma^{-k}\epsilon,
$$
Observing that $|B_t| \leq \binom{k}{t} q^t$ concludes the proof.
\end{proof}

\begin{proof}[Proof of Theorem~\ref{thm:t-wise-mu-approx}]  
We now extend the result to the general case, where $\mu$ is not necessarily bounded away from zero. Let $\gamma = \epsilon^{\nicefrac{1}{k+1}}$. If $\mu(z) \geq \gamma / q$ for all $z \in [q]$, then Lemma~\ref{lem:t-wise-mu-approx} applies directly. Otherwise, there exists at least one $z$ such that $\mu(z) < \gamma / q$.

Without loss of generality, assume that $\mu(0) \geq \mu(1) \geq \cdots \geq \mu(q - 1)$. If this is not the case, we may simply rearrange the elements of $[q]$. We apply Lemma~\ref{lem:t-wise-mu-approx} to an auxiliary measure $\tilde{\mu}$ that satisfies the lower bound assumption of the lemma. We define $\tilde{\mu}$ by coarsening the domain of $\mu$ -- merging all low-probability elements into a single one. Since $\sum_{z=0}^{q-1} \mu(z) = 1$, there exists at least one $z$ with $\mu(z) \geq 1/q > \gamma/q$. Let $\tilde{q}$ be the largest index such that $\mu(z) \geq \gamma / q$ for all $z \leq \tilde{q}$. For $z > \tilde{q}$, we then have $\mu(z) < \gamma / q$.
Define $\tilde{\mu}$ on $[\tilde{q} + 1]$ by preserving the high-probability elements and combining the rest:
\[
\tilde{\mu}(z) =
\begin{cases}
\mu(z), & \text{if } z < \tilde{q}; \\
\sum_{z' = \tilde{q}}^{q - 1} \mu(z'), & \text{if } z = \tilde{q}.
\end{cases}
\]
Then $\tilde{\mu}(z) \geq \gamma / q$ for all $z \in [\tilde{q} + 1]$, and hence $\tilde{\mu}(z) \geq \tilde{\gamma} / (\tilde{q} + 1)$, where $\tilde{\gamma} = (\tilde{q} + 1) \gamma/ q$.

Let $(X_1, \dots, X_k)$ be random variables with joint distribution given by the measure $\nu$. Define the modified random variables $(\tilde{X}_1, \dots, \tilde{X}_k)$ as follows:
\[
\tilde{X}_i =
\begin{cases}
X_i, & \text{if } X_i < \tilde{q}; \\
\tilde{q}, & \text{if } X_i \geq \tilde{q}.
\end{cases}
\]
In other words, $\tilde{X}_i = \min(X_i, \tilde{q})$. We now apply Lemma~\ref{lem:t-wise-mu-approx} to the distribution of $(\tilde{X}_1, \dots, \tilde{X}_k)$, which we denote by $\tilde{\nu}$, and measure $\tilde{\mu}$. To do so, we must verify condition~\eqref{eq:t-wise-mu-approx:main} for  $\tilde{\nu}$ and $\tilde{\mu}$.

Consider an arbitrary function $\tilde{f} : [\tilde{q} + 1]^k \to \mathbb{C}$ that depends on at most $t$ variables. Define $f : [q]^k \to \mathbb{C}$ as follows:
\[
f(a_1, \dots, a_k) = \tilde{f}(\min(a_1, \tilde{q}), \dots, \min(a_k, \tilde{q})).
\]
Then, 
$
\tilde{f}(\tilde{X}_1, \dots, \tilde{X}_k) = f(X_1, \dots, X_k)$ and $\E[\tilde{f}(\tilde{X}_1, \dots, \tilde{X}_k)] = \E[f(X_1, \dots, X_k)]
$.
Similarly, let $(M_1, \dots, M_k)$ and $(\tilde{M}_1, \dots, \tilde{M}_k)$ be random variables with joint distributions given by $\mu^{\otimes k}$ and $\tilde{\mu}^{\otimes k}$, respectively. Then, 
$
\E[\tilde{f}(\tilde{M}_1, \dots, \tilde{M}_k)] = \E[f(M_1, \dots, M_k)]
$.
Inequality~\eqref{eq:t-wise-mu-approx:main} is equivalent to
\[
\left| \E[f(X_1, \dots, X_k)] - \E[f(M_1, \dots, M_k)] \right| \leq \epsilon \|f\|_{\infty}.
\]
 We have
\[
\Big| \E[\tilde{f}(\tilde{X}_1, \dots, \tilde{X}_k)] - \E[\tilde{f}(\tilde{M}_1, \dots, \tilde{M}_k)] \Big| =
\Big| \E[f(X_1, \dots, X_k)] - \E[f(M_1, \dots, M_k)] \Big| 
\leq \epsilon \|f\|_{\infty}.
\]

Therefore, we can apply Lemma~\ref{lem:t-wise-mu-approx} to the distributions $\tilde{\nu}$ and $\tilde{\mu}$. By this lemma, there exist $t$-wise $\tilde{\mu}$-independent random variables $\tilde{Y}_1, \dots, \tilde{Y}_k$ such that the total variation distance between the distributions of 
$\tilde{X}_1, \dots, \tilde{X}_k$
and 
$\tilde{Y}_1, \dots, \tilde{Y}_k$ is at most
\begin{equation}\label{eq:dTV-tX-tY}
\binom{k}{t} \cdot \tilde{\gamma}^{-k}\cdot (\tilde{q} + 1)^t \epsilon 
= \binom{k}{t} \cdot \gamma^{-k} \; \frac{q^k}{(\tilde{q} + 1)^k} \cdot (\tilde{q} + 1)^t \epsilon 
\leq \left(\frac{2q}{\gamma} \right)^k \epsilon.
\end{equation}
Here, we used the fact that the random variables are defined on $[\tilde{q} + 1]^k$, and that $\tilde{\mu}(z) \geq \tilde{\gamma}/(\tilde{q} + 1)$.

Finally, we define the desired random variables $Y_1, \dots, Y_k$ as follows: if $\tilde{Y}_i < \tilde{q}$, we let $Y_i = \tilde{Y}_i$. If $\tilde{Y}_i = \tilde{q}$, then we let $Y_i = a_i$ with probability
\[
\frac{\mu(a_i)}{\mu(\{\tilde{q}, \dots, q - 1\})} \quad \text{for } a_i \geq \tilde{q},
\]
independently for all $i$.

We verify that $Y_1, \dots, Y_k$ are $t$-wise $\mu$-independent random variables. Each $Y_i$ depends only on $\tilde{Y}_i$ and a random coin toss (used to select the value of $Y_i$ if $\tilde{Y}_i = \tilde{q}$), which is independent of all other variables. Thus, if $\tilde{Y}_{i_1}, \dots, \tilde{Y}_{i_t}$ are independent, then so are $Y_{i_1}, \dots, Y_{i_t}$. We get that $Y_1,\dots,Y_k$ are $t$-wise independent. For every $a_i<\tilde{q}$,
\[
\Pr(Y_i = a_i) = 
\Pr(\tilde{Y}_i = a_i) = \mu(a_i).
\]
For $a_i \geq \tilde{q}$, we have
\[
\Pr(Y_i = a_i) = 
\Pr(\tilde{Y}_i = \tilde{q}) \cdot \frac{\mu(a_i)}{\mu(\{\tilde{q}, \dots, q - 1\})} = \mu(a_i).
\]
Therefore, the marginal distribution of each $Y_i$ is $\mu$.

It only remains to bound the total variation distance between the distributions of 
$(X_1, \dots, X_k)$ and $(Y_1, \dots, Y_k)$. The total variation distance between $(\tilde{X}_1, \dots, \tilde{X}_k)$ and 
$(\tilde{Y}_1, \dots, \tilde{Y}_k)$ is at most 
$
\big( \frac{2q}{\gamma} \big)^k \epsilon
\quad \text{(see~\eqref{eq:dTV-tX-tY})}
$. 
The total variation distance between 
$(X_1, \dots, X_k)$ and 
$(\tilde{X}_1, \dots, \tilde{X}_k)$ is
\[
\Pr(\exists i \text{ such that } X_i \neq \tilde{X}_i) 
= \Pr(\exists i \text{ such that } X_i > \tilde{q}) 
\leq \sum_{i=1}^{k-1} \Pr(X_i > \tilde{q}).
\]
Each $\Pr(X_i > \tilde{q})$ is upper bounded by $\mu(\{\tilde{q}+1, \dots, q - 1\}) + \epsilon$, because
\[
\Pr(X_i > \tilde{q}) = 
\E[\mathbf{1}\{X_i > \tilde{q}\}] \leq \mu(\{\tilde{q}+1, \dots, q - 1\}) + \epsilon,
\]
as the indicator function $\mathbf{1}\{X_i > \tilde{q}\}$ depends only on the single variable $X_i$.
Hence,
\[
\Pr(\exists i \text{ such that } X_i \neq \tilde{X}_i) 
\leq k \left( \mu(\{\tilde{q}+1, \dots, q - 1\}) + \epsilon \right) 
\leq kq\gamma + k\epsilon.
\]
Similarly, the total variation distance between 
$(Y_1, \dots, Y_k)$ and 
$(\tilde{Y}_1, \dots, \tilde{Y}_k)$ is
\[
\Pr(\exists i \text{ such that } Y_i \neq \tilde{Y}_i) 
= \Pr(\exists i \text{ such that } {Y}_i > \tilde{q}) 
\leq \sum_{i=1}^k \Pr({Y}_i > \tilde{q})
= k \mu(\{\tilde{q}+1, \dots, q - 1\}) 
\leq kq\gamma.
\]

Therefore, by the triangle inequality, the total variation distance between the distributions of 
$(X_1, \dots, X_k)$ and 
$(Y_1, \dots, Y_k)$ is at most
$
\left( \frac{2q}{\gamma} \right)^k \epsilon + 2kq\gamma + k\epsilon
$. Since $\gamma = \epsilon^{1/(k+1)}$, we obtain the bound
$$
(2q)^k \epsilon^{1/(k+1)} + 2kq \epsilon^{1/(k+1)} + k\epsilon 
= O\left((2q)^k \epsilon^{1/(k+1)}\right).
$$
\end{proof}

\subsection{Examples from Ordering CSPs: Betweenness and Cyclic Ordering}\label{subsec:BTW}

\paragraph{Betweenness Predicate.} We prove a lower and matching upper bound for the probability that any pairwise uniform distribution satisfies the betweenness predicate. To that end, we show that no pairwise uniform distribution over $[0,1]^3$ satisfies the predicate with probability greater than $1/2$. We then provide an example of a pairwise uniform distribution that satisfies the Betweenness predicate with probability $1/2$. Recall that real numbers $a,b,c\in [0,1]$ satisfy the betweenness predicate if $a<b<c$ or $c<b<a$, for convenience for the remaining of this section we will use $f_{\uprightsf{btw}}$ to denote the betweenness predicate, in other words $f_{\uprightsf{btw}}(a,b,c)=\ind{a< b<c \text{ or } c<b<a}$.
We have the following lemma:
\begin{lemma}
    Let $\calD$ be a pairwise uniform distribution over $[0,1]^3$ and $(a,b,c)\sim \calD$, then:
    \begin{align*}
        \Prob{f_{\uprightsf{btw}}(a,b,c)=1}\leq 1/2
    \end{align*}
\end{lemma}
\begin{proof}
    We will, equivalently, lower bound the probability of the event that $f_{\uprightsf{btw}}(a,b,c)=0$. We have that:
    \begin{align*}
        \Prob{f_{\uprightsf{btw}}(a,b,c)=0}\geq \Prob{b\leq \frac{1}{2} \text{ and } a,c>b}+\Prob{b> \frac{1}{2} \text{ and 
 } a,c< b}
    \end{align*}
    We bound the probability $\Prob{b\leq \frac{1}{2} \text{ and } a,c>b}$, we have that:
    \begin{align*}
        \Prob{b\leq \frac{1}{2} \text{ and } a,c\geq b}&=\int_{0}^{1/2}\cProb{a,c\geq b}{b=x}dx\\
        &=\int_{0}^{1/2}(1-\cProb{a< x \text{ or }c<x  }{b=x})dx\\
        &\geq \int_{0}^{1/2}(1-\cProb{a<x}{b=x}-\cProb{c<x}{b=x})dx\\
        &=\int_{0}^{1/2}(1-2x) dx\\
        &=\frac{1}{4}.
    \end{align*}
    By symmetry, we also get that $\Prob{b>\frac{1}{2} \text{ and } a,c<b}\geq \frac{1}{4}$, which gives the result.
\end{proof}

We now provide an example of a pairwise uniform distribution that satisfies the betweenness predicate with probability $1/2$. Let $\calD^{*}$ be the distribution sampled according to the following procedure:
\begin{enumerate}
    \item Sample $a$ and $b$ uniformly at random.
    \item Let
    \begin{align*}
        c=\begin{cases}
            a+b, & \text{ if } a+b\leq 1\\
            a+b-1, & \text{ otherwise}
        \end{cases}
    \end{align*}
\end{enumerate}
We have the following lemma:
\begin{lemma}
    Let $(a,b,c)\sim \calD^*$, then:
    \begin{align*}
        \Prob{f_{\uprightsf{btw}}(a,b,c)=1}=\frac{1}{2}
    \end{align*}
\end{lemma}
\begin{proof}
    Observe that $\Prob{f_{\uprightsf{btw}}(a,c,b)=1}=0$ and furthermore, by symmetry, $\Prob{f_{\uprightsf{btw}}(a,b,c)=1}=\Prob{f_{\uprightsf{btw}}(b,a,c)=1}$, meaning that:
    \begin{align*}
        \Prob{f_{\uprightsf{btw}}(a,b,c)=1}=\frac{1}{2}
    \end{align*}
\end{proof}
\begin{remark}
    This distribution is useful as an example for other predicates as well. In particular, it is an example where the non-betweenness predicate is satisfied with probability $1$ and the triplets predicate on caterpillar is satisfied with probability $1/2$. 
\end{remark}

\paragraph{Cyclic Ordering.}
We analyze the probability that any pairwise uniform distribution satisfies the Cyclic Ordering predicate. Recall that real numbers $a,b,c\in [0,1]$ satisfy the Cyclic Ordering predicate if $a<b<c$, $b<c<a$, or $c<a<b$. For the remainder of this section, we denote the Cyclic Ordering predicate by $f_{\uprightsf{cyc}}$, where $f_{\uprightsf{cyc}}(a,b,c)=\ind{a<b<c \text{ or } b<c<a \text{ or } c<a<b}$.

\begin{lemma}
    Let $\calD$ be a pairwise uniform distribution over $[0,1]^3$ and $(a,b,c)\sim \calD$, then:
    \begin{align*}
        \Prob{f_{\uprightsf{cyc}}(a,b,c)=1} \leq \frac{1}{2}
    \end{align*}
\end{lemma}
\begin{proof}
    We will prove a stronger statement, that $f_{\uprightsf{cyc}}$ is satisfied with probability \textit{exactly} $\frac{1}{2}$, if $(a,b,c)$ is sampled from a pairwise uniform distribution over $[0,1]^3$. The proof follows from the observation that we can write the predicate $f_{\uprightsf{cyc}}$ is as the sum of a constant term and functions that only depend on $2$ out of the three variables. In particular, we have that:
    \begin{align*}
        f_{\uprightsf{cyc}}(a,b,c)=\frac{1}{2}+\frac{1}{2}\left(\sign(b-a)+\sign(c-b)+\sign(a-c)\right)
    \end{align*}
    We have that:
    \begin{align*}
        \Prob{f_{\uprightsf{cyc}}(a,b,c)=1}&=\E[f_{\uprightsf{cyc}}(a,b,c)]\\
        &=\frac{1}{2}+\frac{1}{2}\E[\sign(b-a)+\sign(c-b)+\sign(a-c)]\\
        &=\frac{1}{2}
    \end{align*}
\end{proof}

\begin{remark}
    The uniform distribution over $[0,1]^3$ achieves this bound, thus our refutation results for Cyclic Ordering are optimal.
\end{remark}

\section{The Case of Several Negative Examples}\label{sec:arbitrary-r}

In this appendix, we prove the part of Lemma~\ref{lem:max-comp-pois-process} that is needed for fixed $r>1$. The argument follows the proof for $r=1$, but uses the nonlinear time scale of the process $\calG_{n,\Lambda_r}(t)$.

For $0<c<\lambda_r^*$, define
\[
\psi_r(c)=1-\left(1-\frac{c}{\lambda_r^*}\right)^{1/r},
\]
and let $\psi_r(c)=1$ for $c\geq\lambda_r^*$.

\begin{corollary}\label{cor:thm:giant-connected-comp-r}
For $\lambda_r^*$ defined in~\eqref{eq:def:lambda-r}, the following two claims hold.
\begin{enumerate}
\item 
For every positive $\delta$ and $c$, the size of the largest connected component in the random $\bbG(n,c/n)$ graph is at most $(1+\delta)\psi_r(c)n$ with high probability.
\item 
There exists a positive $c_r^*$ such that for every $\delta\in(0,1)$, the largest connected component in the random $\bbG(n,c_r^*/n)$ graph is at least $(1-\delta)\psi_r(c_r^*)n$ with high probability.
\end{enumerate}
\end{corollary}

\begin{proof}
Fix a positive $c$. If $c\leq 1$, then all components in $\bbG(n,c/n)$ have size $o(n)$ w.h.p., and Part I follows. If $c\geq\lambda_r^*$, then Part I is trivial because $\psi_r(c)=1$. Thus, we assume that $c\in(1,\lambda_r^*)$.

Let $\rho=\rho(c)$ be the asymptotic fraction of vertices in the giant component of $\bbG(n,c/n)$. By Theorem~\ref{thm:giant-connected-comp}, the largest connected component has size $(\rho+o(1))n$ w.h.p. Moreover, $\rho$ satisfies $1-\rho=e^{-c\rho}$, or equivalently, $c=-\ln(1-\rho)/\rho$. Therefore, by the definition of $\lambda_r^*$,
\[
\lambda_r^*
\leq
\frac{-\ln(1-\rho)}{\rho(1-(1-\rho)^r)}
=
\frac{c}{1-(1-\rho)^r}.
\]
Hence $1-(1-\rho)^r\leq c/\lambda_r^*$, and since $c<\lambda_r^*$, we get
\[
\rho
\leq
1-\left(1-\frac{c}{\lambda_r^*}\right)^{1/r}
=
\psi_r(c).
\]
This proves Part I.

We now prove Part II. Let $s_r^*\in(0,1)$ be a point where the function
$h_r(s)=-\ln(1-s)/(s(1-(1-s)^r))$ attains its minimum, and let
$c_r^*=-\ln(1-s_r^*)/s_r^*$. Since $h_r(s_r^*)=\lambda_r^*$, we have
$c_r^*=\lambda_r^*(1-(1-s_r^*)^r)$. Therefore, $\psi_r(c_r^*)=s_r^*$. Also, the equation defining $c_r^*$ is equivalent to
$1-s_r^*=e^{-c_r^*s_r^*}$. Thus, $s_r^*$ is the giant-component fraction in $\bbG(n,c_r^*/n)$. By Theorem~\ref{thm:giant-connected-comp}, the largest connected component in $\bbG(n,c_r^*/n)$ has size $(s_r^*+o(1))n$ w.h.p. Since $s_r^*=\psi_r(c_r^*)$, Part II follows.
\end{proof}

We now use Corollary~\ref{cor:thm:giant-connected-comp-r} to finish the proof of Lemma~\ref{lem:max-comp-pois-process} for fixed $r>1$.
First suppose that $\lambda>\lambda_r^*$. Let $c_r^*$ be the constant from Part II of Corollary~\ref{cor:thm:giant-connected-comp-r}, and let
\[
u=1-\left(1-\frac{c_r^*}{\lambda}\right)^{1/r}.
\]
Since $\lambda>\lambda_r^*$, we have $u<\psi_r(c_r^*)$. Choose $\delta\in(0,1)$ such that $(1-\delta)\psi_r(c_r^*)>u$. Let $t=u(n-2)$. The time $t$ need not be an integer, since the process $\calG_{n,\Lambda_r}(t)$ is defined for all real $t\in[0,n-2]$. By the definition of $\Lambda_r$, at time $t$ the graph $\calG_{n,\Lambda_r}(t)$, after replacing parallel edges by single edges, has the same distribution as
\[
\bbG\left(n,\frac{c_r^*+o(1)}{n}\right).
\]
By Part II of Corollary~\ref{cor:thm:giant-connected-comp-r}, the graph $\calG_{n,\Lambda_r}(t)$ contains a connected component of size at least $(1-\delta)\psi_r(c_r^*)n$ w.h.p. Since $(1-\delta)\psi_r(c_r^*)>u$, this size is at least $t+2$ for all sufficiently large $n$. This proves Part II of Lemma~\ref{lem:max-comp-pois-process} for fixed $r>1$.

We now turn our attention to Part 1 of the lemma. Fix $\lambda<\lambda_r^*$. Choose a small positive number $\beta$ such that $\lambda(1-(1-\beta)^r)<1$. For $u\in(0,1]$, we have
$\psi_r(\lambda(1-(1-u)^r))<u$. Indeed, this follows directly from $\lambda<\lambda_r^*$ and the definition of $\psi_r$. By compactness, we can choose a positive number $\varepsilon$ and an integer $s$ large enough so that, for $\Delta=(1-\beta)/s$ and every $j\in\{1,\dots,s\}$,
\[
(1+\varepsilon)\psi_r\!\left((1+\varepsilon)\lambda(1-(1-\beta-j\Delta)^r)\right)
<
\beta+(j-1)\Delta.
\]
Note that $\beta$, $\varepsilon$, $s$, and $\Delta$ do not depend on $n$.

Then, for every $n$, we define time points $t_0<\dots<t_{s+2}$, where $t_0=1$ and $t_{s+2}=n-2$, and show that for every $i\geq1$, with high probability the largest connected component in $\calG_{n,\Lambda_r}(t_i)$ has size less than $t_{i-1}+2$. We let $t_0=1$, $t_1=\sqrt[3]{n}$, and
$t_i=(\beta+(i-2)\Delta)(n-2)$ for $i\geq2$. These time points need not be integers, since the process $\calG_{n,\Lambda_r}(t)$ is defined for all real $t\in[0,n-2]$.

We first consider $t_1=n^{1/3}$. Since $1-e^{-\Lambda_r(t_1)}=O(n^{-5/3})$, the graph $\calG_{n,\Lambda_r}(t_1)$ has maximum degree at most $1$ w.h.p. Hence, its largest connected component has size at most $2$, which is less than $t_0+2=3$. Next consider $t_2=\beta(n-2)$. Since $1-e^{-\Lambda_r(t_2)}=(\lambda(1-(1-\beta)^r)+o(1))n^{-1}$ and $\lambda(1-(1-\beta)^r)<1$, the graph $\calG_{n,\Lambda_r}(t_2)$ is subcritical. Therefore, its largest connected component has size $O(\log n)$ w.h.p., which is less than $t_1+2$.

Now consider $i\geq3$ and let $j=i-2$. Then $t_i=(\beta+j\Delta)(n-2)$, and graph $\calG_{n,\Lambda_r}(t_i)$ is a
$\bbG(n,(\lambda(1-(1-\beta-j\Delta)^r)+o(1))n^{-1})$ graph. For all sufficiently large $n$, the $o(1)$ term can be absorbed into the factor $1+\varepsilon$. By Corollary~\ref{cor:thm:giant-connected-comp-r}, the size of its largest connected component is less than
\[
(1+\varepsilon)\psi_r\!\left((1+\varepsilon)\lambda(1-(1-\beta-j\Delta)^r)\right)n
<
(\beta+(j-1)\Delta)n
\leq
t_{i-1}+2
\]
with high probability. Hence, the largest connected component in $\calG_{n,\Lambda_r}(t_i)$ has size less than $t_{i-1}+2$ w.h.p.

The number of time points $t_0,\dots,t_{s+2}$ does not depend on $n$. Thus, with high probability, for all $i\in\{1,\dots,s+2\}$, the size of the largest connected component in $\calG_{n,\Lambda_r}(t_i)$ is less than $t_{i-1}+2$. Suppose that this event holds. We claim that, in this case, the size of the largest connected component in $\calG_{n,\Lambda_r}(t)$ is less than $t+2$ for all real $t\in[1,n-2]$. Pick such a $t$. It belongs to one of the intervals $[t_{i-1},t_i]$. Graph $\calG_{n,\Lambda_r}(t)$ is a subgraph of $\calG_{n,\Lambda_r}(t_i)$ because the multiplicity of every edge $(a,b)$ in $\calG_{n,\Lambda_r}(t)$ is a non-decreasing function of $t$. The size of the largest connected component in $\calG_{n,\Lambda_r}(t)$ is at most the size of the largest connected component in $\calG_{n,\Lambda_r}(t_i)$, which in turn is less than $t_{i-1}+2\leq t+2$.
\subsection{Proof of Lemma~\ref{lem:f-decreasing}}\label{sec:lem:f-decreasing}
\begin{proof}
Let $L(x)=-\ln(1-g(x))$. Then $L(x)>0$ on $\left[\frac13,\frac12\right]$, and
\[
f(x)=\frac{H(x)}{L(x)}.
\]
We prove that $f'(x)<0$ for $x\in\left(\frac13,\frac12\right)$. Write,
\[
f'(x)
=
\frac{H'(x)L(x)-H(x)L'(x)}{L(x)^2}.
\]
We show that the numerator is negative, or equivalently, that the following
expression is positive:
\[
H(x)L'(x)-L(x)H'(x)
=
(H(x)-L(x))L'(x)+L(x)(L'(x)-H'(x)).
\]
To do so, we prove that $L'(x) \geq 0$; $H(x)-L(x)>0$; and $L'(x)-H'(x)>0$.

Observe that both $H$ and $L$ are increasing, since
$H$ is the binary entropy (increasing on $[0,1/2]$) and $L$ is increasing because $g(x)=(2x(1-x))^2$ is increasing on
$[1/3,1/2]$. Thus, $L'(x)\geq 0$ (as $L(x)$ is differentiable). Also, since $H(1/3) > L(1/2)$, we have
\[
H(x) \geq H(1/3) > L(1/2) \geq L(x)
\qquad
\text{for }x\in\left[\frac13,\frac12\right].
\]
It remains to prove that
$
L'(x)>H'(x)
$
for $x\in\left(\frac13,\frac12\right)$.
Let $y=1-2x$. Then $0<y<1/3$, and
\begin{align*}
H'(x)
&=
\ln\frac{1-x}{x}
=
\ln\frac{1+y}{1-y}
=
2\operatorname{arctanh} y,\\
L'(x)&=\frac{g'(x)}{1-g(x)}
=
\frac{8x(1-x)(1-2x)}{1-4x^2(1-x)^2}
=
\frac{8y(1-y^2)}{3+2y^2-y^4}.
\end{align*}
Therefore, in order to prove that $L'(x)>H'(x)$, it suffices to show that
\[
\operatorname{arctanh} y
<
\frac{4y(1-y^2)}{3+2y^2-y^4}.
\]
We bound the left-hand side as follows:
\[
\operatorname{arctanh}y
=
\sum_{n=0}^{\infty}\frac{y^{2n+1}}{2n+1}
\le
y+\frac13\sum_{n=1}^{\infty}y^{2n+1}
=
\frac{y(3-2y^2)}{3(1-y^2)}.
\]
We now verify that this last expression is smaller than
$\frac{4y(1-y^2)}{3+2y^2-y^4}$. Since all denominators are positive and
$y>0$, this is equivalent to
$
(3-2y^2)(3+2y^2-y^4)<12(1-y^2)^2
$.
Expanding, the difference between the right-hand side and the left-hand side is
\[
12(1-y^2)^2-(3-2y^2)(3+2y^2-y^4)
=
3-24y^2+19y^4-2y^6.
\]
This is positive for $0<y<1/3$, because
$
3-24y^2>\frac13$ and $19y^4-2y^6>0$.  Hence $L'(x)>H'(x)$.

\medskip

This completes the proof that $f'(x)<0$ for $x\in\left(\tfrac13,\tfrac12\right)$. Hence, $f$ is decreasing on
$\left[\tfrac13,\tfrac12\right]$. 
\end{proof}

\section{Unsatisfiability at High Linear Regime}\label{sec:unsat-high-density}


In this section, we prove Theorem~\ref{th:linear-density-upper}, where we consider \textit{ordered} phylogenetic CSPs that contain Ordering CSPs and (unordered) Phylogenetic CSPs (for the remaining of this section we will refer to ordered Phylogenetic CSPs as Phylogenetic CSPs). Consider a random instance $\calI_{f_{\uprightsf{phy}}}(n,m)$ with $n$ variables and $m=C n$ constraints ($C$ is constant depending on $\epsilon$) sampled according to our random model. 

\begin{tcolorbox}
\begin{theorem}[Information-Theoretic Bound for $\opt(\mathcal{I}_{f})$]\label{th:linear-density-upper}
 For every $\epsilon>0$, let $C$ be a sufficiently large constant and consider a random instance with $m=Cn$ constraints. For Ordering CSPs, there is no assignment satisfying more than a fraction of $\alpha^*+\epsilon$ of constraints w.h.p., where $\alpha^*$ is the approximation ratio achieved by outputting a uniformly random permutation. Moreover, for Phylogenetic CSPs, there is no assignment satisfying more than a fraction of $\alpha^*+\epsilon$ of constraints w.h.p., where $\alpha^*$ is the approximation ratio of the best biased random assignment (see Definition \ref{def:biased}). 
\end{theorem}
\end{tcolorbox}

The goal of this section is to prove that the optimum solution for a random instance does not attain a value significantly better than the best biased random assignment. We will show that the best solution for $\calI_{f_{\uprightsf{phy}}}(n,m)$ is close to the biased random assignment approximation ratio $\alpha^*(f_{\uprightsf{phy}})$:
\begin{equation}\label{eq:upper}
\val(\varphi^*,\calI) \le \alpha^*(f_{\uprightsf{phy}})+\epsilon 
\end{equation}
\noindent where $\varphi^*$ is the optimum assignment for instance $\calI$ that maps the $n$ variables $V$ to the $n$ leaves of an ordered phylogenetic tree (see Section \ref{subsec:phylogenetic_definitions}). Recall that for a solution $\varphi$, we defined value $\val(\varphi, \calI)$ to be the fraction of satisfied phylogenetic constraints or, more generally, the average value of phylogenetic constraints of $\calI$ in solution $\varphi$ divided by $m$, the total number of constraints in $\calI$.

Our proof uses again the notion of coarse solutions defined in Section \ref{sec:coarse-soln}. Consider the optimal solution $\varphi^*$ for the random instance $\calI$. Fix a positive $\epsilon$ and  sufficiently large number $q$.
By Lemma~\ref{lem:cs-good}, there exists a coarse solution 
$\xi^* \in \Xi_{\frac{16}{q}, q}$ such that 
$\val(\varphi^*, \calI) \leq \val^+(\xi^*,\calI)$. In Lemma~\ref{lem:cs-bad}, we show that $\val^+(\xi,\calI) \leq \alpha^*(f_{\uprightsf{phy}}) + \epsilon$ for large enough constant $C$. Therefore,
$\val(\varphi^*, \calI) \leq \alpha^*(f_{\uprightsf{phy}}) + \epsilon$ as claimed in Theorem~\ref{th:linear-density-upper}.

\begin{lemma}\label{lem:cs-bad}
For every phylogenetic payoff function $f_{\uprightsf{phy}}$ and every positive $\epsilon$, there exists $q$, such that for every sufficiently large constant $C_{\ref{lem:cs-bad}}$ the following holds with high probability: 
    \[
\max_{\xi \in \Xi_{\frac{16}{q}, q}}    
\val^+(\xi,\calI_{f_{\uprightsf{phy}}}(n,m)) \leq \alpha^*(f_{\uprightsf{phy}}) + \epsilon.
    \]
\end{lemma}
\begin{proof}
Pick an integer $q > 8k^2/\epsilon$.
Fix a coarse solution $\xi \in \Xi_{\frac{16}{q}, q}$. We shall use this solution to define a biased random assignment algorithm. First, the algorithm randomly and independently assigns each variable $x \in V$ to a leaf $l$ with probability $n_l/n$, where $n_l = |\{x : \xi(x) = l\}|$ is the number of variables assigned to leaf $l$ by the coarse solution $\xi$.  Then, it recursively partitions variables in each leaf splitting them with a 50\%-50\% probability at every step. We call this biased randomized assignment algorithm the $\xi$-biased algorithm and denote its approximation factor by $\alpha_\xi(f_{\uprightsf{phy}})$. In other words, $\alpha_\xi(f_{\uprightsf{phy}})$ is the expected value of the solution created by the $\xi$-biased algorithm. 

Consider arbitrary $k$ variables $y^1,\dots,y^k$. The  $\xi$-biased algorithm first assigns $y^1,\dots,y^k$ to random leaves $\zeta(y^1),\dots,\zeta(y^k)$ and then after recursively splitting variables in the leaves assigns $y^1,\dots,y^k$ to 
$\eta(y^1),\dots, \eta(y^k)$. If leaves $\zeta(y^1),\dots,\zeta(y^k)$ are all distinct, then 
$$f_{\uprightsf{phy}}(\zeta(y^1),\dots,\zeta(y^k)) = f_{\uprightsf{phy}}(\eta(y^1),\dots,\eta(y^k)).$$
Moreover, if leaves $\zeta(y^1),\dots,\zeta(y^k)$ have distinct colors, then
$$f_{\uprightsf{phy}}^+(\zeta(y^1),\dots,\zeta(y^k)) =  f_{\uprightsf{phy}}(\zeta(y^1),\dots,\zeta(y^k)) = f_{\uprightsf{phy}}(\eta(y^1),\dots,\eta(y^k)).$$

The probability that two or more leaves among $\zeta(y^1),\dots,\zeta(y^k)$ have the same color is at most $\nicefrac{16}{q}\cdot \binom{k}{2} < \eps$, because the probability that variable $y^j$ is assigned to a leaf of any given color is at most $16/q$ (by the definition of class $\Xi_{\frac{16}{q},q}$).
Therefore, we have
$$
\E[f_{\uprightsf{phy}}^+(\zeta(y^1),\dots,\zeta(y^k))] \leq \E[f_{\uprightsf{phy}}(\eta(y^1),\dots,\eta(y^k))] + \epsilon
= \alpha_{\xi}(f_{\uprightsf{phy}}) +\epsilon\leq \alpha^*(f_{\uprightsf{phy}}) + \epsilon.
$$

Now consider one of the constraints $(X_i^1, \dots, X_i^k)$ in the random instance $\calI_{f_{\uprightsf{phy}}}(n,m)$. Variables $X_i^1, \dots, X_i^k$ are chosen randomly in $V$ and then mapped to the leaves of $T_\xi$ by the coarse solution $\xi$. Since $n_u$ variables from $V$ are mapped to leaf $u$, we have $\Pr(\xi(X_i^j) = u) = n_u/n$. Thus, each $\xi(X_i^j)$ has the same marginal distribution as $\zeta(y^j)$. Note that $\xi(X_i^1),\dots,\xi(X_i^k)$ are not independent since $X_i^1,\dots,X_i^k$ are chosen in $V$ without replacement. So, the distributions of the tuples 
$(\zeta(y^1),\dots,\zeta(y^k))$
and 
$(\xi(X_i^1),\dots,\xi(X_i^k))$ are not identical but they are very close if $n$ is sufficiently large. Indeed, let $\tilde X^1,\dots, \tilde X^k$ be $k$ random variables sampled in $V$ \emph{with replacement}. Then, $(\zeta(y^1),\dots,\zeta(y^k))$
and 
$(\xi(\tilde X^1),\dots,\xi(\tilde 
X^k))$ have the same distribution, and the total variation distance between distributions of $(X_i^1,\dots,X_i^k)$ and 
$(\tilde X^1,\dots, \tilde X^k)$ is at most $\binom{k}{2}/n$, which is an upper bound on the probability that $\tilde X^a = \tilde X^b$ for some $a$ and $b$.
This probability is less than $\eps$ for sufficiently large $n$. Hence,
$$
\E[f_{\uprightsf{phy}}^+(\xi(X_i^1),\dots,\xi(X_i^k))]
\leq
\E[f_{\uprightsf{phy}}^+(\zeta(y^1),\dots,\zeta(y^k))] + \epsilon \leq \alpha^*(f_{\uprightsf{phy}}) + 2\epsilon.
$$
We showed that the expected ``+''-value of every payoff function in $\calI_{f_{\uprightsf{phy}}}(n,m)$ is at most $\alpha^*(f_{\uprightsf{phy}}) + 2\epsilon$. We now estimate the probability that the value $\val^+(\xi,\calI_{f_{\uprightsf{phy}}}(n,m))$
of the instance is larger than $\alpha^*(f_{\uprightsf{phy}}) + 3\epsilon$. By Hoeffding's inequality, we have 
$$
\Pr\big[\val^+(\xi,\calI_{f_{\uprightsf{phy}}}(n,m))
\geq
\alpha^*(f) + 3\epsilon\big] 
\leq e^{-2\epsilon^2 m}.
$$
The number of all coarse solutions in $\Xi_{\frac{16}{q},q}$ is upper bounded by $q^{3q}q^n$ (there are at most $q^q$ trees with $q$ leaves; for each tree there are  at most $q!<q^q$ orderings of the leaves and each tree can be colored in at most $q^q$ different ways; there are at most $q^n$ ways to assign variables in $V$ to the leaves of the tree).
Using the union bound over all coarse solutions 
$\xi$ in $\Xi_{\frac{16}{q},q}$, we get
$$
\Pr\left[\exists \xi \in \Xi_{\frac{16}{q}, q} : \val^+(\xi,\calI_{f_{\uprightsf{phy}}}(n,m))
\geq
\alpha^*(f_{\uprightsf{phy}}) + 3\epsilon\right] \leq q^{3q} \cdot q^n \cdot e^{-2\epsilon^2 m}
= q^{3q} \cdot e^{-(2\epsilon^2 C-\ln q)n}
.
$$
This probability tends to $0$ as $n$ tends to infinity if $C > \frac{\ln q}{2\epsilon^2}$. This finishes the proof.
\end{proof}

\section{Figures}\label{sec:figures}

\begin{figure}[h]
    \centering
    {

\begin{tikzpicture}[
    dot/.style={draw, thick, circle, fill=black, inner sep=0pt, outer sep=0pt, opacity=0},
    leaf text/.style={font=\sffamily, anchor=west, align=left},
    thick,
    line cap=round,
    line join=round,
    scale=0.9, every node/.style={scale=0.9} 
]

\newcommand{\connect}[4]{
    \draw (#1) to[out=#3, in=#4] (#2);
}


\begin{scope}
    \coordinate (r1) at (-0.9, 0);          
    \coordinate (i1) at (0.1, 0.3);      

    \draw (r1) -- (i1);

    \node[leaf text] (c1) at ($(r1)+(1.7, -0.6)$) {tiger};
    \draw (r1) -- (c1.west);

    \node[leaf text] (a1) at ($(i1)+(1.3, 0.3)$) {lion};
    \draw (i1) -- (a1.west);
    \node[leaf text] (b1) at ($(i1)+(1.3, -0.3)$) {wolf};
    \draw (i1) -- (b1.west);
    \node[dot] at (r1) {};
    \node[dot] at (i1) {};
\end{scope}

\begin{scope}[yshift=2cm]
    \coordinate (r1) at (-0.9, 0);          
    \coordinate (i1) at (0.1, 0.3);      

    \draw (r1) -- (i1);

    \node[leaf text] (c1) at ($(r1)+(1.7, -0.6)$) {eagle};
    \draw (r1) -- (c1.west);

    \node[leaf text] (a1) at ($(i1)+(1.3, 0.3)$) {tiger};
    \draw (i1) -- (a1.west);
    \node[leaf text] (b1) at ($(i1)+(1.3, -0.3)$) {bear};
    \draw (i1) -- (b1.west);
    \node[dot] at (r1) {};
    \node[dot] at (i1) {};
\end{scope}

\begin{scope}[yshift=4cm]
    \coordinate (r1) at (-0.9, 0);          
    \coordinate (i1) at (0.1, 0.3);      

    \draw (r1) -- (i1);

    \node[leaf text] (c1) at ($(r1)+(1.7, -0.6)$) {dolphin};
    \draw (r1) -- (c1.west);

    \node[leaf text] (a1) at ($(i1)+(1.3, 0.3)$) {shark};
    \draw (i1) -- (a1.west);
    \node[leaf text] (b1) at ($(i1)+(1.3, -0.3)$) {tuna};
    \draw (i1) -- (b1.west);
    \node[dot] at (r1) {};
    \node[dot] at (i1) {};
\end{scope}

\begin{scope}[xshift=6cm, yshift=3.5cm]
    \node[dot] (rootL) at (0,-0.5) {};
    \coordinate (stemL) at (-1.0,-0.5);
    \draw (stemL) -- (rootL);

    \coordinate (nFish_pos) at (1.0, 0.5);
    \node[dot] (nFish) at (nFish_pos) {};
    \draw (rootL) to[out=85, in=180] (nFish);

    \coordinate (nTetra_pos) at (1.0, -1.4);
    \node[dot] (nTetra) at (nTetra_pos) {};
    \draw (rootL) to[out=-85, in=180] (nTetra);

    \node[leaf text] (tunaL) at ($(nFish)+(4.0, -0.3)$) {tuna};
    \connect{nFish}{tunaL.west}{-20}{180}

    \node[leaf text] (sharkL) at ($(nFish)+(4.0, 0.4)$) {shark};
    \connect{nFish}{sharkL.west}{20}{180}

    \coordinate (nMammal_pos) at ($(nTetra)+(1.4, -1.0)$);
    \node[dot] (nMammal) at (nMammal_pos) {};
    \draw (nTetra) to[out=-85, in=180] (nMammal);

    \node[leaf text] (eagleL) at ($(nTetra)+(4.0, 0.8)$) {eagle};
    \draw (nTetra) to[out=45,in=180] (eagleL.west);

    \coordinate (nCarn_pos) at ($(nMammal)+(1.2, -0.8)$);
    \node[dot] (nCarn) at (nCarn_pos) {};
    \draw (nMammal) to[out=-85, in=180] (nCarn);

    \node[leaf text] (dolphL) at ($(nMammal)+(3, 0.8)$) {dolphin};
    \connect{nMammal}{dolphL.west}{45}{180}

    \coordinate (nCani_pos) at ($(nCarn)+(1.2, 0.5)$);
    \node[dot] (nCani) at (nCani_pos) {};
    \draw (nCarn) to[out=85, in=180] (nCani);

    \coordinate (nFeli_pos) at ($(nCarn)+(1.2, -0.5)$);
    \node[dot] (nFeli) at (nFeli_pos) {};
    \draw (nCarn) to[out=-85, in=180] (nFeli);

    \node[leaf text] (wolfL) at ($(nCani)+(1.2, 0.3)$) {wolf};
    \connect{nCani}{wolfL.west}{30}{180}
    \node[leaf text] (bearL) at ($(nCani)+(1.2, -0.3)$) {bear};
    \connect{nCani}{bearL.west}{-30}{180}

    \node[leaf text] (lionL) at ($(nFeli)+(1.2, 0.3)$) {lion};
    \connect{nFeli}{lionL.west}{30}{180}
    \node[leaf text] (tigerL) at ($(nFeli)+(1.2, -0.3)$) {tiger};
    \connect{nFeli}{tigerL.west}{-30}{180}
\end{scope}

\end{tikzpicture}

    }
    \caption{\textsc{Triplet Reconstruction}. \textbf{Left:} three triplet constraints of the form $ab \mid c$. A triplet is satisfied if $\mathrm{LCA}(a,b)$ is a descendant of $\mathrm{LCA}(a,c)$ (equivalently, $\mathrm{LCA}(b,c)$). \textbf{Right:} A rooted binary tree that satisfies the first and the second triplet, but does not satisfy the third ($\textsf{lion wolf} \mid \textsf{tiger}$).}
    \label{fig:phylogenetic_tree_triplet}
\end{figure}
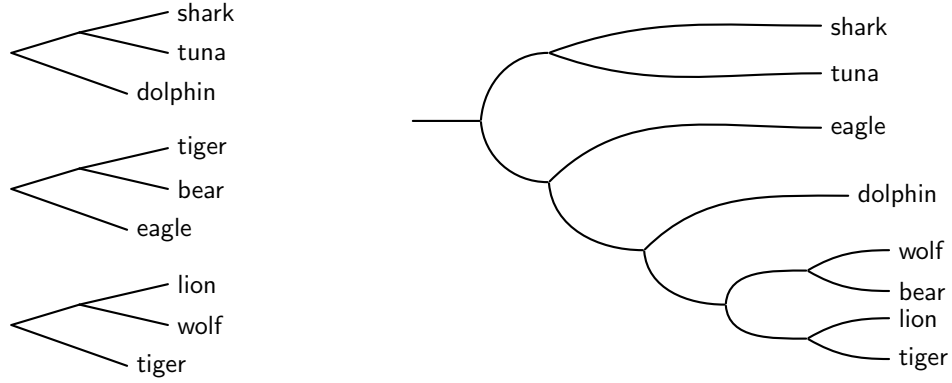

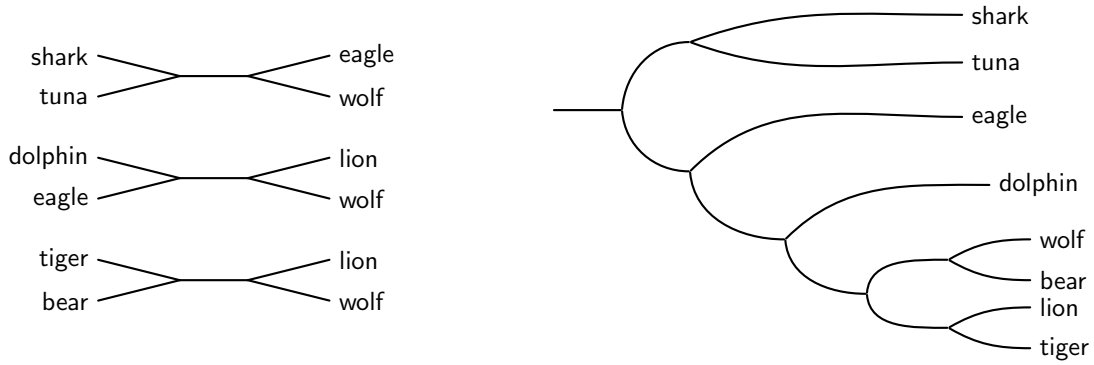
\begin{figure}[h]
    \centering
    {

\begin{tikzpicture}[
    dot/.style={draw, thick, circle, fill=black, inner sep=0pt, outer sep=0pt, opacity=0},
    leaf text/.style={font=\sffamily, anchor=west, align=left},
    clade text/.style={font=\ttfamily, anchor=bottom, midway, yshift=2pt},
    title text/.style={font=\ttfamily\bfseries, align=center, anchor=north},
    thick,
    line cap=round,
    line join=round,
    scale=0.9, every node/.style={scale=0.9} 
]

\newcommand{\connect}[4]{
    \draw (#1) to[out=#3, in=#4] (#2);
}

\begin{scope}
    \coordinate (n1a) at (-0.5, 0);
    \coordinate (n1b) at (0.5, 0);

    \node[leaf text, anchor=east] (lion) at ($(n1a)+(-1.2, 0.3)$) {tiger};
    \draw (n1a) -- (lion.east);
    \node[leaf text, anchor=east] (leopard) at ($(n1a)+(-1.2, -0.3)$) {bear};
    \draw (n1a) -- (leopard.east);

    \node[leaf text] (jaguar) at ($(n1b)+(1.2, 0.3)$) {lion};
    \draw (n1b) -- (jaguar.west);
    \node[leaf text] (tiger) at ($(n1b)+(1.2, -0.3)$) {wolf};
    \draw (n1b) -- (tiger.west);

    \draw (n1a) -- (n1b);
    \node[dot] at (n1a) {};
    \node[dot] at (n1b) {};
\end{scope}

\begin{scope}[yshift=1.5cm]
    \coordinate (n1a) at (-0.5, 0);
    \coordinate (n1b) at (0.5, 0);

    \node[leaf text, anchor=east] (lion) at ($(n1a)+(-1.2, 0.3)$) {dolphin};
    \draw (n1a) -- (lion.east);
    \node[leaf text, anchor=east] (leopard) at ($(n1a)+(-1.2, -0.3)$) {eagle};
    \draw (n1a) -- (leopard.east);

    \node[leaf text] (jaguar) at ($(n1b)+(1.2, 0.3)$) {lion};
    \draw (n1b) -- (jaguar.west);
    \node[leaf text] (tiger) at ($(n1b)+(1.2, -0.3)$) {wolf};
    \draw (n1b) -- (tiger.west);

    \draw (n1a) -- (n1b);
    \node[dot] at (n1a) {};
    \node[dot] at (n1b) {};
\end{scope}

\begin{scope}[yshift=3cm]
    \coordinate (n1a) at (-0.5, 0);
    \coordinate (n1b) at (0.5, 0);

    \node[leaf text, anchor=east] (lion) at ($(n1a)+(-1.2, 0.3)$) {shark};
    \draw (n1a) -- (lion.east);
    \node[leaf text, anchor=east] (leopard) at ($(n1a)+(-1.2, -0.3)$) {tuna};
    \draw (n1a) -- (leopard.east);

    \node[leaf text] (jaguar) at ($(n1b)+(1.2, 0.3)$) {eagle};
    \draw (n1b) -- (jaguar.west);
    \node[leaf text] (tiger) at ($(n1b)+(1.2, -0.3)$) {wolf};
    \draw (n1b) -- (tiger.west);

    \draw (n1a) -- (n1b);
    \node[dot] at (n1a) {};
    \node[dot] at (n1b) {};
\end{scope}

\begin{scope}[xshift=6cm, yshift=3cm]
    \node[dot] (rootL) at (0,-0.5) {};
    \coordinate (stemL) at (-1.0,-0.5);
    \draw (stemL) -- (rootL);

    \coordinate (nFish_pos) at (1.0, 0.5);
    \node[dot] (nFish) at (nFish_pos) {};
    \draw (rootL) to[out=85, in=180] (nFish);

    \coordinate (nTetra_pos) at (1.0, -1.4);
    \node[dot] (nTetra) at (nTetra_pos) {};
    \draw (rootL) to[out=-85, in=180] (nTetra);

    \node[leaf text] (tunaL) at ($(nFish)+(4.0, -0.3)$) {tuna};
    \connect{nFish}{tunaL.west}{-20}{180}

    \node[leaf text] (sharkL) at ($(nFish)+(4.0, 0.4)$) {shark};
    \connect{nFish}{sharkL.west}{20}{180}

    \coordinate (nMammal_pos) at ($(nTetra)+(1.4, -1.0)$);
    \node[dot] (nMammal) at (nMammal_pos) {};
    \draw (nTetra) to[out=-85, in=180] (nMammal);

    \node[leaf text] (eagleL) at ($(nTetra)+(4.0, 0.8)$) {eagle};
    \draw (nTetra) to[out=45,in=180] (eagleL.west);


    \coordinate (nCarn_pos) at ($(nMammal)+(1.2, -0.8)$);
    \node[dot] (nCarn) at (nCarn_pos) {};
    \draw (nMammal) to[out=-85, in=180] (nCarn);

    \node[leaf text] (dolphL) at ($(nMammal)+(3, 0.8)$) {dolphin};
    \connect{nMammal}{dolphL.west}{45}{180}

    \coordinate (nCani_pos) at ($(nCarn)+(1.2, 0.5)$);
    \node[dot] (nCani) at (nCani_pos) {};
    \draw (nCarn) to[out=85, in=180] (nCani);

    \coordinate (nFeli_pos) at ($(nCarn)+(1.2, -0.5)$);
    \node[dot] (nFeli) at (nFeli_pos) {};
    \draw (nCarn) to[out=-85, in=180] (nFeli);

    \node[leaf text] (wolfL) at ($(nCani)+(1.2, 0.3)$) {wolf};
    \connect{nCani}{wolfL.west}{30}{180}
    \node[leaf text] (bearL) at ($(nCani)+(1.2, -0.3)$) {bear};
    \connect{nCani}{bearL.west}{-30}{180}

    \node[leaf text] (lionL) at ($(nFeli)+(1.2, 0.3)$) {lion};
    \connect{nFeli}{lionL.west}{30}{180}
    \node[leaf text] (tigerL) at ($(nFeli)+(1.2, -0.3)$) {tiger};
    \connect{nFeli}{tigerL.west}{-30}{180}
\end{scope}
\end{tikzpicture}

    }
    \caption{\textsc{Quartet Reconstruction}. \textbf{Left:} three quartet constraints of the form $ab \mid cd$. A quartet is satisfied if, in the tree, the unique path between $a$ and $b$ is vertex-disjoint from the unique path between $c$ and $d$. \textbf{Right:} A rooted binary tree (the same as above) which satisfies the first and the second quartet, but does not satisfy the third ($\textsf{tiger bear} \mid \textsf{lion wolf}$).}
    \label{fig:phylogenetic_tree_quartet}
\end{figure}

\end{appendices}

\end{document}